\documentclass[11pt]{article}

\usepackage[margin=1in]{geometry}

\usepackage{hyperref}
\usepackage{graphicx}
\usepackage{booktabs}
\usepackage{amsthm}
\usepackage{lstbayes}
\usepackage{multirow}

\usepackage{float,placeins}
\usepackage{amssymb}
\usepackage{amsmath,amsfonts}
 \usepackage{tablefootnote}
\usepackage{amsthm}
\usepackage{thm-restate}
\usepackage{thmtools}
\usepackage{cleveref}
\usepackage{subcaption}
\usepackage{caption}
\newcommand{\parbold}[1]{\vspace{.25em}\noindent\textbf{#1}}

\renewcommand\cite[1]{\citep{#1}}

\newcommand{\Comments}{1}
\newcommand{\mynote}[3]{\ifnum\Comments=1\textcolor{#1}{#2: #3}\fi}

\usepackage{color}

\newcommand\obs{\ensuremath{\text{observed}}}

\definecolor{DarkBlue}{rgb}{0.1,0.1,0.5}
\hypersetup{
	colorlinks=true,       % false: boxed links; true: colored links
	linkcolor=DarkBlue,          % color of internal links
	citecolor=DarkBlue,        % color of links to bibliography
	filecolor=DarkBlue,      % color of file links
	urlcolor=DarkBlue,          % color of external links
	pdftitle={},
	pdfauthor={},
}

\usepackage[group-separator={,}]{siunitx}
\newcommand\blfootnote[1]{%
	\begingroup
	\renewcommand\thefootnote{}\footnote{#1}%
	\addtocounter{footnote}{-1}%
	\endgroup
}

\newcommand\bbR{\ensuremath{\mathbb{R}}}
 \newcommand\bbP{\ensuremath{\mathbb{P}}}
\newtheorem{lemma}{Lemma}
\newtheorem{theorem}{Theorem}

\usepackage[sort&compress,numbers]{natbib}
\usepackage{amssymb}

\renewcommand{\cite}[1]{\citep{#1}}

\AtBeginDocument{%
	\providecommand\BibTeX{{%
			\normalfont B\kern-0.5em{\scshape i\kern-0.25em b}\kern-0.8em\TeX}}}

\begin{document}

% \title{Spatial Reporting Disparities in Resident Crowdsourcing}

%     \author{
%     Zhi Liu\\
%     Cornell Tech\\
%     \texttt{zl724@cornell.edu} \\
%     \and
%     Nikhil Garg\\
%     Cornell Tech\\
%     \texttt{ngarg@cornell.edu}
%     }

    % \maketitle

% \begin{titlepage}
    
        % \title{Equity in Resident Crowdsourcing: Measuring Under-reporting without Ground Truth Data}

        \title{Quantifying Spatial Under-reporting Disparities\\in Resident Crowdsourcing}

    \author{
    Zhi Liu\\
    Cornell Tech\\
    \texttt{zl724@cornell.edu} \\
    \and
    Uma Bhandaram\\
    NYC Dept of Parks \& Recreation\\
    \texttt{Uma.Bhandaram@parks.nyc.gov}
    \and
    Nikhil Garg\\
    Cornell Tech\\
    \texttt{ngarg@cornell.edu}
    }
    \date{}
    \maketitle
    
    % \todo{title? counting duplicate reports (?) uncovers disparate under-reporting in resident crowdsourcing}

    % \parbold{Classification:} Computer and Information Sciences, Economic Sciences
\begin{center}
        \parbold{Keywords:} Crowdsourcing, benchmark problem, disparities, Bayesian statistics 
\end{center}

\begin{abstract}
Modern city governance relies heavily on crowdsourcing to identify problems such as downed trees and power lines. A major concern is that residents do not report problems at the same rates, with heterogeneous reporting delays directly translating to downstream disparities in how quickly incidents can be addressed. Measuring such under-reporting is a difficult statistical task, as, by definition, we do not observe incidents that are not reported or when reported incidents first occurred. Thus, low reporting rates and low ground-truth incident rates cannot be naively distinguished, and reporting \textit{delays} are unobserved. We develop a method to identify (heterogeneous) reporting delays, without using external ground-truth data. Our insight is that rates on \textit{duplicate} reports about the same incident can be leveraged to disambiguate whether an incident has occurred with its reporting rate once it has occurred. Using this idea, we theoretically reduce the question to a standard Poisson rate estimation task---even though the full incident reporting interval is also unobserved.

We apply our method to over \num{100000} resident reports made to the New York City Department of Parks and Recreation and to over \num{900000} reports made to the Chicago Department of Transportation and Department of Water Management, finding that there are substantial spatial disparities in how quickly incidents are reported, even after controlling for incident characteristics -- some neighborhoods report three times as quickly as do others. These spatial disparities correspond to socioeconomic characteristics: in NYC, higher population density, the fraction of people with college degrees, income, and the fraction of the population that is white all positively correlate with reporting rates. We further validate our methods using external data for which ``ground truth'' incident times are known.

Finally, leveraging a collaboration with the NYC Department of Parks and Recreation, we demonstrate how estimating reporting delays leads to \textit{practical} insights and interventions for a more equitable, efficient government service. 

\blfootnote{
			We benefited from discussions with Aaron Schein, Allison Koenecke, Alex Kobald, Ben Laufer, Emma Pierson, Gabriel Agostini,  Kiran Shiragur, Tao Jiang, and Qian Xie. We also thank the New York City Department of Parks and Recreation for their valuable work, inside knowledge, and data, and we especially thank Fiona Watt.	This work was funded in part by the Urban Tech Hub at Cornell Tech, and we especially thank Anthony Townsend and Nneka Sobers. 
   % \todo{Add others in lab who gave feedback}
	}

\end{abstract}

\section{Introduction}

\label{sec:intro}

% \begin{quote}
% 	If a tree falls in a forest, and no one reports it... does the city know about it?
% \end{quote}

Crowdsourcing systems are an essential component of how modern cities are managed \cite{yuan_co-production_2019,brabham_crowdsourcing_2015}. People report problems that they encounter (via phone call, text, app, social media), which are then logged and responded to by city agencies. Such reporting systems are widely used in North American cities, including the four most populous cities in the U.S.; for example, New York City's 311 system receives over 3 million reports a year. Reports provide real-time updates on on-the-ground conditions \citep{schwester2009examination, hacker_spatiotemporal_2020,minkoff_nyc_2016, lee_crowdsourcing_2021-1}; these updates are used to make immediate decisions---such as which potholes to inspect and fix---and longer-term planning decisions, such as which streets to resurface. 

Although reporting systems have gained popularity in recent decades, there are many ongoing concerns regarding their function. One long-standing concern is whether---conditional on experiencing a similar issue---some residents\footnote{We use the term ``residents'' to denote anyone who might report an issue to the 311 system, including visitors.} are more likely to log a report and request government services \citep{thijssen2016you,clark_coproduction_2013,cavallo_digital_2014,minkoff_nyc_2016,kontokosta_bias_2021, hacker_spatiotemporal_2020, pak_fixmystreet_2017, kontokosta_equity_2017}, and even those who do submit reports do so with different delays\citep{obrien_ecometrics_2015,obrien_uncharted_2017,obrien_urban_2018, klemmer2021understanding}. If reporting behavior is heterogeneous (varies by neighborhood), then government services that rely on such reporting will also be inefficient and inequitable; for example, reporting delays may correspond to delays in incidents being addressed, if the agency responds to incidents in the order that they are reported. 

The central methodological challenge in measuring heterogeneous behavior is the same reason cities rely on crowdsourcing: ground truth conditions are unknown. This induces a missing data challenge: researchers do not directly observe what incidents occurred; even when an incident is reported, researchers and the city do not directly observe when they occurred. Thus, it is difficult to determine whether one neighborhood is \textit{under-reporting} issues compared to other neighborhoods, or whether they truly experience fewer issues (due to, e.g., differences in infrastructure). More generally, this challenge regarding identifiability between reporting and ground truth rates is referred to as the \textit{benchmark} problem, and similarly appears in criminal justice, healthcare, and other domains; arrest data for predictive policing systems, for example, may conflate historical policing disparities versus true crime rates \cite{lum2016predict}; \citet{akpinar2021effect} study the downstream consequences of mistakenly using observed incident count as a proxy for incident rate.

To this long-standing literature, we contribute a method with two distinct advantages: (1) whereas previous literature primarily considers the estimation of {whether an incident is reported at all}, we provide estimates for reporting \textit{delays} -- for an incident of a given type and neighborhood, how long after incident occurrence do we expect the incident to be first reported. The accuracy of these estimates are further validated in \Cref{sec:validation_hurricanes}. (2) The method does not require external, ground truth data to estimate incident occurrence rates as distinct from reporting rates. As a result, we can estimate heterogeneous reporting rates as they differ across high-dimensional incident characteristics (incident category, type, risk level, exact location, etc.). Previous literature either (a) assume uniform incident occurrence and so assume that heterogeneity in the observed reports is a result of differential reporting rate \citep{clark_coproduction_2013,cavallo_digital_2014,minkoff_nyc_2016}; or (b) use external data to construct custom, domain-specific proxies for differences in incident occurrence, and then compare the estimates with observed rates \citep{kontokosta_bias_2021, hacker_spatiotemporal_2020, pak_fixmystreet_2017, kontokosta_equity_2017,obrien_ecometrics_2015}. Such an approach requires high-quality data on incident occurrence rates, which may be expensive to obtain for each incident characteristic.\footnote{See \Cref{app:related} for an extended discussion of this literature.}

As we show through a collaboration with the New York City Department of Parks and Recreation (NYC DPR), these advantages directly translate to \textit{operational relevance} in practice: agencies can use our method's estimates to change incident response policies, for more equitable and efficient government service allocation.  

\parbold{A method to identify reporting rates without ground truth data.} We develop a method to identify (heterogeneous) reporting rates, using \textit{just} the reporting data itself and without requiring proxies for ground truth incident rates---and the method can thus be used across crowdsourcing settings without extensive domain-specific model development or data. Our insight is that we can use information on {duplicate} reports about the same incident (which cities commonly log for operational reasons) to {disentangle} the reporting process from the incident occurrence process, thus identifying reporting rates \textit{conditional on an incident having occurred}. This insight originates from the \textit{missing species} estimation literature, where the number of {duplicate} observations provides evidence regarding the existence of unseen species \citep{orlitsky_optimal_2016, wu_chebyshev_2016, han2021competitive, charikar2019efficient}. However, our setting differs from that literature as incidents, unlike species, occur and are resolved over short time scales; and, even for observed incidents, we do not know \textit{when} they occurred, or \textit{how long} the incident existed. 

We show how to theoretically convert the task to a standard \textit{Poisson rate estimation} challenge, for which a large empirical and theoretical literature exists. In particular, one can use the number of reports between the time of the first report (but not including it) and an estimated incident resolution time, to recover a {Poisson} likelihood with a corresponding rate function. %We show that the method is valid if the estimated death time is an \textit{under-}estimate of the resolution time and is a \textit{stopping time} (does not depend on the future).
An advantage of the method is that we can then leverage standard Poisson rate estimation methods to recover the rates that incidents of different types and in various neighborhoods are reported, conditional on the incident having occurred. For example, we empirically fit Bayesian Poisson regression models with spatial smoothing. We further validate our method with external data, including storm events for which incident times are approximately known and voter participation rates.

\parbold{Application of our method to understand and address disparities.} We apply our method to resident reports made in New York City and the City of Chicago. Using primarily public data,\footnote{NYC DPR provided us with anonymized reporter details so that we could filter out duplicate reports made by the same caller. The public data includes each report, its timestamp, report details, and agency response details.} we analyze reports made to NYC DPR and the City of Chicago Departments of Transportation and Water Management. 
%Our data is primarily publicly accessible, with the addition of anonymized reporter details and internal logs on incident inspection outcomes and work orders provided to us by the New York City Department of Parks and Recreation (NYC DPR).
%
We estimate how the reporting rate varies as a function of incident characteristics (e.g., the request category), spatial differences, and socioeconomic characteristics of the neighborhood in which the incident occurred. We illustrate that our method is precise enough to find differences in reporting rates; for example, we find that reports for more hazardous incidents are submitted at substantially higher rates, conditional on the incident having occurred.

We then study spatial and socioeconomic heterogeneity in New York City and Chicago. We find that there are substantial inter-location reporting differences; for example, comparing census tract-level fixed effects, we find that incidents are reported at a three times higher rate in downtown Manhattan (denser, higher income Borough) than in Queens, even after controlling for the incident-level characteristics. These differences are correlated with socioeconomic and demographic factors: for example, in NYC, population density, fraction of college graduates, median income, and fraction of the population that is white are all positively associated with reporting rates. These findings suggest that the allocation of government services, when reliant only on incident reporting, is inefficient and inequitable. %Reporting rates further vary by season. 

As we discuss, in partnership with NYC DPR, these estimates are being used to inform operational decisions, to reduce the effects of resulting inefficiencies and disparities. 
% \todo{update this summary with updated application}

\smallskip
% We believe that our method is a general, powerful approach to understanding reporting behavior in resident crowdsourcing systems, delivering precise reporting estimates without requiring ground truth data---only duplicate report data, that are already collected by city agencies. 

\textbf{Code availability statement.} Our reproduction code is available at \url{https://github.com/nikhgarg/spatial_underreporting_crowdsourcing}. Code to apply our method to other data sets is available at \url{https://github.com/ZhiLiu724/reporting_rate_estimation}.

% \nikhil{also add link to easier package when you have it}

\textbf{Data availability statement.} Public 311 data for both NYC and Chicago is available via their Open Data portals. In the reproduction GitHub repository above, we have included official links to the portals and further host the exact public data used in this work. The primary NYC results in \Cref{sec:results} are run on private data confidentially provided to us by NYC DPR that cannot be shared -- that data additionally contains columns for anonymized reporter identification, to allow filtering duplicate reports by the same user. All results replicate with the available public data, as reported in the Appendix. The Chicago analyses are fully based on publicly available data. 

% Our full code and public data reproducibility information is available at \url{https://github.com/nikhgarg/spatial_underreporting_crowdsourcing}. %The rest of this paper is organized as follows. \todo{organize}

\section{Theoretical results: Identification of reporting rates}

\subsection{Model}
\label{sec:model}

 \begin{figure}[tbh]
	\centering
	\includegraphics[width=.8\textwidth]{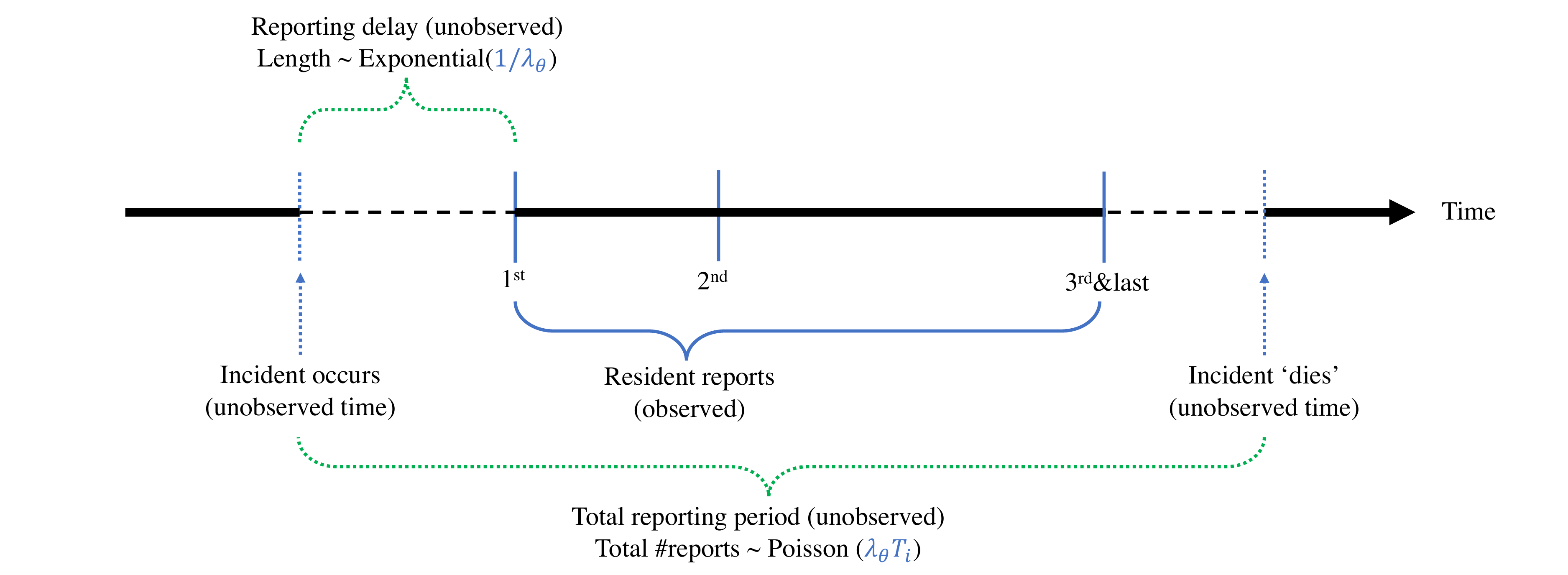}
	\caption{Model of resident reports of an incident $i$ of type $\theta$. Observed and unobserved events are marked with solid and dotted lines, respectively. In this example, there are a total of 3 reports made about this incident. When the reports are generated according to a homogeneous Poisson process with rate $\lambda_\theta$, the average reporting delay is $\frac{1}{\lambda_\theta}$. Incident `death' refers to the incident being resolved or otherwise marked such that no future reports are submitted and logged. The goal is to estimate the reporting rate $\lambda_\theta$ (and consequently, the unobserved reporting delay) for each incident of type $\theta$.}
	\label{fig:reports_gen}
\end{figure}

We directly model and distinguish the two underlying processes that lead to a reported incident: incident occurrence, and the public's reporting behavior conditional on an incident having occurred.  \Cref{fig:reports_gen} illustrates the model for a given incident.

\parbold{Model part 1: incident birth and death.} Incidents of type $\theta$ occur (``are born'') according to a random process parameterized by $\Lambda_\theta$. (For example, a homogeneous Poisson process with rate $\Lambda_\theta$). Type $\theta$ encodes geographic and incident-level (e.g., whether it is a hazard) characteristics. For each incident $i$, let $t_i$ denote its birth time. Furthermore, the incident ``dies'' at some time $t_i + T_i$, and no more reports are logged. The death time reflects, for example, the incident being resolved by the agency or by an outside group, or reports no longer being logged because the issue has been inspected. Crucially, birth times are unobserved in the model; death times may be unobserved or partially observed.\footnote{In practice, death times are partially observed: agency actions such as inspections and work orders are observed by the agency. However, (1) such data may not be available to external researchers; (2) many incidents are resolved by outside community groups, and these times are unobserved even by the agency. Fully observed death times improve our method's accuracy but do not change the problem difficulty: crucially, birth times are always unobserved by both the agency and researchers, by definition.} 

\parbold{Model part 2: reports.} An incident $i$ of type $\theta$ is active in time window $[t_i, t_i + T_i]$. During this time, reports are submitted according to a homogeneous Poisson process, with rate $\lambda_\theta > 0$.
%(\tau)$ at incident age $\tau \in [0, T_i]$. 
Thus, the number of reports $M_i$ for incident $i$ follows a Poisson distribution, $M_i \sim \text{Poisson}\left(\lambda_\theta T_i\right)$.
% \todo{add comment about non-homogeneous somewhere, maybe in the model discussion bit?}

\parbold{Observed data.} From the viewpoint of the city and researchers, we observe reports and city actions. The observed data for a time interval $(0, T)$ consists of the following, for each type $\theta$. We observe the subset of incidents for which there was at least one report. For each such incident, we observe $M_i$ and the time of each report. Finally, we observe how the agency responded to the reported incidents, such as when an inspection was completed, and the details of any inspection and completed work order. 
% \nikhil{removed the N observed notation since not used in main text. Move to appendix if used there }
% Suppose there are $N_\theta$ incidents of the type. $N^{\obs}_\theta \leq N_\theta$ of these incidents, corresponding to the incidents $i$ with at least one report, $M_i \geq 1$

\parbold{Research question.} We are interested in auditing the efficiency and equity of the reporting process: how quickly are incidents reported, and how does that reporting rate $\lambda_\theta$ vary with the incident type $\theta$? Estimating this value directly reveals heterogeneous reporting \textit{delays}, as the expected time between incident occurrence and first report in the model is $\frac{1}{\lambda_\theta}$. The primary methodological task is overcoming the missing data challenge: we do not observe incidents that are never reported, nor birth times $t_i$ of the incidents that are reported. This makes it difficult to distinguish between incident rate $\Lambda_\theta$ and reporting rate conditioned on an incident, $\lambda_\theta$. As we prove in Appendix \Cref{prop:identi}, it is impossible to do so with just the above data, if one throws away the information on duplicates. The idea -- confounding incident occurrence and reporting -- is an example of the canonical \textit{benchmark problem}. Nevertheless, many works in the resident crowdsourcing literature conduct such analyses.

% , if we rely solely on \textit{either} the \textit{observed} number of incidents \textit{or} number of reports, it is impossible to distinguish heterogeneous reporting behavior under the influence of varying ground truth conditions, even in the limit of dataset size. 

\subsection{Theoretical identification result and empirical method summary}
\label{sec:identification}

% The central methodological challenge in identifying and estimating the reporting rate $\lambda_\theta$ is the lack of reliable ground truth data. 

% our main insight is that by using this duplicate data, in contrast, one can identify $\lambda_\theta$ 

% In particular, the observed data can result from either incident rate differences or reporting disparities.  

We show that $\lambda_\theta$ is identifiable in practically reasonable settings, using \textit{duplicate} reports about the same incidents: the first incident report informs us that this incident exists, and the remaining reports identify the reporting rate. In this section, we assume a \textit{fixed} incident type $\theta$, and so omit it from the notation.

The following theorem reduces the research question to a standard Poisson rate estimation task. The chief challenge is that we do not observe incident birth and death times, and so we must construct an observation interval $(s_i, e_i]$ such that the number of reports within the interval follows a known distribution as a function of $\lambda$. The challenge is that this interval \textit{itself must depend on the reporting rate $\lambda$} -- for example, we can only start the counting interval at the time of the first report, and the agency responds to the incident partially as a function of the number of reports. As a result, the distribution of the entire data, not just the number of reports within the interval (i.e., including interval start and end times) could be a complex, unknown function of $\lambda$. Thus, some choices of the interval are not valid; for example, if $s_i$ is the time of the first report and $e_i$ is the time immediately before the second report, then we would always observe zero reports in the interval, independent of $\lambda$. The theorem formalizes that, as long as the start and end times depend only on the rate $\lambda$ through the number of reports up to those times, we can nevertheless decompose the likelihood and reduce the task to Poisson rate estimation.

% using the data of reports that occur in some interval  for each incident $i$

\begin{restatable}{theorem}{thmst}\label{thm:stoppingtimes}
Suppose reporting follows a homogeneous Poisson process with reporting rate $\lambda$. For each incident, construct (a) interval start time random variable $S_i \geq t_i$ that is independent of the reporting rate $\lambda$, conditional on the first report; (b) end time random variable $E_i \leq t_i + T_i$ that is independent of $\lambda$ and $S_i$, conditional on the reports observed up to that time. Let data $\mathcal{D}_i$ include $S_i = s_i$, $E_i = e_i$, the number of reports $\tilde{M}_i$ within $(s_i, e_i]$ and their associated report times.

Then, the likelihood of the data $\mathcal{D}_i$ as a function of the reporting rate $\lambda$, conditional on all reports before $s_i$, can be decomposed as
\begin{equation}
    \mathcal{L}(\lambda|\mathcal{D}_i) = p\left(\tilde{M}_i;\lambda(e_i-s_i)\right)f(\mathcal{D}_i), \label{eq:decompose}
\end{equation} 
where
\begin{equation}
    p\left(\tilde{M}_i;\lambda(e_i-s_i)\right) = \frac{\left[\lambda (e_i-s_i)\right]^{\tilde{M}_i}}{\tilde{M}_i!} \exp\left(-\lambda (e_i-s_i)\right) \label{eq:poissonlikelihood}
\end{equation}
is the probability mass function of a Poisson distributed random variable with $\tilde{M}_i$ occurrences and rate parameter $\lambda(e_i-s_i)$, and $f(\mathcal{D}_i)$ is a function that does not depend on $\lambda$.
\end{restatable}

The decomposition in \Cref{eq:decompose} enables standard Poisson rate estimation methods to be applied. For example, %when we observe data associated with multiple reports indexed by $i$, 
%reporting rate is homogeneous ($\lambda(\tau) \triangleq \lambda$), 
the maximum likelihood estimate for $\lambda$ is:
\begin{equation}
\hat{\lambda}^{MLE} = \frac{\sum_{i}\tilde M_i}{\sum_{i}(e_i-s_i)}. \label{eq:MLE}
\end{equation}

% \nikhil{I added the homogeneous thing in the result, check what I say here is correct}

The proof is in \Cref{sec:mainthmproof} and follows from decomposing the likelihood function into full conditional probabilities. The key step is to find the conditions under which one can use data-dependent observation period length $E_i - S_i$, while maintaining a Poisson likelihood for reports within the period.\footnote{Perhaps surprisingly, we could not find rate estimation literature in which $E_i - S_i$ also depends on the sample path. Note that we could leverage classical Poisson process results (see e.g., \cite{resnick1992adventures}) if we had considered a non-data-dependent end time $E_i$, e.g., $E_i - S_i$ is a constant. However, this choice would be data inefficient. Some incidents are inspected and addressed within a day of being first reported, while others may take a year to be inspected. As incident inspection (and death time $t_i + T_i$) may be a function of the number of reports, we would be forced to pick a fixed interval end time of less than a single day, for all incidents.} The result extends to the zero-inflated model we use in our empirical estimation. %\zhi{i am still not entirely comfortable with claims about non-homogenous cases, feel free to change though. -- i believe the zero-inflated model is not what is conventionally considered part of non-homogenous poisson process?}

\subsubsection*{Applying the theory to our setting}
First, we need to choose interval $(S_i, E_i]$ that meets the conditions of the theorem. The best choice for $S_i$ is the (observed) time of the first report, and so we count reports starting \textit{after} the first report---this time is the earliest that we know the incident exists. This choice formalizes the idea that our method uses the first report to establish an incident existing (and a timestamp for it), and then duplicate reports to identify the reporting rate. To set $E_i$, we leverage government \textit{agency response} data detailing the time when each incident was inspected or worked on. We choose the minimum of these times and a fixed period after the first report, $\tilde t_i + \bar T$, where $\bar T$ is a researcher choice for which we conduct robustness checks. 

Second, we must empirically estimate reporting rates given high dimensional incident types $\theta$. We leverage Poisson regression, 
% \textbf{Poisson regression.} We use a Bayesian Poisson regression approach to understand the association between incident characteristics $\theta$ and reporting rate $\lambda_\theta$
i.e., we assume the form
% First, we assume that the reporting rate is affected by the type vector $\theta$ as in a Poisson regression: for each model, we have a \textbf{standard} reporting rate $\lambda_\theta$, defined as:
\begin{align}
\label{eq:poissonregression}
 \lambda_\theta &= \exp\left(\alpha + \beta^T\theta\right).
% \label{eqn:lambdaregressionassumption}
% \intertext{and so, for each incident $i$, the likelihood of its data follows from \Cref{eq:decompose} and satisfies the following:}
% \tilde M_i &\sim \text{Poisson}\left(\tilde T_i\lambda_{\theta(i)} \right).
% \mathcal{L}(\lambda_{\theta_i}|\mathcal{D}_i, \mathcal{H}_i) &\varpropto p\left(\tilde M_i ; \lambda_{\theta_i}(e_i-s_i)\right),
\end{align}
To better match the data distribution, we further fit a \textit{zero-inflation} component (corresponding to the probability that we receive no reports in the given period). To accommodate the high dimensionality of the incident type vector $\theta$, we fit our models in a Bayesian setting and report posterior distributions of regression coefficients $\alpha$ and $\beta$.

% \nikhil{Zhi, move rest of this summary to methods section or appendix} Setting $E_i$ is a design choice. In our empirical application, ; we validate our estimates with multiple choices of $\bar T$ in \Cref{sec:robustnesstables}. Note that even though government agencies may vary their response time based on the number of reports received of one incident, these responses are made without knowledge of the actual reporting rate, and thus satisfy our required condition. The addition of $\bar T$ alleviates the concern that $E_i$ may extend beyond the actual and potentially unknown death time $t_i + T_i$ of an incident, and thus risk underestimating $\lambda$ -- agency responses may be delayed due to various reasons, by which time the incident may have already been resolved by other parties.

% How do we apply the above result to our setting? Given that we do not know birth time $t_i$ or lifetime $T_i$, we require a known interval $(S_i, E_i] \subseteq [t_i, t_i + T_i]$ such that we know the number of reports in that interval, and that $S_i$ and $E_i$ satisfy the desired conditions. Then, various standard Poisson rate estimation methods can be applied. 

% We develop a method that decomposes the likelihood of the data, enabling the use of standard  techniques; then, we leverage this result to develop Bayesian Poisson regression approaches to tackle the high-dimensionality of $\theta$ and estimate $\lambda_\theta$ as a function of it.

%  and before our chosen end time $\tilde t_i + \tilde T_i$

% \section{Empirical method}

% \input{supershort/3_empirical_1_identification}

\section{Empirical Results: Heterogeneous reporting in NYC and Chicago}
\label{sec:results}

Our method captures fine-grained reporting \textit{heterogeneity}, which is important in both (1) auditing the efficacy and equity of the reporting system, and (2) designing better approaches to respond to reports. We focus on understanding this heterogeneity, across \textit{incident characteristics} and \textit{neighborhoods}, especially regarding potential \textit{socioeconomic equity} concerns. We find substantial (and statistically significant) reporting disparities on all three dimensions.

% Here, we report results for the \textbf{zero-inflated Poisson regression}. %, including with ICAR spatial coefficients.

Our data from NYC contains \num{223416} service requests made by the public between 6/30/2017 and 6/30/2020 to the NYC Department of Parks and Recreation(DPR) about street trees. All statements regarding statistical significance in this section refer to whether 95\% Bayesian credible intervals include $0$. Methodological details are in \Cref{sec:empiricsmethod}, with data and preprocessing in \Cref{sec:data}.

\begin{table}[tbh]
	% \caption{(a) Regression coefficients for \textbf{zero-inflated Poisson regression} with our \textbf{Base} covariates: incident-level covariates and Borough fixed effects with a Max Duration of 100 days, and default repeat caller removal. $R\_hat$ is a Bayesian MCMC convergence diagnostic, where a value close to 1 indicates that multiple chains have mixed well. A positive coefficient for a categorical variable indicates that an incident with such a characteristic is likely to be reported more frequently compared to the average. }
  \caption{(a) Regression coefficients for \textbf{zero-inflated Poisson regression} with our \textbf{Base} covariates: incident-level covariates and Borough fixed effects with a Max Duration of 100 days, and default repeat caller removal. More positive coefficients indicate that the covariate is associated with higher reporting rates. (b) Implied reporting delays calculated according to the coefficient estimates in (a) and a homogeneous Poisson process model, for an average-sized tree. These numbers are calculated by estimating the reporting rate using \Cref{eq:poissonregression} and the learned coefficients, and then calculating the mean delay for Poisson processes with the given rate. We showcase one such calculation in \Cref{app:exampledelaycalc}. For Borough and category, we enforce a \textit{zero-sum} constraint on the coefficients and so can present coefficients for each level, without collinearity issues.}
	\label{tab:basicandcontext}
 
	\centering
 \subfloat[][Posterior distributions for regression coefficients.]{\small
 \centering
	\begin{tabular}{lrrrrrr}
		\toprule
		{} &   Mean &  StdDev &   2.5\% &  97.5\% \\
\midrule
Intercept                           & -3.229 &   0.027 & -3.285 & -3.179  \\
Zero inflation fraction                 &  0.661 &   0.003 &  0.655 &  0.668 \\
Inspection condition: Tree is Dead                & -0.338 &   0.035 & -0.411 & -0.278 \\
Inspection condition: Tree is Fair                & -0.168 &   0.027 & -0.222 & -0.114 \\
Inspection condition: Tree is Good/Excellent    & -0.274 &   0.029 & -0.330 & -0.214 \\
Risk assessment score                 &  0.240 &   0.011 &  0.218 &  0.259 \\
Log(Tree Diameter at Breast Height) & -0.035 &   0.009 & -0.054 & -0.019\\
Borough[Bronx]                      & -0.051 &   0.025 & -0.105 & -0.008 \\
Borough[Brooklyn]                   & -0.382 &   0.019 & -0.420 & -0.345 \\
Borough[Manhattan]                  &  0.438 &   0.049 &  0.333 &  0.535 \\
Borough[Queens]                     & -0.249 &   0.019 & -0.291 & -0.215 \\
Borough[Staten Island]              &  0.245 &   0.032 &  0.182 &  0.304 \\
Category[Hazard]                    &  1.418 &   0.017 &  1.380 &  1.450 \\
Category[Illegal Tree Damage]       &  0.224 &   0.034 &  0.155 &  0.285 \\
Category[Prune]                     & -0.087 &   0.028 & -0.148 & -0.033 \\
Category[Remove Tree]               &  0.034 &   0.023 & -0.012 &  0.078 \\
Category[Root/Sewer/Sidewalk]       & -1.589 &   0.035 & -1.664 & -1.525 \\
\bottomrule
\end{tabular}
\label{tab:basicinfrefcoef}
}
\newline
\subfloat[Contextualized average reporting delays for average-sized trees.]{\small
		% \begin{table}[tb]
% 	\caption{Average reporting delay $\frac{1}{\lambda_\theta}$ for incidents of different types $\theta$ as implied by \Cref{tab:basicinfrefcoef}.}
% 	\label{tab:contextualizingrates}
% 		\centering
	\begin{tabular}{lcc}
		\toprule
		{Incident characteristics}              & {Manhattan} & {Queens} \\ \midrule
		Hazard, tree in Poor condition, risk assessment score 12         & 2.2 days           & 4.3 days        \\
		Illegal Tree Damage, tree in Poor condition, risk assessment score 9                          & 15.9 days            & 30.7 days         \\
		Root/Sewer/Sidewalk issue, tree in Fair condition, risk assessment score 5 & 111.3 days           & 221.2 days\\
		\bottomrule
	\end{tabular}
% \end{table}
		\label{tab:ratescontext}
}
\vspace{-.5cm}
\end{table}

\subsection{Heterogeneous reporting in NYC}

\paragraph{Reporting heterogeneity across incident type} We first study how reporting rates vary by incident-level characteristics. \Cref{tab:basicinfrefcoef} shows the coefficients resulting from the zero-inflated regression, using our \textbf{Base} analysis covariates, which include incident-level covariates and \textit{Borough} level fixed effects (full list included in \Cref{tab:basicinfrefcoef}). \Cref{tab:ratescontext} contextualizes the coefficient estimates into average estimated reporting delays, assuming that the reporting rate coefficients correspond to a homogeneous reporting process.

We find that the reporting rate substantially (and statistically significantly) varies by incident characteristics. In particular, we have several columns characterizing the results of an inspection by a forester. The following incident characteristics (as marked by the inspector) are associated with higher resident reporting rates: higher risk ratings, trees in Poor or Critical condition (as opposed to Excellent, Good, Fair, or Dead), and incidents deemed as Hazards (as opposed to e.g., a Root/Sewer/Sidewalk issue) by the first reporter. These differences are practically relevant: the difference in coefficients between the Hazard and Illegal Tree Damage categories is $1.2$, and so Hazard incidents are reported at approximately a  $e^{1.2} \approx 3.3$ times higher rate than Illegal Tree Damage incidents.\footnote{The interpretation is not exact, as we learn a single zero-inflation coefficient. However, coefficients in the non-zero-inflation model are similar.}
 
% \footnote{The coefficient for trees in Poor or Critical condition is not in \Cref{tab:basicinfrefcoef}, rather, since the inspection condition is a categorical covariate, to maintain identifiability, the Poor or Critical level is dropped and its coefficient absorbed into the intercept. So coefficients for all other conditions being negative indicate that incidents in these conditions are all reported at lower rates.} 

 These incident-level associations are as expected -- more serious incidents are reported more quickly, providing evidence that our method is precise enough to capture such heterogeneity. For example, NYC DPR inspectors suggest that Root/Sewer/Sidewalk issues are particularly difficult for most residents to find, as the problem is often underground. This level of precision is not available for other methods to measure under-reporting, as one would need to construct equally precise estimates on the ground truth rate at which incidents of different types occur.

%   \begin{figure}[bt]
%  	\centering
% \subfloat[][Census tract coefficients in New York City]{
%  		\includegraphics[width=.48\textwidth]{plots/nyc_spatial_map.pdf}
%  		\label{fig:spatialmap}
% }
%  	\hfill
% 	\subfloat[][Census block group coefficients in Chicago]{
%         \vspace{-.3cm}
% 		\includegraphics[width=.43\textwidth]{plots/chicago_spatial_map.pdf}
% 		\label{fig:spatialmapchicago}
% }
%  	\caption{Coefficients on spatial variables in NYC and Chicago. These spatial coefficients are estimated using the ICAR spatial zero-inflated Poisson regression. More positive coefficients indicate higher reporting rates.\todo{move to appendix}}
%  	\label{fig:spatialheterog}
%  \end{figure}

\begin{table}[tb]
	\centering
 \caption{Coefficients on socioeconomic covariates. These coefficients are each estimated \textit{alone} in a regression alongside the incident-level covariates. A positive coefficient means that an increase along this covariate is associated with an increase in reporting rate and vice versa. Coefficients marked with (*) indicate that their 95\% posterior credible intervals do not contain 0. Full table presented in Appendix \Cref{tab:coeffull}.}
  \label{tab:censuscoefficients} \small
% %\begin{table}
% %\centering
% %\caption{Census Tract Demographic coefficients}
% \begin{tabular}{lrr}
% \toprule
%                       &   Mean &  StdDev \\
% \midrule
%            Median Age & -0.019 &  0.0085 \\
%     Fraction Hispanic &  0.050 &  0.0083 \\
%        Fraction white &  0.054 &  0.0090 \\
%        Fraction Black & -0.053 &  0.0097 \\
%     Fraction noHSGrad & -0.027 &  0.0085 \\
% Fraction college grad &  0.030 &  0.0088 \\
%      Fraction poverty & -0.021 &  0.0084 \\
%       Fraction renter &  0.029 &  0.0090 \\
%  Fraction single unit & -0.028 &  0.0089 \\
%       Log(Avg income) &  0.025 &  0.0088 \\
%          Log(Density) &  0.058 &  0.0090 \\
% \bottomrule
% \end{tabular}
% %\end{table}

% above is the old results with borough effects

%\begin{table}
%\centering
%\caption{Census Tract Demographic coefficients}
\begin{tabular}{lrlr}
\toprule
%                       &   Mean &  StdDev &  &   Mean &  StdDev \\
% \midrule
%                     Median age & -0.014 &   0.008 & 
%              Fraction Hispanic &  0.059 &   0.009   \\
%                 Fraction white &  0.064 &   0.009 &  
%                 Fraction Black & -0.042 &   0.009  \\
% Fraction no high school degree & -0.035 &   0.009 & 
%        Fraction college degree &  0.029 &   0.009   \\
%               Fraction poverty & -0.023 &   0.009 & 
%                Fraction renter &  0.027 &   0.009  \\
%                Fraction family & -0.084 &   0.009 & 
%        Log(Median house value) &  0.014 &   0.010  \\
%         Log(Income per capita) &  0.054 &   0.009 & 
%                   Log(Density) &  0.042 &   0.010 \\
                      &   Mean &    &   Mean  \\
\midrule
                    Median age & -0.033* &  
             Fraction Hispanic &  0.030*     \\
                Fraction white &  0.057* &   
                Fraction Black & -0.039*    \\
Fraction no high school degree & -0.031* &   
       Fraction college degree &  0.043*      \\
              Fraction poverty & -0.010\phantom{*} &   
               Fraction renter &  0.054*    \\
               Fraction family & -0.081* &   
       Log(Median house value) &  0.065*     \\
        Log(Income per capita) &  0.048* &   
                  Log(Density) &  0.077*   \\
\bottomrule
\end{tabular}
%\end{table}

\end{table}

\paragraph{Socioeconomic and spatial reporting inequities} There are substantial socioeconomic and demographic disparities in reporting rates, beyond those that can be explained via incident-level characteristics. To analyze these disparities, we first fit our model with socioeconomic characteristics from Census data, where each socioeconomic covariate is included \textit{alone} in a regression, alongside the incident-level covariates. Table \ref{tab:censuscoefficients} contains the corresponding coefficients. One of the strongest associations is with population density: the more people per unit area, the more people who might encounter and report an issue. Similarly, a higher fraction of people with college degrees, log income per capita, and a fraction of the population that is white all positively correlate with reporting rates, with the Bayesian Credible Intervals not overlapping with $0$.  For example, one standard deviation increase in the proportion of white people is associated with a $5.8\%$ increase in reporting rate, even conditional on incident-level covariates.\footnote{This is obtained by calculating $e^{0.057}$, according to the Poisson regression model. Since all covariates are standardized, a unit increase in each covariate corresponds to one standard deviation increase in the original data. Causal interpretation of these variables is out of scope, given the co-linearity of such socioeconomic characteristics. }

We next study the relationship between socioeconomic characteristics and reporting rates in each census tract, by fitting our model with the incident-level characteristics and a set of socioeconomic characteristics jointly (log income per capita, fraction of white residents, fraction of renter, median age, and fraction of residents with college degree, but \textit{not }population density); then, we use the socioeconomic coefficients to learn the cumulative association for each census tract.\footnote{We give one example of such calculation in \Cref{app:socioeconomic}.} \Cref{fig:spatialheterogproj} shows the resulting map---there are substantial spatial differences in reporting across census tracts as explained by demographics. For example, downtown and midtown Manhattan (the blue region in the middle) have substantially higher values than Harlem and the Bronx (the red region to the north of the blue region), and the latter are substantially socioeconomically disadvantaged compared to the former. In \Cref{app:voterparticipation}, we further show that these census tract-level values further correlate with voter participation rates in each tract -- i.e., that disparities in 311 system usage further relate to other forms of civic participation and representation. We note that these associations cannot be explained via correlations with population density -- as shown in Appendix \Cref{tab:nycmultidemonodensity,tab:nycmultidemo}, further controlling for population density does not substantially affect the other coefficients. These reporting disparities suggest substantial downstream effects in how quickly incidents are addressed, even if the agency does not prioritize one demographic group over another after receiving reports.

 Finally, we study spatial heterogeneity beyond which may occur due to socioeconomic information. We analyze this heterogeneity by including indicator variables for each of the over \num{2000} census tracts (with spatial smoothing \cite{morris2019bayesian}), alongside incident-level covariates. %using the \textbf{ICAR Spatial zero-inflated Poisson regression}.
 \Cref{fig:spatialmap} shows the coefficient associated with each census tract---there are substantial (and statistically significant) spatial effects. Notably, these coefficients are \textit{on the same order} as the incident-level characteristics, further suggesting substantial inefficiencies and inequities, even beyond those that are associated with socioeconomic characteristics. %Overall,  \todo{writing after decide which maps to include in main text}

 \begin{figure}[bt]
 \centering
        % \subfloat[][Census tract coefficients in New York City]{
 \includegraphics[width=.6\textwidth]{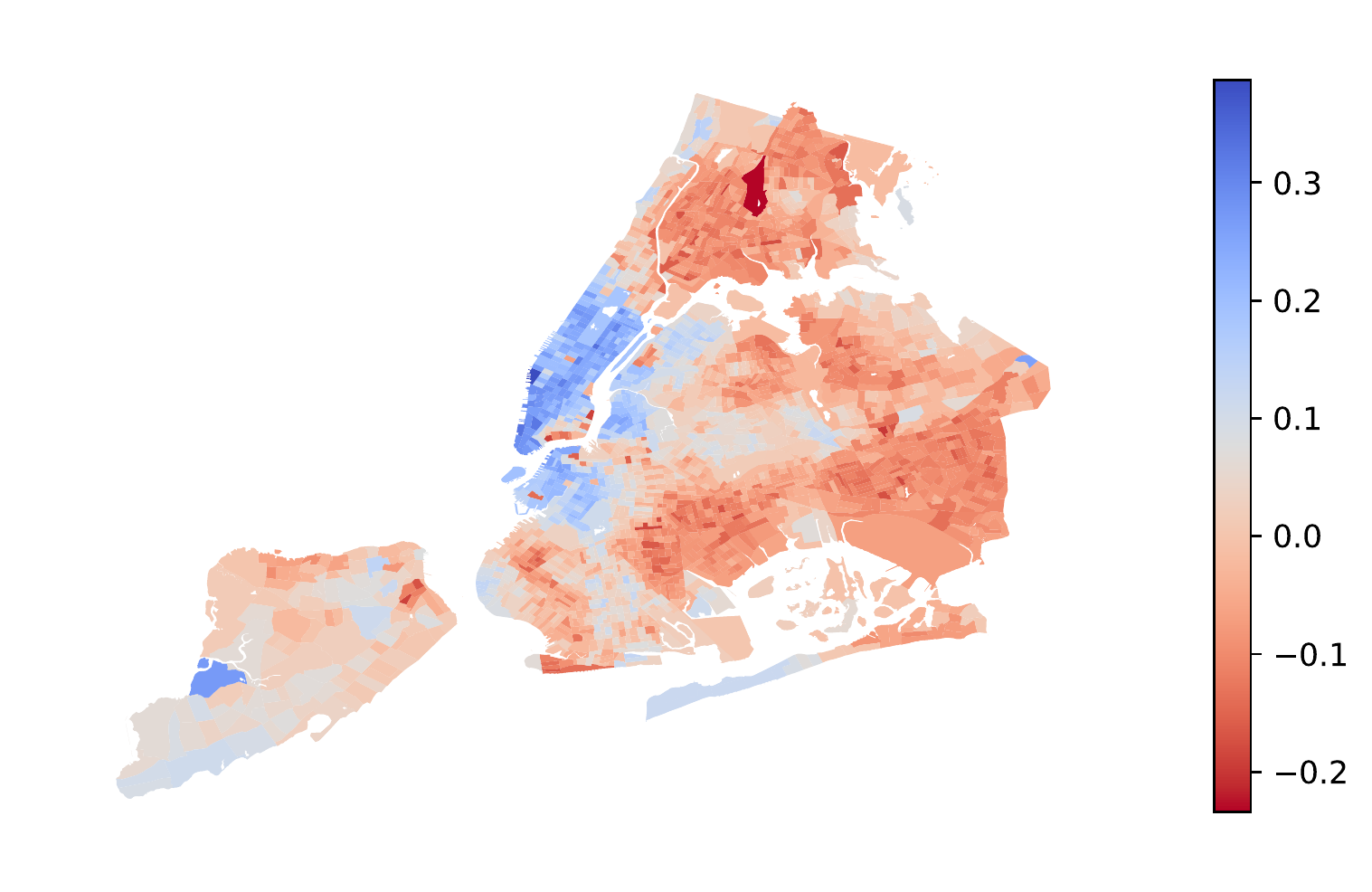}
 		% \label{fig:spatialmapproj}
% }
%  	\hfill
% 	\subfloat[][Census block group coefficients in Chicago]{
%         \vspace{-.3cm}
% 		\includegraphics[width=.43\textwidth]{plots/chicago_spatial_map_projected.pdf}
% 		\label{fig:spatialmapchicagoproj}
% }
 \caption{Coefficients for each census tract, representing the combined association of socioeconomic variables (log income per capita, fraction of white residents, fraction of renter, median age, and fraction of residents with college degrees) on reporting rates. More positive coefficients indicate higher reporting rates. Overall, more socioeconomically privileged neighborhoods have substantially higher reporting rates in NYC, and there is substantial spatial heterogeneity.}
 	\label{fig:spatialheterogproj}
 \end{figure}

\paragraph{Comparison to methods from prior work} The spatial effects further demonstrate that our methods -- correcting for heterogeneous incident type distributions -- give different results than those of prior work. Appendix \Cref{fig:reginccoef} shows the relationship between our measured census tract coefficient and the number of unique incidents per tree observed for that census tract in the period. Perhaps surprisingly, there is a slight negative relationship. These differences may emerge for two reasons: (a) the non-identifiability between the rates incidents occur and are reported (\Cref{prop:identi}), and (b) that tree counts do not control for `legitimate' incident-level characteristics (such as risk), in ways that may correlate with geography. Similar results emerge when restricting the analysis to Hazard-type incidents, and comparing the census tract coefficients to the total number of observed incidents or reports per census tract, without normalizing by number of trees. While it may be possible to construct accurate normalizing estimates for how many incidents of each type (severity, tree characteristics, location) are expected to occur, our approach does not require doing so.

\subsection{Heterogeneous reporting in Chicago}
Our data from Chicago comprise service requests made to the City of Chicago's Department of Transportation and Department of Water Management. Results from Chicago are largely consistent with those from NYC. More urgent incidents get reported more quickly (``Open Fire Hydrant'' incidents have the highest reporting rate), there is substantial spatial heterogeneity, and these spatial differences correspond to socioeconomic disparities. \Cref{fig:spatialmapchicago} illustrates the results of the {Spatial} analysis in Chicago. Full results are in the Appendix; \Cref{tab:censuscoefficientschicago} summarizes the findings of the {socioeconomic} analysis. Population density and fraction of the population that is Black directionally remain the same, while some other demographic factors in Chicago differ in their relationship with reporting rates. %For example, in \Cref{tab:censuscoefficients}, the coefficient for log average income is positive, suggesting that neighborhoods with higher average income tend to report more frequently; whereas in \Cref{tab:censuscoefficientschicago}, this coefficient is negative.
Causal explanations for city- and neighborhood-specific heterogeneity are out of the scope of this paper but may be of particular interest for future work and city governments.

\section{Results: applications of quantifying reporting delays
}
\label{sec:conclusion}

  \begin{figure}[bt]
  % \vspace{-0.2cm}
%  	\centering
% % \subfloat[][Census tract spatial coefficients in New York City]{
%  		\includegraphics[width=.48\textwidth]{plots/impact_analysis/publicdatarisk_risk_notimputed_A_Borough.pdf}
%  		% \label{fig:delays}
\subfloat[][End-to-end median delays in Boroughs.]{
		\includegraphics[width=.45\textwidth]{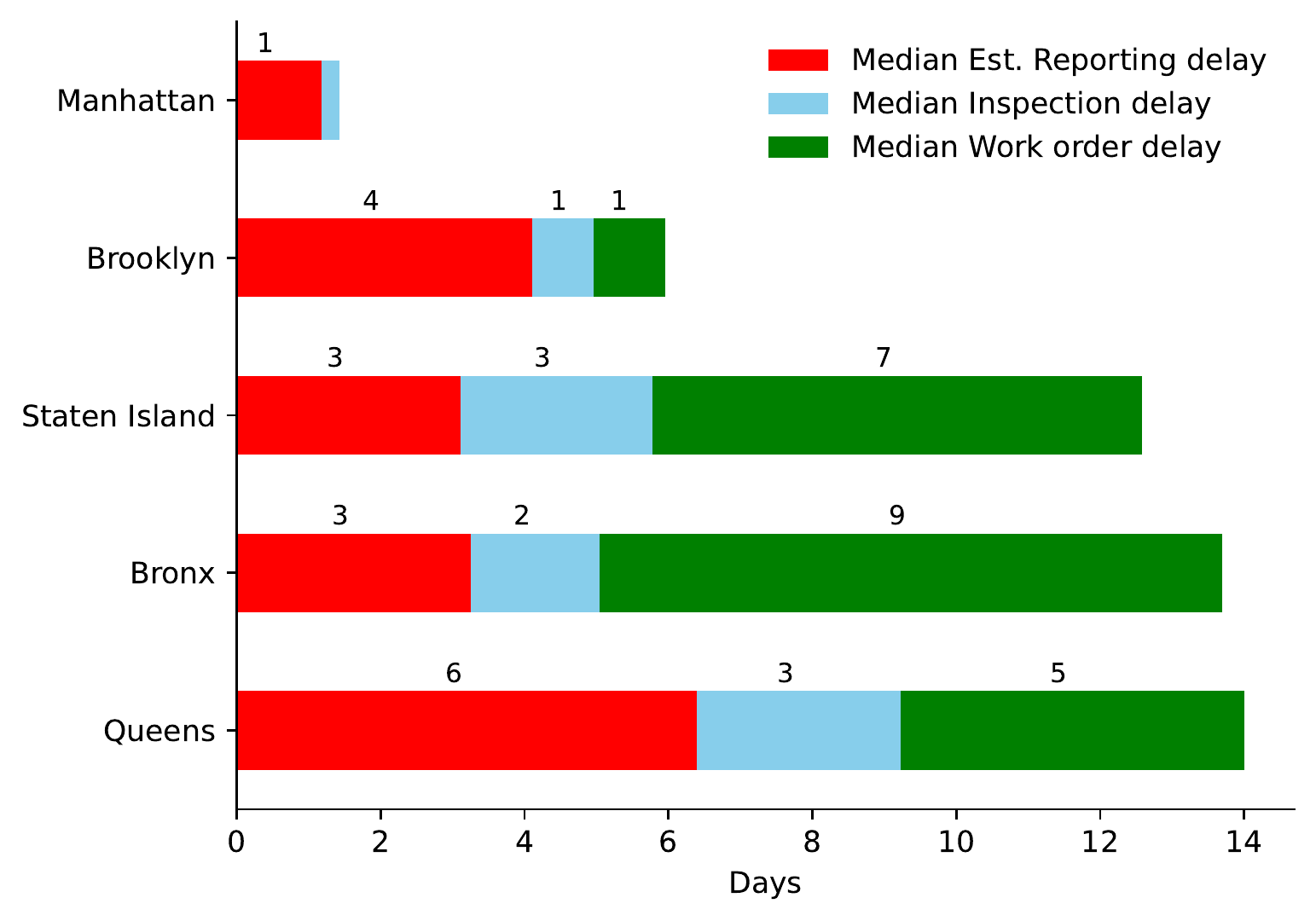}
		% 		\caption{Census tract fixed effects}
		\label{fig:delaysall_raw}
		% 		\end{subfigure}
	}
	\hfill
	\subfloat[][Delays in Boroughs relative to city median.]{
		\includegraphics[width=.45\textwidth]{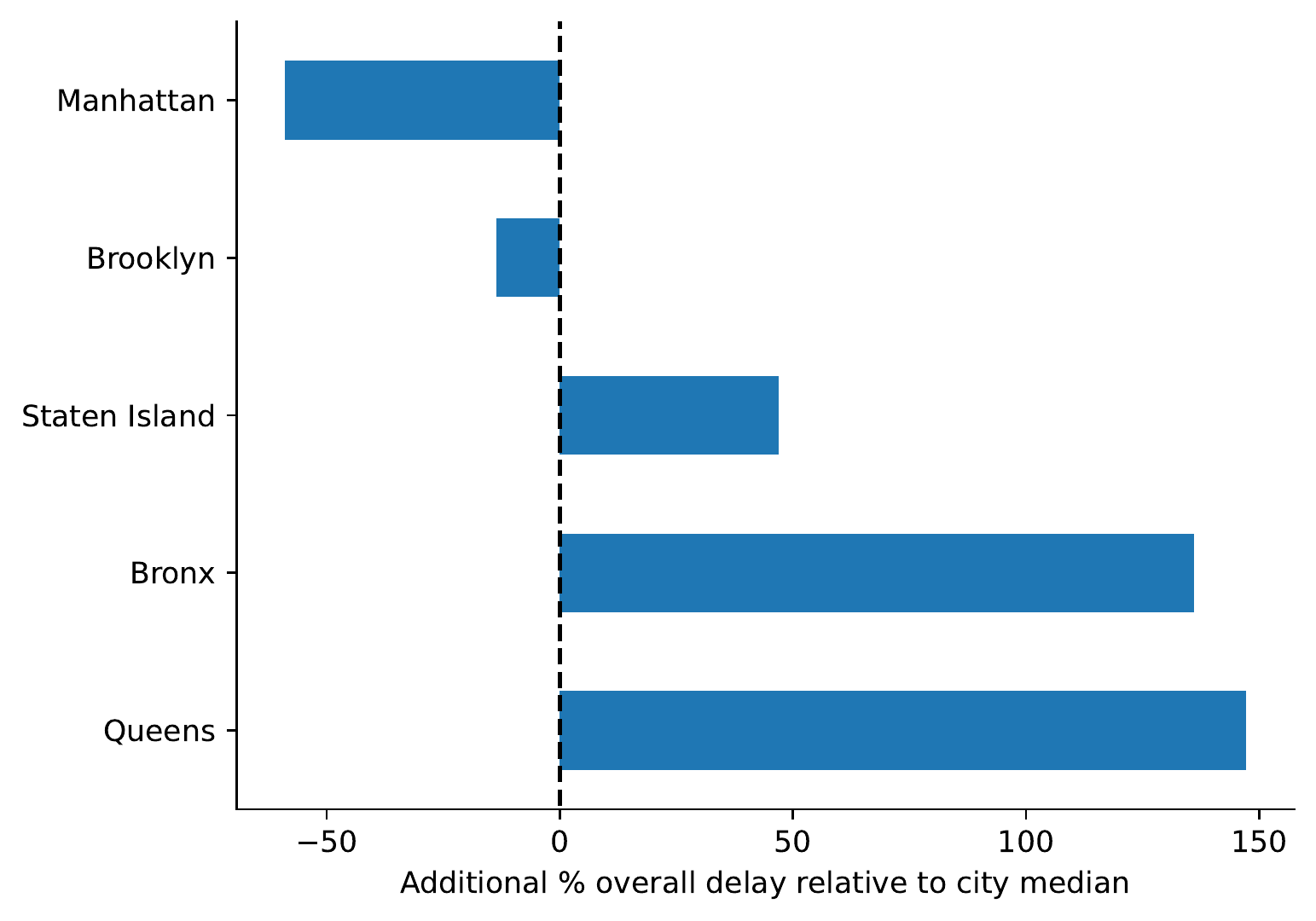}
		% 		\caption{Census tract fixed effects}
		\label{fig:delaysall_equity}
		% 		\end{subfigure}
	}
 	\caption{(a) Median estimated reporting delay, inspection delay, and work order delay in each Borough in NYC DPR data, for the \textit{highest risk} incidents. (Note that Manhattan median inspection and work order delays for such incidents are each less than 1 day.)  Together, these delays represent the end-to-end delay for an incident to be resolved after it occurs. These proportions indicate that reporting delays are a substantial part of the time between incident occurrence to being resolved for such incidents. (b) What is the relative difference between the status quo and if all incidents in the city were reported and responded to at the same rate? For the \textit{highest risk} incidents and each Borough, we calculate the relative difference between the total historical end-to-end delays and the total hypothetical end-to-end delays if each incident had the same delay as the city-wide median for incidents with that risk profile. These comparisons suggest that there are substantial equity concerns in this pipeline. In \Cref{app:alternativeimpactanalyis}, we present analysis details and results with alternative analysis choices for both (a) and (b), and the results are qualitatively similar.}
 	\label{fig:delaysall}
 \end{figure}

We develop a method to quantify reporting delays in crowdsourcing systems and apply it to data from the 311 system in New York City and Chicago. We believe that two aspects of our method lead to uniquely actionable insights for municipal agencies: (a) the measurement of reporting \textit{delays}, even for incidents that are eventually inspected; (b) not relying on ground-truth data on incident occurrence rates, which enables measuring reporting rates as they differ as a function of fine-grained incident characteristics.

 % measuring such \textit{delays}, compared to methods that aim to quantify whether an incident will be reported at all.

In particular, reporting is just the first stage of a response pipeline. After an incident is reported, NYC DPR may schedule an inspector to travel to the report site to analyze the incident. This scheduling decision is made by the agency based on reports -- more hazardous report types are prioritized, as are incidents with more reports and incidents that were reported earlier. Then, based on the inspector's risk determination, the incident may be scheduled for a maintenance crew to address the issue. In each case, operational concerns such as locations and schedules of individual employees and incidents also play a role.

The overall risk and welfare loss an incident poses is a function of the \textit{total} delay between the time that an incident \textit{occurs} and when it is \textit{resolved}. Measurement of reporting delays thus contextualizes \textit{where} in the incident response pipeline is the most inefficiency or inequity -- are reporting delays large compared to downstream delays in agency responses? What are the most impactful \textit{intervention} steps, attempting to improve reporting rates (such as through advertising) or modifying how incidents are addressed once they are reported? Measuring how these values differ by incident characteristics is important: some incidents must be more urgently addressed than other incidents. In this section, in partnership with our collaborators in the DPR, we demonstrate how to use our methods to answer such questions. 

In this section we use the \textbf{zero-inflated Poisson regression} model with all the \textbf{Base} covariates (i.e., the model shown in \Cref{tab:basicinfrefcoef}), with two differences: (a) Risk assessment scores are discretized into four bins (\textit{A-D}), following DPR practice in how risk levels lead to work order prioritization; (b) we include an interaction term between the risk category and the Borough, to highlight heterogeneity. As above, we analyze service requests between 6/30/2017 and 6/30/2020; however, we include inspections and work orders up to 9/01/2022, the latest update before we retrieved the public data. %These adjustments are justified operationally, as subsequent work orders are scheduled based on a priority queue according to these Risk Assessment groups.

We now compare reporting delays to inspection and maintenance delays: for each incident, we estimate its expected reporting delay and compare this value to the inspection delay (after the first report) and work order delay (after inspection). \Cref{fig:delaysall_raw} shows the results: median reporting, inspection, and work order delays, in each Borough, \textit{restricted to the highest risk incidents group A} (which are also prioritized for work orders). \Cref{fig:delaysall_equity} converts the end-to-end delays to relative differences from the city-wide median for such incidents -- we calculate incident-level overall days and then compare Borough-specific median delays to the citywide median. Note that, because the analysis is restricted to a single work order priority level, incidents across Boroughs are comparable; ideal, risk-aware reporting and responses would equalize delays across Boroughs. \Cref{app:alternativeimpactanalyis} contains similar analyses for other priority groups and incident categories.

\paragraph{Substantial end-to-end resolution inequities} There are large differences in end-to-end median delays to address incidents across Boroughs. \textit{Even though the analysis is restricted to incidents that all received the highest work order priority level,} there are differences in how long incidents take to be reported, inspected, and worked on. The cumulative differences are meaningful and point to substantial inequity: highest-risk incidents are resolved within 2 days in Manhattan, and only within 14 days in Queens. If the city instead had \textit{resolution parity} -- incidents of similar risk levels being addressed at similar delays after incident occurrence -- then incidents in Manhattan would be addressed approximately 2x slower, and in Queens approximately 2.5x more quickly than in the status quo.

\paragraph{Reporting delays are substantial compared to downstream delays} 
These estimates also suggest that reporting delays are substantial, on the same scale as both inspection delays and work order delays. Reporting delays further contribute to overall heterogeneity across boroughs: for example, on average, incidents are already resolved in Manhattan by the time they are first reported in any of the other Boroughs. However, it may also be misleading to consider one stage in isolation -- for example, incidents in the Bronx are reported quickly, but not resolved quickly, leading to large overall disparities.

\paragraph{Estimating relative potential of different interventions} These estimates further lend insight on the potential effects of \textit{interventions}: how much would the public's reporting behavior or agency response policies have to change, to improve system performance? For example, while public reporting behavior may be difficult to change (as it is a complex function of many factors, including access factors and trust in government), it may be relatively easy to change agency actions to mitigate the undesirable effects of heterogeneous reporting. Agencies could incorporate reporting delay estimates into their inspection and maintenance scheduling policies: (a) prioritize \textit{individual incidents} which are estimated to have had substantial reporting delays -- for example, for two incidents of the same type, prioritizing the incident with the earlier estimated time of occurrence, instead of report time as in the status quo; (b) invest more resources overall in \textit{neighborhoods} with higher overall estimated reporting delays to respond to reported incidents -- doing so would shorten inspection and work order delays to render overall time from incident to resolution more equitable; (c) designing a proactive inspection program which emphasizes inspections in areas of the city that are known to have lower rates of reporting. NYC DPR is actively considering these interventions, %\nikhil{Uma -- what can we/do we want to say here?}
in response to this work.

Our work provides a principled approach to compare reporting and response delays, enabling such analysis. For example, reducing work order delays (and to a lesser extent, inspection delays) in Staten Island, The Bronx, and Queens would substantially mitigate Borough-level inequity -- e.g., reducing inspection and work order delays down to one day each (as in Manhattan and Brooklyn) would make end-to-end delays in these Boroughs comparable to those in Brooklyn. However, to reach parity with Manhattan (without introducing new delays there), the city would need to increase the reporting rate in every other Borough, such as through advertising or proactive inspections. Similar insights hold for other risk prioritization levels in \Cref{app:alternativeimpactanalyis}. %\todo{nikhil can expand a little bit}

\section{Discussion and Conclusion}

% \nikhil{reviewer comment, I'll do a writing pass -- I found the Discussion / Conclusion a bit underwhelming. I think there is a need to
% describe the limitations of the current demonstration, but that also feeds into a missed opportunity
% to instruct readers on how to carry the work forward. The demonstration here was of a very narrow
% use case—not just 311, but a very small slice of the types of issues that 311 supports. How does this
% generalize? Are there considerations that will need to be made for any generalization? What are the
% fundamental components that need to be adapted to other use cases, both within and beyond the 311
% system? Could we use this logic in some way to adjust for other types of naturally occurring data sets,
% like Yelp ratings, Craigslist postings, 911 reports, restaurant code violations, and more? The authors
% don’t need to answer these questions definitively for any of these or others examples, but I think it’s
% important they provide a roadmap for people excited to work with those data sets}

\parbold{Model applicability to resident crowdsourcing.} Our model captures several key aspects of resident crowdsourcing systems. First, incidents are not directly observed by the city. Rather, the city observes {reports} made by the public, once an incident has occurred. Second, reports usually do not come in immediately after an incident occurs -- it takes a (random) amount of time for someone to first notice an incident and consider it worth reporting. Third, there often are \textit{duplicate} reports, about the same incident. Multiple residents may encounter the same incident and report it. Crucially, city agencies often identify and log such duplicates, as it is important to do so to decide which reports to address and to audit response rates and times---e.g., \citet{kontokosta_equity_2017} mention filtering out duplicates in their work. 

% \textit{Mis-reporting} is also a concern in many settings, where residents manipulate the reported type of an incident to gain an advantage. We in contrast assume that the reported types are correct (in the data section, if different reports for the same incident disagree, we keep the \textit{first} report type). Estimating such misreporting alongside under-reporting is an interesting avenue for future work. Furthermore, the data also enables calculating the \textit{inter-arrival times} between reports. Such information is valuable for estimating non-homogeneous Poisson rates, but we do not do so here. Finally, we note that our method does not directly consider \textit{causal} reasons that a neighborhood may have a lower reporting rate than another.

\paragraph{Applying the method to other settings. }In what other settings is the model applicable? First, while the incident occurrence process $\Lambda_\theta$ can be arbitrary, the reporting %rates $\lambda_\theta$ must be Poisson 
process must be Poisson with rates $\lambda_\theta$ --  this means reports about the same incident must be independent; i.e., our method primarily works for \textit{public} incidents such as potholes and downed trees that may be encountered and reported by separate people. However, the method does not easily handle \textit{private} incidents, such as bed bugs within an apartment or food poisoning, in which potential reporters are likely to communicate with each other. Second, the analysis requires that the agency log information about duplicate reports and that this data be available to researchers -- our empirical analyses were restricted to agencies that make this data publicly available via Open Data portals; applying the methods to other cities or agencies may require data sharing agreements or automated methods for researchers to identify likely duplicate reports for the same incident. Third, it is important to note that our method does \textit{not} require external estimates for how many incidents of each type $\theta$ one expects to see, which may be extremely difficult to estimate for high-dimensional types.

While our data application focuses on reports made in NYC and Chicago to specific city agencies, these requirements are widely met in other crowdsourcing systems in other cities, and in other domains such as software bug reporting (for example, users can report issues and bugs in open source repositories such as GitHub, and there may be heterogeneous reporting based on who the bug affects). Similarly, it is possible to apply our methods to 911 or other emergency systems, for types of incidents where there may be multiple independent witnesses. Such broad applicability is especially important because reporting behavior (and its association with socioeconomic factors) may vary by context. For example, in India, marginalized communities tend to use \textit{formal} channels (such as reporting systems) to access government resources, while dominant groups use informal ones \cite{kruks2011seeking}.

\paragraph{Conclusion.} We believe that our method provides a powerful, general approach to precisely quantifying reporting efficiency and equity in the use of resident crowdsourcing systems, without needing external ground-truth data. In both NYC and Chicago, we find that there is substantial spatial heterogeneity in resident reporting and that this heterogeneity further corresponds to socioeconomic disparities. These estimates directly provide actionable insights, which we are deploying alongside our agency collaborators.  Our work further provides a promising avenue to study the use of reporting systems more broadly in future work.

%Our method can be applied off-the-shelf to other co-production systems, as duplicate reports are already logged by many agencies. \todo{better conclusion}

\section{Empirical method}
\label{sec:empiricsmethod}
% \subsection{Bayesian Poisson regression approach}
\subsection{Methods for parameter estimation: a Bayesian Poisson regression approach}
\label{sec:poissonreg}

\Cref{thm:stoppingtimes} transforms our setting into a standard Poisson rate estimation task. However, that does not mean that the estimation is computationally or statistically tractable. In particular, in our setting, incident type $\theta$ is high-dimensional, but incidents of similar types may have similar rates $\lambda_\theta$. We adopt a tractable Bayesian Poisson regression approach. % to apply to our data.

First, we assume that the reporting rate is affected by the type vector $\theta$ as in a Poisson regression: for each model, we have a \textbf{standard} reporting rate $\lambda_\theta$, defined as:
\begin{align}
\lambda_\theta &= \exp\left(\alpha + \beta^T\theta\right), \label{eqn:lambdaregressionassumption}
\intertext{and so, for each incident $i$, the likelihood follows from \Cref{eq:decompose} and satisfies the following:}
% \tilde M_i &\sim \text{Poisson}\left(\tilde T_i\lambda_{\theta(i)} \right).
\mathcal{L}(\lambda_{\theta_i}|\mathcal{D}_i) &\varpropto p\left(\tilde M_i ; \lambda_{\theta_i}(e_i-s_i)\right),
\end{align}
where $\beta$ is a coefficient vector of appropriate dimension, $(\alpha, \beta)$ are coefficients to be estimated, and $p(\cdot )$ is the Poisson likelihood function as defined in \Cref{eq:poissonlikelihood}.

We further fit a \textbf{zero-inflated Poisson regression} \cite{lambert1992zero}, which corresponds to a data generating process in which, independently for each incident, only one report can be generated with probability $\gamma$, and thus there are 0 duplicate reports: $\tilde M_i = 0$.\footnote{The framework allows for type-dependent $\gamma_\theta$. We do not include type-dependence in the main text, for ease of interpretation. Appendix \Cref{fig:zifposterior_bycat} contains results with type-dependence; results are qualitatively similar.} For example, perhaps those incidents are personal requests for service (e.g., to plant a tree somewhere or to prune a tree to enable certain construction), for which duplicates are unlikely; alternatively, officials may physically mark the location of an incident as ``report received; full inspection pending,'' preventing future reports. Under this case, the likelihood of the data satisfies:
\begin{equation}
    \mathcal{L}(\lambda_{\theta_i}|\mathcal{D}_i) \varpropto 
    \begin{cases}
        (1-\gamma)p\left(\tilde M_i ; \lambda_{\theta_i}(e_i-s_i)\right), & \tilde M_i \ge 1;\\
        \gamma + (1-\gamma)p\left(0 ; \lambda_{\theta_i}(e_i-s_i)\right), & \tilde M_i = 0. 
    \end{cases}
\end{equation}

Such a model is typically used when one observes more zeros than would be expected from a Poisson distribution.

As detailed in \Cref{sec:datapplication}, these models reasonably match the observed distribution of the number of reports received in the observation period, and we run analyses both with and without zero inflation. The feature set can differ across analyses. For example, in the spatial analyses, we add coefficients for each census tract and then include spatial smoothing \cite{morris2019bayesian}. 

Note that our choices are not exhaustive; other specifications may also be appropriate, preferable in various dimensions, or correspond to more realistic data-generating processes; the empirical strategy developed in \Cref{sec:identification} allows a large Poisson rate estimation literature (spatial, time-dependent, feature-dependent) to be applied, off-the-shelf. Our approaches here are chosen for simplicity in model interpretation.

We analyze our methods using simulated data in \Cref{app:simulation}. The simulation results confirm that: (1) under our model, evaluating data within an observation interval correctly disentangles reporting rates from incident rates, unlike a naive approach; and (2) a Poisson regression approach is more data-efficient than a naive MLE approach, yielding estimates with lower variance.

\subsection{Empirical validation of model results}
\label{sec:validation_hurricanes}
We now empirically validate our model estimates. A distinct advantage of our method is that it provides estimates of the reporting delay for specific incidents of each type. In this section, we validate our model-estimated reporting delays from our spatial model -- the model with the richest set of covariates -- against two sets of approximate ground truth on the reporting delays: when the incident time is approximately known due to storm events, and when we predict the known period between the first and second incident. \Cref{app:additionalvalidation} contains additional validations, including comparing our demographic coefficient model to voter participation rates. When reporting correlation significance in this section, we report $p$ values as is standard for Pearson correlations, testing the null hypothesis that the relevant samples are drawn from uncorrelated normal distributions.

% To validate the results shown in \Cref{sec:results}, and primarily the results of the spatial reporting heterogeneity analysis -- which contains coefficients for each census tract and incident-level variables, the most fine-grained results we present -- we leverage a distinct advantage of our methods: the ability to provide estimates on the reporting, which is a point that we elaborate on and rely on in \Cref{tab:ratescontext} and \Cref{sec:conclusion}.

\subsubsection*{Reports after hurricanes: approximately observed incident time}
% To validate the results of our model, we leverage service requests received immediately after strong storms hit New York City. 

Storms cause substantial serious tree damage, including knocking down trees; crucially, \textit{we can approximate the true incident time as the time that the storm in question arrived}. Thus for incidents reported in the days following a storm, we approximately know the true reporting delay. On August 4\textsuperscript{th}, 2020, Tropical Storm Isaias hit New York City, which caused major damage and triggered a travel advisory \cite{isaiasnyc}, predicting that the strongest winds and rains would be from 12 PM to 2 PM on that day. DPR received 15,266 service requests on 8/4 alone, far exceeding their daily average of around 200. We analyze service requests for incidents first submitted to DPR from 12 PM on 8/4 to 12 PM on 8/14 and calculate the true reporting delay of an incident as the time between the submission of the first report associated with it, and 12 PM on 8/4. We then use results from our spatial model to calculate the model-estimated reporting delays for each incident.\footnote{To calculate the model estimated reporting delay, we take the reporting rate for each incident implied by the model and calculate the average reporting delay for that rate, \textit{conditional} on the reporting delay being less than 10 days (to reflect our data filtering to only include incidents reported in that period).} 

\begin{figure}
    \centering
	%	\begin{subfigure}{.5\textwidth}
	% 		\centering
	\subfloat[][Validation on incidents caused by Tropical Storm Isaias.]{
		\includegraphics[width = .44\textwidth]{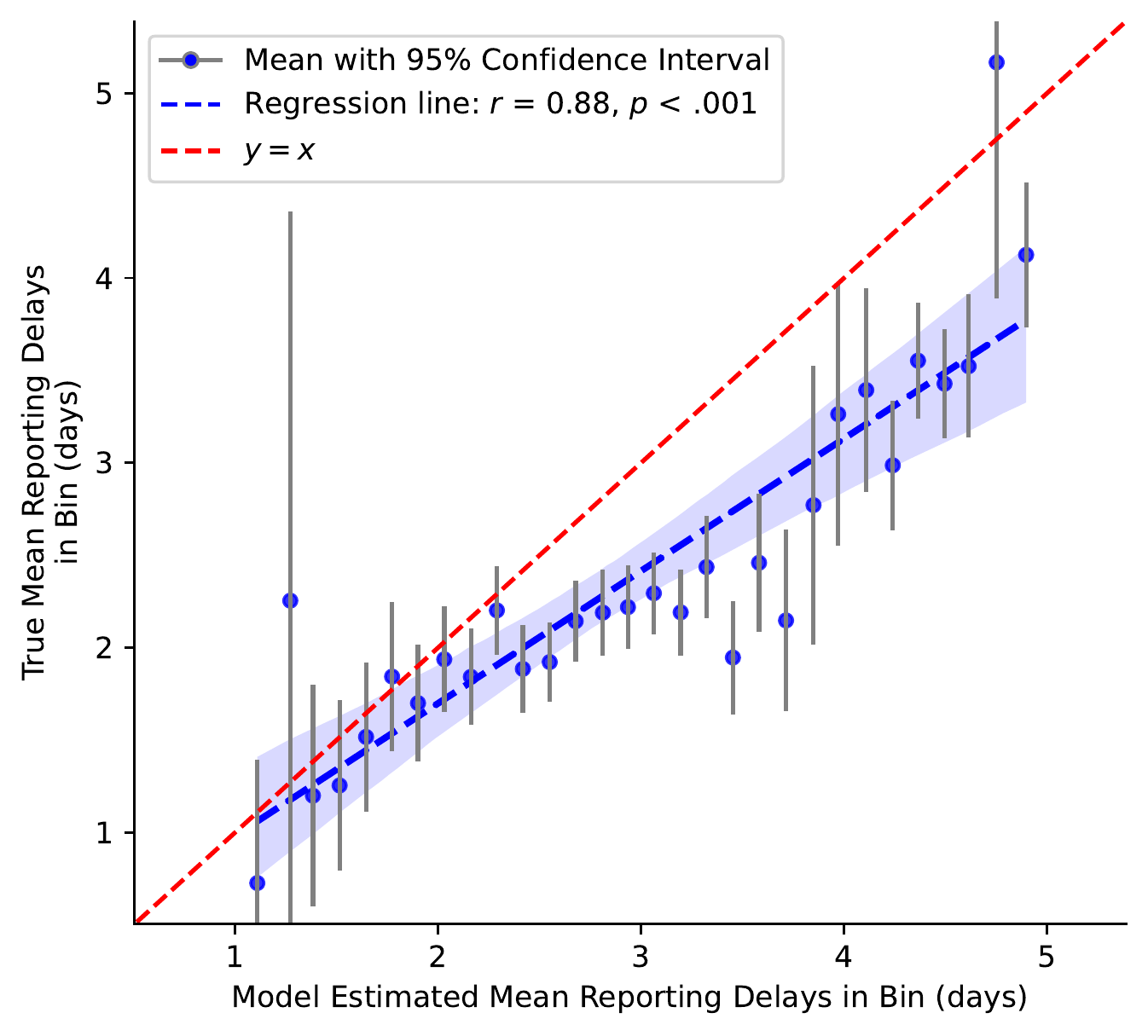}
		% 		\caption{Census tract fixed effects}
		\label{fig:isaiasdelay}
		% 		\end{subfigure}
	}
	\subfloat[][Validation on time between observed first and second reports.]{
		\includegraphics[width=.48\textwidth]{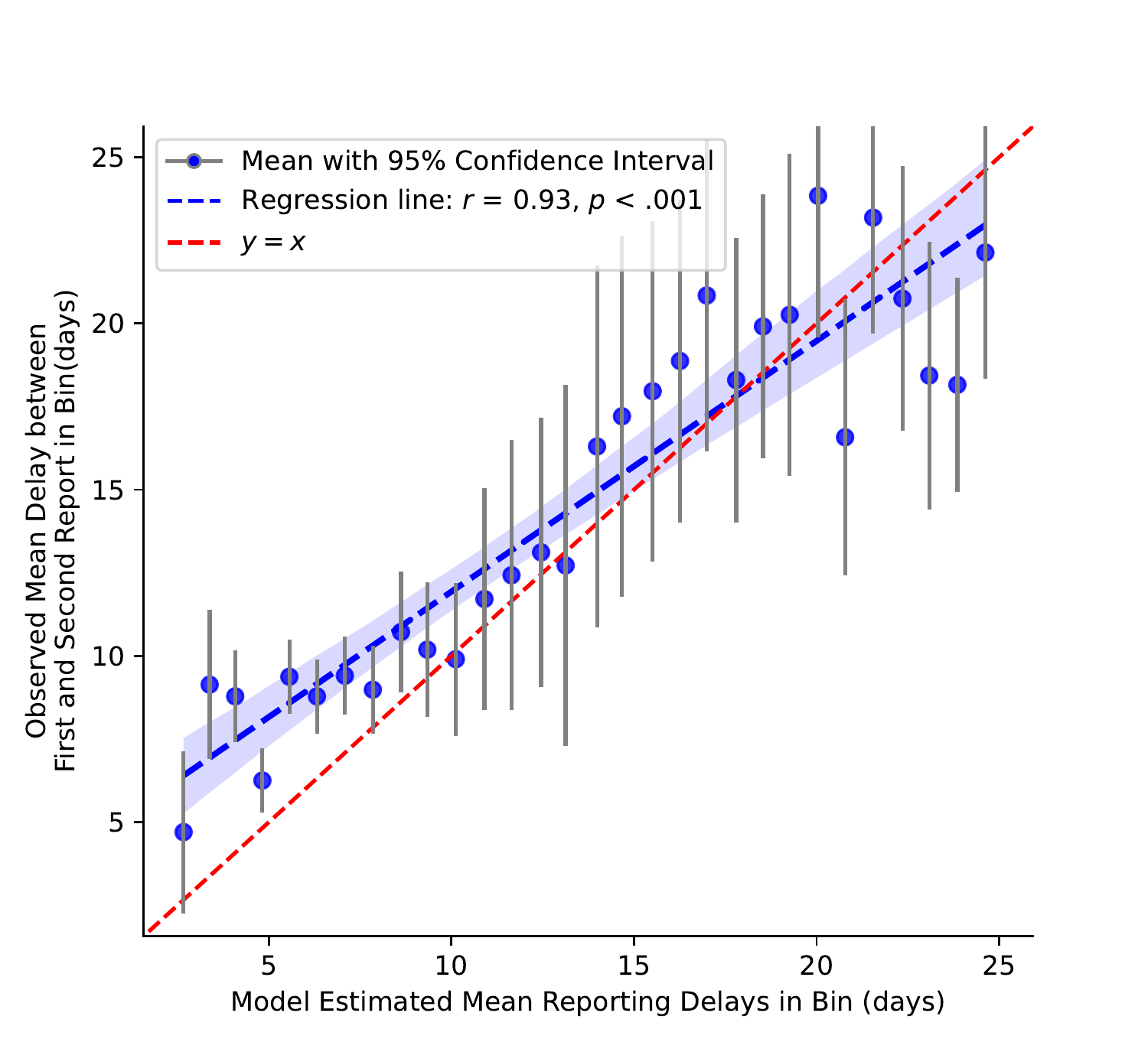}
		% 		\caption{Census tract fixed effects}
		\label{fig:oospredictions}
		% 		\end{subfigure}
	}
 	\hfill
    \caption{(a) Comparison of true reporting delays and model estimated reporting delays, based on data after Tropical Storm Isaias hit New York City on 8/4/2020 until 12 PM on 8/14/2020. True reporting delays for these incidents are calculated as the time between the first report of an incident and 12 PM on 8/4; model-estimated reporting delays are obtained using coefficients learned using the spatial model. (b) Comparison of observed delay between first and second reports and model estimated reporting delays, based on data between 9/1/2020 and 8/31/2022. Model-estimated delays are obtained using the spatial model. For both (a) and (b), all incidents are then categorized into 30 bins based on their estimated reporting delays, and we calculate the means of true and estimated delays within each bin. Model estimates remain correlated without binning. For the storm analysis, at the individual incident level, we observe Pearson $r = 0.20$ with $p<.001$. For true first-to-second report delays at the individual level, Pearson $r = 0.21$ with $p<.001$. Together, these results indicate that our model's estimates of reporting delay are accurate.} %, Spearman $\rho = 0.34$,
\end{figure}

\Cref{fig:isaiasdelay} shows the relationship between model-estimated delays and true reporting delays. There is a strong correlation (Pearson $r = 0.88$, $p < .001$) between binned model estimated reporting delays and the true reporting delays, and predictions are approximately on the $y=x$ line suggesting that our model accurately recovers heterogeneity in incident-level reporting delays.\footnote{We note that our model slightly \textit{over-predicts} reporting delays: reporting behavior likely changes after natural disasters, due to the amount of damage and city communication asking the public to report damage.} \Cref{app:hurricane} contains similar results from analyses using different specifications of the end date, using data on service requests induced by Tropical Storm Ida, and without binning.

\subsubsection*{Delays between first and second reports: substitute for true reporting delay}

% Using the same model as above, we compare our estimated reporting delay against another set of inferred ground truth: the reporting delay between the first and second report. 
Next, we use another form of validation: can our model predict (out-of-sample) true delays between the \textit{first} and \textit{second} report between incidents? If reporting rates are homogeneous (not changing over time for a given incident), then our model should be able to accurately recover these rates, and the model would similarly be accurate for estimating the (unknown) delay before the first report. 
% Based on our model assumptions, for an incident of type $\theta$, the time between any two consecutive reports and the reporting delay are all distributed according to an Exponential distribution with rate $\lambda_\theta$, which we empirically verify as shown in \Cref{fig:daysfirstcomplaint}. Using this insight, we validate our model against another set of approximate ground truth: the observed time between the first and second reports.
% Thus, for incidents which we received two or more reports on, the (observed) time between their first and second reports approximates the (unobserved) reporting delay, which we similarly estimate using our model as in the previous section. 
We analyze service requests submitted to DPR from 9/1/2020 to 8/31/2022 (after the model training data period). For each incident that received 2 or more service requests, we calculate the time between the first and second reports and estimate their reporting delay using the spatial model. 

\Cref{fig:oospredictions} shows the relationship between model-estimated delays and observed delays between the first and second reports when model-estimated delays are binned. As above, there is a strong correlation, and predictions are approximately on the $y=x$ line.

\subsubsection*{Discussion}

These results illustrate the power of our model results in predicting reporting delay, by comparing model estimates to two sets of inferred ground truth: (a) using the arrival of the hurricane as an indication of incident occurrence time; (b) using the observed delay between the first and second reports as a substitute of the unobserved reporting delay. Both validations demonstrate that our model estimates are accurate -- model estimates correspond to actual delays, and  \textit{differences} in model estimates for different incident characteristics correspond to actual differences in reporting delays. We note that both these validation tasks are performed on \textit{out-of-sample} data: the data we trained on only includes service requests up to 6/30/2020, whereas the predicted delays are of incidents first reported after this date. This further suggests that our results are robust to changes over short time periods. \Cref{app:oosmodel} contains further validations of our empirical approach (a) training a model on data that is constructed by defining the time of the second report as the start of the observation interval and then comparing the estimated reporting delay obtained through this model and the observed time between first and second report; (b) comparing our demographic estimates to voter participation rates.

\section{Data and pre-processing}
\label{sec:data}
\label{sec:datapplication}

This section details data and methods used to derive results in \Cref{sec:results}, focusing on NYC analysis.

\subsection{Problem setting, dataset, and preprocessing}

The NYC dataset is composed of partially public data,\footnote{The publicly available part is here: \url{https://data.cityofnewyork.us/Environment/Forestry-Service-Requests/mu46-p9is/data}. The NYC DPR additionally provided us with internal data on inspection and work order outcomes and hashed, anonymized identifiers for the caller's name, email, and phone number, if provided by the caller. We primarily focus on this dataset given a working relationship with NYC DPR, who provided domain expertise in terms of data pre-processing and interpretation.} from the NYC Department of Parks and Recreation (NYC DPR). It contains \num{223416} service requests made by the public between 6/30/2017 and 6/30/2020 about street trees. For each request, we have location and incident-level covariates. We further observe whether the service request led to an inspection (an NYC DPR inspector attempting to find and evaluate the issue), the outcome of the inspection (e.g., an assessment of the risk the problem poses), and what type of work order was created (if any). For each event, we further observe time stamps. Crucially, for inspected requests (63\%), the NYC DPR marks which other requests refer to the same incident. See \Cref{tab:Summary} for summary statistics.

There are two potential sources of bias not captured by the model: (1) duplicate reports may be reported by the same person, and so would not be independent; %(2) due to data labeling practices by NYC DPR, we only observe labels for duplicate reports for inspected incidents. 
(2) due to practices by NYC DPR, we only observe whether a report is a duplicate for inspected incidents. We analyze these biases in \Cref{appendix:bias}, arguing that they should have minimal impact on the validity of our results. Notably: (a) using data provided by NYC DPR, we can partially filter out duplicate reports by the same person and so mitigate bias (1); (b) the Chicago data does not suffer from bias (2). 

% new stuff on preprocessing begins
Before applying our Bayesian Poisson regression approach, we need to construct a dataset where each row corresponds to an incident and contains all relevant covariates. \Cref{app:nycpreprocessing} details this procedure. In short, our pre-processing proceeds in the following steps: (1) we filter uninspected service requests for which we have no duplicate information, duplicate reports by the same caller, and incidents in which the first inspection is logged to be before the first service request (due to, e.g., data errors); (2) we construct an observation interval for each incident, according to  \Cref{thm:stoppingtimes} and the discussion that follows, and then count the number of service requests in this observation interval; (3) for each incident, we join relevant information about it from the service requests it receives, the inspection outcomes logged by NYC DPR inspectors, and 2020 Census data.

After pre-processing, we are left with a dataset of \num{81638} incidents on which we conduct our main analyses; for each incident $i$, we have the duration of observance $e_i-s_i$, the number of total reports $\tilde M_i$, and all the various geographic and demographic covariates associated with it. Appendix \Cref{tab:desc} presents all the covariates in this dataset.

% \usepackage{graphicx}
\begin{table}[tb]
\caption{Summary statistics: service requests and inspections statistics are directly from the NYC DPR reports data; unique incidents statistics are derived from the inspected service requests, where the average duration is calculated using the stopping times definition. The \textit{Others} category includes four other categories: Rescue/Preservation, Remove Stump, Pest/Disease, Planting Space, that together account for less than 0.4\% of the service requests; we exclude them in the analysis. Categories are reported by the person raising the service request: for example, Prune reflects a request to prune the leaves of an over-grown tree, and Hazard reflects a request for NYC DPR to attend to a potentially hazardous condition concerning trees.}
\label{tab:Summary}
\centering
\resizebox{.9\textwidth}{!}{%

\begin{tabular}{r|r|rr|rrr}
\multicolumn{1}{c|}{\textbf{}}             & \multicolumn{1}{c|}{\textbf{\begin{tabular}[c]{@{}c@{}}Service \\ requests\end{tabular}}} & \multicolumn{2}{c|}{\textbf{Inspections}}                                                                                                                          & \multicolumn{3}{c}{\textbf{\begin{tabular}[c]{@{}c@{}}Incidents\\ (from inspected SRs)\end{tabular}}}                                                                                                                                             \\
\multicolumn{1}{l|}{}                      & \multicolumn{1}{l|}{}                                                                     & \multicolumn{1}{c}{\begin{tabular}[c]{@{}c@{}}Inspected \\ SRs\end{tabular}} & \multicolumn{1}{c|}{\begin{tabular}[c]{@{}c@{}}Fraction\\ inspected\end{tabular}} & \multicolumn{1}{c}{\begin{tabular}[c]{@{}c@{}}Unique\\ incidents\end{tabular}} & \multicolumn{1}{c}{\begin{tabular}[c]{@{}c@{}}Avg. reports\\per incident\end{tabular}} & \multicolumn{1}{c}{\begin{tabular}[c]{@{}c@{}}Median Days\\to Inspection\end{tabular}} \\ \hline
\multicolumn{1}{l|}{\textbf{Total number}} &

    \num{223416} &              \num{140057} &                0.63 &            \num{98994} &                  1.41 &                              5.79 \\
\multicolumn{1}{l|}{\textbf{By Borough}}   &                                                                                           &                                                                              &                                                                                     &                                                                                &                                                                                &                                                                                     \\
      \textit{Queens} &     \num{90930} &               \num{55904} &                0.61 &             \num{42724} &                  1.31 &                              5.50 \\
\textit{Brooklyn} &     \num{67852} &               \num{45666} &                0.67 &             \num{27368} &                  1.67 &                             8.68 \\
\textit{Staten Island} &     \num{27263} &               \num{15601} &                0.57 &             \num{11755} &                  1.33 &                              4.07 \\
\textit{Bronx} &     \num{22629} &               \num{14928} &                0.66 &             \num{10059} &                  1.48 &                              6.76 \\
\textit{Manhattan} &     \num{14702} &                7,938 &                0.54 &              7,072 &                  1.12 &                              2.25 \\
\multicolumn{1}{l|}{\textbf{By Category}}  &                                                                                           &                                                                              &                                                                                     &                                                                                &                                                                                &                                                                                     \\
\textit{Hazard} &     \num{87864} &               \num{59667} &                0.68 &             \num{40167} &                  1.49 &                              1.93 \\
\textit{Prune} &     \num{48649} &               \num{26706} &                0.55 &             \num{20589} &                  1.30 &                             10.29 \\
\textit{Remove Tree} &     \num{44177} &               \num{29307} &                0.66 &             \num{22275} &                  1.32 &                              7.73 \\
\textit{Root/Sewer/Sidewalk} &     \num{30856} &               \num{17392} &                0.56 &             \num{14694} &                  1.18 &                             17.86 \\
\textit{Illegal Tree Damage} &     \num{11061} &                6,533 &                0.59 &              5,525 &                  1.18 &                             22.87 \\
\textit{Other} &       \num{809} &                 \num{452} &                0.56 &               \num{429} &                  1.05 &                              5.48 \\
\end{tabular}%
}

\end{table}
% new stuff on preprocessing ends

\subsection{Analyses, training, and evaluation}
\label{sec:analysis}

We fit the models using Stan \citep{carpenter2017stan}, a Bayesian probabilistic programming language that estimates models using Hamiltonian Monte Carlo. All models were run with 150 warm-up iterations and 300 sampling iterations on 4 chains. Priors are generally zero mean Normal distributions, with variance chosen according to the dimensionality and expected scale of the parameters. Models (with full prior information) are in \Cref{app:stancode} and the code repository. 

\paragraph{Covariates included in the analyses} We use three sets of covariates for the three analyses introduced below. Every analysis includes our incident-level covariates (inspection results which include a risk assessment score and the condition of the tree, report \textit{Category}, tree size) and an intercept. 

Our \textbf{Base} analysis includes just these covariates along with \textit{Borough} level fixed effects. As the covariates are standardized, these coefficients all have a zero-mean Normal prior. \textit{Borough} and \textit{Category} are categorical variables; instead of dropping one level for each variable, to maintain identifiability, we enforce a sum-zero constraint\footnote{\url{https://mc-stan.org/docs/2_28/stan-users-guide/parameterizing-centered-vectors.html}} on their coefficients to ease interpretability. For the other categorical covariates, we drop one level each.

In our \textbf{Spatial} analysis, we include coefficients for each of the (over 2,000) census tracts. We use an Intrinsic Conditional Auto-Regressive (ICAR) model~\cite{morris2019bayesian} for the spatial coefficients, alongside the incident level coefficients. In such models, each census tract coefficient is assumed to be normally distributed with a mean equal to the average of its neighbors -- i.e., spatial smoothness is encoded such that census tracts with few incidents borrow information from neighboring census tracts.

In our \textbf{Socioeconomic} analyses, we first include each census \textit{socioeconomic} covariate \textit{individually} with the incident-level covariates (e.g., the regression with \textit{Median age} contains as coefficients the base covariates, but not \textit{Fraction of Hispanic residents}; this choice is due to the census covariates being collinear). These covariates are also standardized, and coefficients' prior means are set to zero. The result of this analysis is presented in \Cref{tab:censuscoefficients}. We further conduct an analysis where we include a subset of these socioeconomic covariates jointly, including log income per capita, the fraction of white residents, the fraction of renters, median age, the fraction of residents with college degrees, and log population density. Note that we do not, for example, further include the fraction of Hispanic residents, which would be highly collinear with the fraction of white residents.

% : log income per capita, fraction of white residents, fraction of renter, median age, and fraction of residents with college degree 

%

\paragraph*{Evaluation of model fit} For the \textbf{Base} covariates, we fit both the \textbf{standard} and \textbf{zero-inflated} Poisson regressions and evaluate model fit. Both models perform similarly in terms of overall predictive power, with Pearson correlations of $0.1335$ and $0.1248$, respectively, for the mean prediction for the number of reports for each incident versus the observed number of reports. Next, we perform a standard check for Bayesian models: whether the posterior prediction distribution matches the observed data distribution, as a test for whether the model captures the underlying data generating process; Appendix \Cref{fig:posteriors} compares the distribution for each model specification. We observe that each model roughly matches the observed distribution, but differs in various ways. The \textbf{standard} Poisson regression under-estimates the (large) fraction that $\tilde M_i = 0$, but more accurately captures the distribution tail. Inversely, the \textbf{zero-inflated} Poisson regression matches the fraction for small $\tilde M_i$ (where the majority of the mass is), at the cost of misestimating the tail. We further experiment with a \textbf{zero-inflated} Poisson regression where the zero-inflation coefficient depends on incident \textit{Category}, with qualitatively similar results. We primarily report results with the \textbf{zero-inflated} Poisson regressions, choosing it because it better matches the primary distribution mass while keeping the model succinct. Results with the \textbf{standard} model are similar.

\paragraph{Reproducibility with public NYC data} The NYC results presented in the main text leverage non-public data provided to us by the NYC DPR. In \Cref{app:nycpublic}, we replicate all analyses with data that are publicly available from the NYC Open Data platform. The results are qualitatively similar. % and we note that though there are minor differences in the resulting estimates, they do not affect our overall evaluations.

\subsection{Analyses using Chicago data} To demonstrate the generalizability of our method, we additionally apply it to data from the City of Chicago. This dataset comes from resident reports made to the City of Chicago's Department of Transportation (CDOT) and Department of Water Management (CDWM) through Chicago's 311 system. The data can be accessed through Chicago Data Portal,\footnote{\url{https://data.cityofchicago.org/Service-Requests/311-Service-Requests/v6vf-nfxy}} which logs all 311 service requests. This dataset is actively maintained by the Chicago Data Portal, and at the time of our access, it contained \num{949352} service requests made by the public in the period between 3/1/2019 and 7/4/2022. \Cref{app:chicagopreprocessing} details the preprocessing. We repeat the analysis done on the NYC DPR data. We use the same sets of covariates for analysis, using the \textbf{zero-inflated} Poisson regression model: a {Base} analysis using all incident-level type characteristics, Spatial model with census block group coefficients and spatial smoothing, and socioeconomic analysis.  

% We use three sets of covariates for analysis. Every analysis uses the \textbf{zero-inflated} Poisson regression model, and the covariates all include the incident \textit{Type} and an intercept term, which we refer to as the \textbf{Base} variables. In our \textbf{Spatial} analysis we include coefficients for each of the census block groups with the ICAR model; in our \textbf{socioeconomic} analysis we include each \textit{socioeconomic} covariate individually with the incident-level covariates.

\FloatBarrier

\FloatBarrier
\begin{footnotesize}  
\bibliographystyle{plainnat}
	\bibliography{311}
\end{footnotesize}

\newpage
\FloatBarrier
\appendix
	 \normalsize

\section{Detailed Related Work}
\label{app:related}
This section discusses the literature on equity and efficiency in crowdsourcing systems, and briefly outlines our contribution to this literature. Given the importance of 311-like systems in allocating government services, there has been much interest in quantifying how efficient and equitable these systems are. Previous works mainly aim to answer two questions: (1) is \textit{reporting} behavior similar across socioeconomic groups, and (2) do governments \textit{respond} to the reporting equitably, or do they prioritize some neighborhoods for the same types of requests? If there are disparities in either stage of the process, then government services may be allocated inequitably.

Towards the second question, there has been a long line of work since the late 1970s, documenting how the allocation of government services has changed since the adoption of co-production systems. Early works warned of the potential of biases in how governments responded \citep{jones1977,thomas1982}. Researchers have found response time differences between neighborhoods, e.g., during Hurricane Katrina \citep{elliott_race_2006}; other recent studies suggest that the practical differences induced by these differential response times are negligible \citep{clark_advanced_2020} or explained by other factors \citep{xu_closing_2020}.

This work considers the first question, of how reporting behavior varies across demographic and incident characteristics. The literature here focuses on two types of reporting behavior: \textit{under-reporting}, and \textit{misreporting}. The latter refers to when residents report problems that, upon inspection by the agency, are not found or are less severe than reported; there is some evidence that such misreporting occurs and is heterogeneous across areas \citep{mclafferty_placing_2020}.

On the other hand, differential \textit{under-reporting} is when people, faced with similar problems, differentially report those problems to the government. Such differences might emerge, for example, due to access to communications technologies, familiarity with the system, or trust in government \cite{buell2021surfacing}. As discussed above, the biggest challenge to studying \textit{under-reporting} is that researchers rarely directly observe an incident \textit{unless it is reported}. It is thus difficult to disentangle whether certain areas have fewer reports because problems truly occur less frequently there, or whether people in those areas are reporting less frequently given a similar distribution of problems.

One line of work in answering this question does not attempt to distinguish between these two possibilities. Such work entails regressing the number of reports (or the number of unique incidents, if some incidents are reported multiple times) as a function of socioeconomic and report characteristics, as well as potentially space and time \citep{clark_coproduction_2013,cavallo_digital_2014,minkoff_nyc_2016}. Such works have found that the number of reports may (but does not always) differ by wealth, race, and education level. However, as we formalize in \Cref{sec:identification}, such a method cannot identify whether these fewer reports are due to fewer true incidents, or less reporting for the same incidents. Rather, these works must assume that the latter is the cause. \citet{akpinar2021effect} point out that such an assumption may not always hold, affecting downstream models, in the context of crime reporting and predictive police systems.

Another approach to quantifying under-reporting is to construct a proxy for the true incident rate and then compare the estimates with the observed rates. \citet{kontokosta_bias_2021} analyze pothole complaints in Kansas City, Missouri, at a fine-grained spatial level; they leverage additional street assessment data (resulting from scheduled visual inspections) that grade the quality of roadways to construct relative predictions for the true number of potholes. They find that low-income and minority neighborhoods are less likely to report street conditions or ``nuisance'' issues while prioritizing more serious problems. \citet{hacker_spatiotemporal_2020} use a similar metric, with the addition of temporal trends geared more towards an epidemiological study, and \citet{pak_fixmystreet_2017} use roadway length as part of their proxy. \citet{kontokosta_equity_2017} analyze residential building problem complaints in New York City's 311 system---using building conditions as a proxy---and also find that many socioeconomic factors contribute to differential reporting behavior. \citet{obrien_ecometrics_2015} use administrative records from active broken streetlight inspections. While the approach can disentangle reporting behavior from true incident rates, it heavily relies on the accuracy of the `ground-truth' proxy estimate. It may be difficult to find such proxies or to validate their accuracy, especially for different types of incidents within the same class (e.g., more vs less serious potholes), or more generally for high dimensional types: for example, estimating ground truth rates in every census tract for every kind of incident. In particular, note that each of the above examples estimates ground truth rates for a specific type of incident, in a specific time period. 

The challenge between identifying occurrence rates as opposed to reporting rates is also present in the ecology literature when counting the number of animals of a given species in an area. Some works similarly use proxy methods -- such as by normalizing by occurrence rates in neighboring geographies or using reference species whose occurrence rates are approximately known \cite{hill2012local,renner2013equivalence,kery2008hierarchical}.

To this literature, our work contributes a statistical technique to identifying under-reporting---that does not require such external data or the construction of proxies for the ground truth incident rate. Rather, it relies only on report data already logged by many agencies and connects this task to a large literature on Poisson rate estimation. In applying our method to 311 data, we find reporting rate differences across neighborhoods associated with socioeconomic status and show that these differences are beyond what one would expect from just incident-level characteristics.

\section{Omitted technical results and proofs}

\subsection{Non-identifiability of reporting rate with observed incidents count}

\begin{restatable}{proposition}{propnonidenti}\label{prop:identi} Consider the simplest setting: both incident occurrence and reporting follow time-homogeneous Poisson processes, with $\Lambda_\theta(\tau) = \Lambda_\theta$ and $\lambda_\theta(\tau) = \lambda_\theta$. Further, let incident reporting duration $T_i$ be distributed according to $F$. Then using just the number of observed incidents $N_\theta^\obs(T)$ of type $\theta$ in a known time duration $[0, T]$, the reporting rate $\lambda_\theta$ is not identifiable. In other words, $N_\theta^\obs(T)$ is a function of both $\Lambda_\theta$ and $\lambda_\theta$:
    \begin{align*}
        &\lim_{T \to \infty} \frac{N_\theta^\obs(T)}{T} = \Lambda_{\theta}'\\
        \text{where\,\,\,\,\,\,\,\,\,\,\,\,\,\,\,\,\,} &\Lambda_{\theta}' = \Lambda_{\theta}\left[1-\int_{0}^\infty\exp\left(-\lambda_\theta t\right)\text{d}F(t)\right].
    \end{align*}
\end{restatable}

The proof of Proposition~\ref{prop:identi} follows directly from the following Lemma.

\begin{lemma}\label{observed_rate}
Suppose each incident gets reported independently, and the distribution of the interval of reporting duration $T_i$ has density function $f(\cdot):[0,\infty)\mapsto \mathbb{R}$. Then under steady state, $N_\theta^{\obs}$ follows a Poisson process with rate $\Lambda_{\theta}'$, where
$$
\Lambda_{\theta}' = \Lambda_{\theta}\left[1-\int_{0}^\infty\exp\left(-\int_{0}^t\lambda_{\theta}(u)du\right)f(t)dt\right].
$$

In the simplest time-homogeneous case, this rate simplifies to:
$$
\Lambda_{\theta}' = \Lambda_{\theta}\left[1-\int_{0}^\infty\exp\left(-\lambda_\theta t\right)f(t)dt\right].
$$
\end{lemma}

\begin{proof}[Proof of \Cref{observed_rate}]
Let $m(t)$ be the number of times an incident is reported in an interval of $t$, starting from its birth. We know from the model assumption that $m(t)$ follows a Poisson distribution:
$$
m(t)\sim \text{Poisson}\left(\int_0^t\lambda_\theta(u)du\right).
$$

In steady state, each unique incident gets reported with probability
\begin{align*}
p &= \int_{0}^{\infty}\text{Pr}[m(t)\ge 1|T_i=t]f(t)dt\\
&=\int_{0}^\infty \left[1-\exp\left(-\int_{0}^t\lambda_{\theta}(u)du\right)\right]f(t)dt\\
&= 1-\int_{0}^\infty\exp\left(-\int_{0}^t\lambda_{\theta}(u)du\right)f(t)dt,
\end{align*}
which simplifies under time-homogeneity to
\begin{align*}
p &= 1-\int_{0}^\infty\exp\left(-\lambda_{\theta}t\right)f(t)dt.
\end{align*}

Over a time interval of $t$, the total number of incidents of type $\theta$ that happen, $N_\theta(t)$ follows a Poisson process with rate $\Lambda_\theta$. Conditional on $N_\theta(t)=n, n=0,1,\dots$, under steady state, $N_{\theta}^\obs(t)$ follows a binomial distribution with parameters $(n,p)$. Thus $N_{\theta}^\obs(t)$ follows a Poisson distribution with rate $\Lambda_\theta p$, which completes the proof.

\end{proof}

\begin{proof}[Proof of Proposition~\ref{prop:identi}.]
Lemma~\ref{observed_rate} establishes that the rate at which we observe unique incidents depends on a lot of various aspects. Under steady state, $N_{\theta}^\obs$ follows a Poisson process with parameter $\Lambda_\theta'$, where $\Lambda_\theta'$ is a function of the incident occurrence rate $\Lambda_\theta$, the (potentially non-homogeneous) reporting rate $\lambda_\theta(\cdot)$ and the distribution of reporting duration $f(\cdot)$.

In practice, when the observation period length $T$ is large, we can safely assume that for each period $[\tau,\tau+1), \tau=0,\dots,T$, the observed reports  $N_\theta^\obs\left([\tau,\tau+1]\right)$ are close to steady state, and thus follow independent and identical Poisson$(\Lambda_\theta')$ distribution. Following the law of large numbers we get
$$
\lim_{T \to \infty} \frac{N_\theta^\obs(T)}{T} = \lim_{T \to \infty} \frac{\sum_{\tau=0}^{T-1}N_\theta^\obs\left([\tau,\tau+1]\right)}{T} = \mathbb{E}\left[N_\theta^\obs\left(1\right)\right]=\Lambda_{\theta}'.
$$

Thus, if we are only using information about the unique incidents, it is impossible to determine $\lambda_\theta$ without having full knowledge about both $\Lambda_\theta$ and $f(\cdot)$.

\end{proof}

\subsection{Identifiability of reporting rate with duplicate reports}
\label{sec:mainthmproof}
In this subsection, we restate \Cref{thm:stoppingtimes} formally in the notation of stochastic processes. Then, we show that when an observation interval is properly defined, evaluating the likelihood of the data within this observation interval identifies the reporting rate.

The core challenge is that, because we do not observe true incident birth (and potentially true incident death), our ``observation interval'' -- when we observe the Poisson reporting process -- \textit{is itself random and dependent on the Poisson reporting process}. For example, we can only start the counting interval at the time of the first report, and the agency responds to the incident partially as a function of the number of reports. Thus, the distribution of the entire data, not just the number of reports within the interval (i.e., including start and stop times) could be a complex, unknown function of $\lambda$. The theorem formalizes that, as long as the start and end times depend only on the rate $\lambda$ through the number of reports up to those times, we can nevertheless decompose the likelihood and reduce the task to Poisson rate estimation.

Formally, we separate terms involving $\lambda$ from terms independent of $\lambda$, which requires the two conditions; then, we arrange the likelihood function to exhibit the structure of a Poisson distribution likelihood. %, where we rely on \Cref{lem:exponentialtime} above.

The following \Cref{thm:decomp} is a restated version of \Cref{thm:stoppingtimes}, where we use standard Poisson process notation, and state mathematically the two conditions we introduce in \Cref{thm:stoppingtimes}. %\Cref{thm:stoppingtimes} is equivalent to \Cref{thm:decomp}, by changing notations to conform to the resident reporting process that we consider.
%, and restricting the process to a finite end time $t_i+T_i$, instead of extending to $\infty$.

\begin{theorem}
\label{thm:decomp}
Consider a Poisson process $\{X(t), T \ge t \ge 0\}$ with rate $\lambda$ (where $X(t)$ denotes the number of jumps up to time $t$), and denote the jumping times $\{T_1, T_2, \dots\}$. Denote the random variables of the start and end of an observation period $S$ and $E$, and suppose $S$ and $E$ satisfy the following conditions, respectively:
    
\noindent\textbf{Condition 1:} Given $T_1 = t_1$, $S$ is defined on $[t_1, T]$, and 
\[
\bbP(S = t|T_1 = t_1) = g(t), \forall T \ge t \ge t_1,
\]
where $g(t)$ is not dependent on $\lambda$ or the sample path of the Poisson process after $t_1$.

\noindent \textbf{Condition 2:} Given the jumping times of the Poisson process up to any time $t$, that is, for some $m\in \mathbb{Z}+$, given $\{T_1 = t_1, \dots, T_m = t_m, T_{m+1}>t\}$, the distribution of $E$ is
\[
\bbP[E = t | \{T_1 = t_1, \dots, T_m = t_m, T_{m+1}>t\}] = h_m(t), \text{ for some } m\in \mathbb{Z}+, \forall T\ge t\ge t_m,
\]
where $h_m(t)$ is a function independent of $\lambda$. Furthermore, given $T_1 = t_1$, the distribution of $E$ is independent of the realization of $S$:
\[
\bbP[E = t| T_1 = t_1, S = s] = \bbP[E = t|T_1 = t_1], \forall T\ge t \ge t_1,  \forall s\ge t_1.
\]

Now, suppose we observe data which contains for each incident, $S = s, E = e, X(e) - X(s) = M$, where $X(s)  = m \ge 1$, and all the jumping times within interval $(s, e]$: $\mathcal{J}_{M} = \{T_{m+1} = t_{m+1}, \dots, T_{m+M} = t_{m+M}\}$, where $t_{m+1}>s$ and $T_{m+M}\le e$. Then, conditional on the observed jumping times before $s$: $\mathcal{H}:=\{T_1 = t_1, \dots, T_m = t_m \le s\}$, the likelihood of the data can be decomposed as follows:
\begin{equation}
\bbP[X(e) - X(s) = M, E = e, S = s, \mathcal{J}_M | \mathcal{H}] = f(e,s, \mathcal{H}, \mathcal{J}_{M})p(M|\lambda(e-s)), \label{eq:thm}
\end{equation}
where $p(M|\lambda(e-s))$ denotes the likelihood of a Poisson distribution with $M$ occurrences and rate parameter $\lambda(e-s)$, and $f(e,s, \mathcal{H}, \mathcal{J}_{M})$ is a function that does not depend on $\lambda$.
\end{theorem}

The proof of the theorem relies on the following lemma, which establishes the distribution of the (random) time between $S$ and the first jump time after $S$.

\begin{lemma}
\label{lem:exponentialtime}
If the distribution of $S$ satisfies \textbf{condition 1}, then the distribution of the time between any realization of $S$ and the next jump time $T_{X(S)+1}$ is an exponential random variable with rate $\lambda$.
\end{lemma}

\begin{proof}
    Following the definition of jump times, we have for any $\epsilon > 0$, 
\begin{align}
    &\bbP[T_{X(S) + 1} - S > \epsilon|T_1 = t_1] \\
    =& \sum_{m = 1}^{\infty}\bbP[X(S+\epsilon) = m, X(S) = m|T_1 = t_1]\\
    =& \sum_{m = 1}^\infty \int_{s \ge t_1} \bbP[X(s+ \epsilon) = m, X(s) = m|S = s, T_1 = t_1] \bbP[S = s|T_1 = t_1] ds\\
    = & \sum_{m = 1}^\infty \int_{s \ge t_1} \bbP[X(s + \epsilon) = m|X(s) = m, S = s, T_1 = t_1] \bbP[X(s) = m| S = s, T_1 = t_1] g(s) ds \label{eq:fullconditional}\\
    = & \sum_{m = 1}^\infty\int_{s \ge t_1} \bbP[X(\epsilon) = 0] \bbP[X(s-t_1) = m-1] g(s) ds \label{eq:memoryless}\\
    = & \bbP[X(\epsilon) = 0] \int_{s\ge t_1} g(s) \left[\sum_{m = 1}^{\infty}\bbP[X(s-t_1) = m-1]\right]ds \label{eq:exchange}\\
    = & \exp(-\lambda\epsilon) \int_{s\ge t_1} g(s) ds \label{eq:sumone}\\
    = & \exp(-\lambda \epsilon) \label{eq:intone}
\end{align}
where from \Cref{eq:fullconditional} to \Cref{eq:memoryless} we use the memoryless property of Poisson processes jump times, and denote $\{X(t), t\ge 0\}$ as a new Poisson process with rate $\lambda$, starting from $t = 0$. From \Cref{eq:memoryless} to \Cref{eq:exchange} we use $\bbP[X(\epsilon) = 0] = \exp(-\lambda \epsilon)$ is independent of $s$, $g(s)$ is independent of $m$ (as a result of it being independent of the sample path after $t_1$), and $\bbP[X(s-t_1) = m-1] \ge 0, \forall m\ge 1, s\ge t_1$, thus we can exchange the summation and integral, and put $g(s)$ outside of the summation. From \Cref{eq:exchange} to \Cref{eq:sumone} we use the fact that the summation is of the probability mass function of a Poisson random variable with rate $\lambda(s-t_1)$ over its range, thus evaluates to 1, and finally the last equation follows from the definition of $g(s)$ as a probability density function. We note that the final equation corresponds to the tail probability of an exponential random variable with rate parameter $\lambda$, and thus conclude our claim.
\end{proof}

\begin{proof}[Proof of Theorem 2]

We now prove \Cref{thm:decomp}. The main idea of the proof is to decompose the likelihood into full conditional probabilities, and then express each part using the conditions we laid out and \Cref{lem:exponentialtime}. The difficulty comes from separating terms involving $\lambda$ from terms independent of $\lambda$, which requires the two conditions, and arranging the likelihood function to exhibit the structure of a Poisson distribution likelihood, where we rely on \Cref{lem:exponentialtime} above.

We prove the result in three cases, based on how many reports are observed within $(s,e]$: 0, 1, and 2 or more, as these scenarios have varying level of complexity.

As reference, the likelihood of a Poisson distribution with $M\ge 0$ occurrences and rate $\lambda(e-s)>0$ is as follows:
\[
p(M|\lambda(e-s)) = \frac{\left[\lambda (e-s)\right]^M}{M!}\exp(-\lambda(e-s)).
\]

First, for $M = 0$, since we do not observe any jump times within $(s,e]$ in this case, $\mathcal{J}_0$ is an empty set and we omit it from the equation.
\begin{align}
    &\bbP[S = s, E = e, X(e) = X(s) = m |\mathcal{H}] \label{eq:m01}\\
    = &\bbP[S = s|\mathcal{H}]\bbP[E = e| S = s, \mathcal{H}] \bbP[T_{m+1} - s > e - s|E = e, S = s, \mathcal{H}]\label{eq:m02}\\
    = & g(s) h_m(e) \exp(-\lambda (e-s)) \label{eq:m03}\\
    = & f(e,s,\mathcal{H}, \mathcal{J}_{M}) p(0|\lambda(e-s)) \label{eq:m0res}
\end{align}

From \Cref{eq:m01} to \Cref{eq:m02} we decompose the likelihood by the full conditional probabilities. The likelihood of the data contains three parts: (i) conditional on the history of reports $\mathcal{H}$, the likelihood of observing a realization of $S=s$; (ii) conditional on the history of reports $\mathcal{H}$ and $S=s$, the likelihood of observing a realization of $E = e$; (iii) conditional on the history of reports $\mathcal{H}$, $S=s$ and $E=e$, the likelihood of \textit{not} observing the next jump time, that is $T_{m+1}> e$.% \nikhil{I changed from <= to >}

From \Cref{eq:m02} to \Cref{eq:m03} we use \textbf{condition 1} to express (i) as $g(s)$, \textbf{condition 2} to express (ii) as $h_m(e)$, and \Cref{lem:exponentialtime} to express (iii) as $\exp(-\lambda(e-s))$. The final equation follows by defining 
\[
f(e,s, \mathcal{H}, \mathcal{J}_{M}) = g(s)h_m(e),
\]
which is independent of $\lambda$.

For the case with $M = 1$, we cannot omit $J_1$, but similarly have:
\begin{align}
    &\bbP[S = s, E = e, X(e) = m + 1, X(s) = m, \mathcal{J}_1 |\mathcal{H}] \label{eq:m11}\\
     = & \bbP[S = s, E = e > t_{m+1}, T_{m+1} = t_{m+1}, T_{m+2}>e |\mathcal{H}] \label{eq:m1insert}\\
    = &\bbP[S = s|\mathcal{H}] \times \bbP[E > t_{m+1}| S = s, \mathcal{H}] \times \bbP[T_{m+1} - s = t_{m+1} - s|S = s, E > t_{m+1}, \mathcal{H}] \nonumber\\
    &\times \bbP[E = e| S = s, \mathcal{H}, \{T_{m+1} =t_{m+1}\}]\times \bbP[T_{m+2}> e| E = e, S = s, \mathcal{H}, \{T_{m+1} = t_{m+1}\}]\label{eq:m12}\\
    =& g(s)\times \left[\int_{t_{m+1}}^{T} h_{m}(t)dt\right]\times \lambda\exp(-\lambda (t_{m+1} - s)) \times h_{m+1}(e)\times \exp(-\lambda (e-t_{m+1}))\label{eq:m13}\\
    = &\lambda \exp(-\lambda(e-s)) g(s)h_{m+1}(e)\left[\int_{t_{m+1}}^T h_m(t)dt\right] \label{eq:m14}\\
    = & f(e,s, \mathcal{H}, \mathcal{J}_{M}) p(1|\lambda(e-s)) \label{eq:m1res}
\end{align}

From \Cref{eq:m11} to \Cref{eq:m1insert} we express the data in a different but equivalent manner. Note that the fact $\mathcal{J}_1$ contains the first jump time $t_{m+1}$. It (along with $X(e) = m + 1$ equivalently can be stated as two events: first, trivially, we know $T_{m+1} = t_{m+1}$; second, we also know that $T_{m+2}$ arrives later than $e$, which is why that jump is not observed. From \Cref{eq:m1insert} to \Cref{eq:m12} we similarly decompose the likelihood into full conditional probabilities. The likelihood now contains five parts: (i) conditional on the history of reports $\mathcal{H}$, the likelihood of observing $S=s$; (ii) conditional on the history of reports $\mathcal{H}$ and $S=s$, the likelihood of observing $E$ to be greater than $t_{m+1}$; (iii) conditional on the history of reports $\mathcal{H}$, $S=s$ and $E$ still not realized, the likelihood of observing $T_{m+1} = t_{m+1}$; (iv) conditional on the history of reports $\mathcal{H}$, $S=s$, $T_{m+1} = t_{m+1}$, the likelihood of observing a realization $E = e$; (v) conditional on the the history of reports $\mathcal{H}$, $S=s$, $E = e$, and $T_{m+1} = t_{m+1}$, the likelihood of \textit{not} observing the next jump time, that is $T_{m+2}>e$. 
% \nikhil{what is T in line 23?} \nikhil{I still think it's not obvious to me why the terms in (22) are equivalent to (21). like before you break it up into conditionals, I think you should write it as a joint where the events are in the same order as their equivalent events in (21), with any non-trivial transformation explained.}

From \Cref{eq:m12} to \Cref{eq:m13} we use \textbf{condition 1} to express (i), \textbf{condition 2} to express (ii), \Cref{lem:exponentialtime} to express (iii), \textbf{condition 2} again to express (iv), and then the memoryless property of Poisson process jump times to express (v). \Cref{eq:m14} follows by collecting terms. The final equation follows by defining
\[
f(e,s, \mathcal{H}, \mathcal{J}_{M}) = \frac{1}{(e-s)}g(s)h_{m+1}(e)\int_{t_{m+1}}^{T}h_m(t)dt,
\]
which is independent of $\lambda$.

Next, for any $M> 1$, we similarly have:
\begin{align}
    &\bbP[S = s, E = e, X(e) = m + M, X(s) = m, \mathcal{J}_M |\mathcal{H}]\label{eq:m21}\\
    = & \bbP[S = s, E = e > t_{m+M} > \dots > t_{m+1}, T_{m+1} = t_{m+1}, \dots, T_{m+M} = t_{m+M}, T_{m+M+1} > e |\mathcal{H}] \label{eq:m2insert}\\
    = &\bbP[S = s|\mathcal{H}]
    %\times \bbP[E > t_{m+1}| S = s, \mathcal{H}] \times \bbP[T_{m+1}= t_{m+1}|S = s, E > t_{m+1}, \mathcal{H}] 
    \nonumber\\
    &\times \prod_{i = 1}^{M} \bbP[E > t_{m+i}|S = s, \mathcal{H}, \{T_{m+1} = t_{m+1}, \dots, T_{m+i-1} = t_{m+i-1}\}] \nonumber\\
    &\times \prod_{i = 1}^{M} \bbP[T_{m+i} = t_{m+i}|S = s, \mathcal{H}, \{T_{m+1} = t_{m+1}, \dots, T_{m+i-1} = t_{m+i-1}\}] \nonumber\\
    &\times \bbP[E = e| S = s, \mathcal{H}, \{T_{m+1} =t_{m+1}, \dots, T_{m+M} = t_{m+M}\}] \nonumber\\
    &\times \bbP[T_{m+M+1}> e| E = e, S = s, \mathcal{H}, \{T_{m+1} =t_{m+1}, \dots, T_{m+M} = t_{m+M}\}] \label{eq:m22}\\
    =&g(s)
    %\int_{t_{m+1}}^{\infty} h_m(t)dt \times \lambda \exp(-\lambda (t_{m+1}-s)) 
    \nonumber\\
    &\times \prod_{i=1}^{M} \lambda \exp(-\lambda (t_{m+i}- t_{m+i-1}))\nonumber\\
    &\times \prod_{i=1}^{M} \int_{t_{m+i}}^T h_{m+i-1}(t)dt \nonumber\\
    &\times h_{m+M}(e) \times \exp(-\lambda (e-t_{m+M})) \label{eq:m23}\\
    = & \lambda^{M} \exp(-\lambda (e-s))  \times g(s)h_{m+M}(e)\prod_{i = 1}^M \int_{t_{m+i}}^T h_{m+i-1}(t)dt \label{eq:m24}\\
    = & f(e,s, \mathcal{H}, \mathcal{J}_{M}) p(M|\lambda(e-s)). \label{eq:m2res}
\end{align}

From \Cref{eq:m21} to \Cref{eq:m2insert} we similarly express the data in an equivalent manner, also noting that the data provides information that $T_{m+M+1} >e$. From \Cref{eq:m2insert} to \Cref{eq:m22} we decompose the likelihood by the full conditional probabilities, which similar to the previous case contains five parts: (i) conditional on the history of reports $\mathcal{H}$, the likelihood of observing $S=s$; (ii) for $i = 1,\dots, M$, conditional on the history of reports $\mathcal{H}$, $S=s$ and all jump times $T_{m+1}$ to $T_{m+i-1}$,\footnote{Where $i=1$, this degenerates to be an empty set, the same for part (iii) below.} the likelihood of observing $E$ to be greater than $t_{m+i}$; (iii) for $i = 1,\dots, M$, conditional on the history of reports $\mathcal{H}$, $S=s$, all jump times $T_{m+1}$ to $T_{m+i-1}$ and $E$ still not realized, the likelihood of observing $T_{m+i} = t_{m+i}$; (iv) conditional on the history of reports $\mathcal{H}$, $S=s$, all jump times $T_{m+1},\dots, T_{m+M}$, the likelihood of observing a realization $E = e$; (v) conditional on the the history of reports $\mathcal{H}$, $S=s$, $E = e$, and $T_{m+1}\dots T_{m+M}$, the likelihood of \textit{not} observing the next jump time, that is $T_{m+M+1}>e$.

From \Cref{eq:m22} to \Cref{eq:m23} we use \textbf{condition 1} to express (i), \textbf{condition 2} to express (ii), \Cref{lem:exponentialtime} to express (iii), \textbf{condition 2} again to express (iv), and then the memoryless property of Poisson process jump times to express (v). \Cref{eq:m24} follows by collecting terms, and the last step follows by defining
\[
f(e,s, \mathcal{H}, \mathcal{J}_{M}) = \frac{(e-s)^M}{M!}g(s)h_{m+M}(e)\prod_{i = 1}^M \int_{t_{m+i}}^T h_{m+i-1}(t)dt,
\]
which is independent of $\lambda$.

Combining Equations \ref{eq:m0res}, \ref{eq:m1res} and \ref{eq:m2res} concludes our claim.

\end{proof}

\FloatBarrier

\section{Comparing estimation methods via simulation}
\label{app:simulation}

% \usepackage{graphicx}
\begin{table}[tb]
\centering
\caption{Simulation results comparing five estimates of reporting rate $\lambda_\theta$ under different incident rates $\Lambda_\theta$. The true reporting rates are all $\lambda_\theta = 2$. We find that with correctly specified stopping times, the Poisson regression estimators are more precise than the MLE, especially as the incident rate decreases. The naive estimator, or the incorrectly specified observation ending time both introduce bias into the estimation. Estimate standard deviations are in parentheses.}
\label{tab:sim}
\resizebox{\textwidth}{!}{%
\begin{tabular}{l|rrrrr}
                                     & \multicolumn{5}{c}{\textbf{Incident rate $\Lambda$}}                                                                                                                                   \\
\textbf{Estimates for reporting rate $\lambda$}                   & \multicolumn{1}{l}{1.0}          & \multicolumn{1}{l}{2.0}          & \multicolumn{1}{l}{3.0}          & \multicolumn{1}{l}{4.0}          & \multicolumn{1}{l}{5.0}          \\ \hline
Naive                          & 1.201(0.016)                     & 2.394(0.032)                     & 3.589(0.051)                     & 4.788(0.066)                     & 5.999(0.084)                     \\
MLE, correct                   & 2.041(0.188)                     & 2.018(0.091)                     & 2.005(0.057)                     & 2.005(0.045)                     & 2.013(0.035)                     \\
MLE, incorrect                 & 8.992(4.276)                     & 8.686(1.696)                     & 8.649(1.221)                     & 8.615(0.859)                     & 8.596(0.635)                     \\
Poisson regression, correct    & \multicolumn{1}{l}{2.033(0.118)} & \multicolumn{1}{l}{2.014(0.058)} & \multicolumn{1}{l}{2.003(0.036)} & \multicolumn{1}{l}{2.004(0.029)} & \multicolumn{1}{l}{2.011(0.022)} \\
Poisson regression, incorrect  & 8.826(2.356)                     & 8.613(1.008)                     & 8.599(0.755)                     & 8.581(0.539)                     & 8.572(0.402)
\end{tabular}%
}
\end{table}

Before applying our methods to real-world 311 data, we demonstrate the effectiveness of our methods via simulated data---in such simulations, we have full control over the data-generating process and thus can compare the estimation results with the true parameters. In particular, we use the simulator to illustrate (a) \Cref{prop:identi}, that attempting to recover reporting rates $\lambda_\theta$ from $N_\theta^\obs(T)$ is prone to bias; (b), that in \Cref{thm:stoppingtimes}, the conditions on $S$ and $E$ are essential, (c), that in a correctly-specified model setting, i.e., following the assumptions in \Cref{thm:stoppingtimes}, both the MLE in \Cref{eq:MLE} and the homogeneous Poisson regression recover the ground truth; and (d), that the regression approach is more data-efficient than the MLE, leading to tighter parameter estimates.

%To establish the robustness of our proposed methods, we first apply them on simulated data. The simulated data are generated through our proposed data-generating process, and validated by comparing with real-world data.

%The main advantage of using simulated data is that we have the `ground-truth' of this system. Due to the imbalance of reporting behavior, in real world data, the observed is only a biased sample of the truth. With

\paragraph{Simulator setup} We simulate a basic time-homogeneous system, as follows. We set the incident type to be a two-dimensional vector $\theta \in \bbR^2$. % = (\theta_1, \theta_2, \theta_3, \theta_4) \in \Theta_1\times\Theta_2\times\Theta_3 \times \Theta_4$.
Our simulator needs parameters for three processes: the incident birth process governed by homogeneous Poisson rate $\Lambda_\theta$, the reporting process governed by homogeneous Poisson rate $\lambda_\theta$, and a lifetime $T_i$ of each incident, generated in the following way: after each report, we sample two competing exponential clocks, one representing incident death and another the next report; the incident death rate can depend on the number of reports so far, reflecting, for example, that the agency prioritizes inspections for incidents with more reports. If the report happens before death, we increment the number of reports for the incident and repeat; otherwise, the incident has ``died'' and no more reports are logged.

Formally, for each type of incident, we specify a parameter $\mu_\theta$ that is independent of $\lambda_\theta$ as the ``death rate'' of incidents; conditional on there has been $m$ reports of an incident of type $\theta$, we generate two exponentially distributed random variables, $d^m_\theta\sim \text{Exponential}(\mu_\theta\times (\gamma_\theta)^m)$, where $\gamma_\theta$ is a scaling parameter by our choice, and $r^m_\theta \sim \text{Exponential}(\lambda_\theta)$; if $d^m_\theta \le r^m_\theta$, we consider the incident dead, and let $T_i=\sum_{j=0}^{m-1}r^j_\theta + d^m_\theta$, otherwise increment $m$ and repeat the process.
%
% \todo{Three or 4 dimensional? Above and below are inconsistent}
Both the incident and reporting rates are set as a function of the type covariates, as in \Cref{eqn:lambdaregressionassumption}:
\begin{align*}
\Lambda_\theta = \exp\left(\alpha_\text{incident} + \beta_\text{incident}^T\theta\right) && \lambda_\theta = \exp\left(\alpha_\text{report} + \beta_\text{report}^T\theta\right)
\end{align*}
where $\alpha_{incident}\in \mathbb{R}$, $\beta_{incident}\in \mathbb{R}^2$, $\alpha_{report}\in \mathbb{R}$ and $\beta_{report}\in \mathbb{R}^2$ are varied across simulations. For simplicity in comparing the various methods, here we report simulation results for the case in which there are only 5 distinct types, and the reporting rate for each type is $2$, i.e., $\alpha_{report}=\log(2)$ and $\beta_{report} = \textbf{0}$, but where incident rates $\Lambda_\theta$ vary by type.

Next, we need to set $\mu_\theta$ and $\gamma_\theta$, which governs the duration that the incident is alive. Given the above parameters, we do so in a manner that matches the distribution of the number of reports received per incident to the real data studied in the next section:  that 18.7\% of the received reports in the reporting period are duplicates of incidents already reported. For simplicity, we set these uniformly across all types. This calibration results in $\mu_\theta =  0.065$ and $\gamma_\theta = 100$ for all $\theta$.

% \zhi{just redid this and the sim, should check out now, conclusion is the same}

%The main purpose of our comparison is to determine the robustness of the various estimators to small sample sizes for a specific type of incidents. To this end, we fixed , and adjusted the incident rates for different types.

Finally, we fix the time span of our observation to be $T=300$ days, which by our parameter settings, is long enough for the system to converge to its long-run stationary distribution (in terms of the number of active incidents). Incidents and reports are generated according to the simulator setup. It is possible that some death times and consequently reports generated may be beyond the 300-day observation period. We discard any such reports.

The output of the simulator consists of all the incidents that occur and are reported at least once during this 300-day period. For all these incidents, the available data for our estimators is the time that each incident was reported the first time, the times and number of the subsequent reports, and the times of the incident death if they occur before the end of the time span. To specify the observation period of each incident, in the setting of \Cref{thm:stoppingtimes}, we let $S_i$ be the time of the first report of such incident, and $E_i$ be the incident death time, or the end of the time span if the death time is greater than it. This specification satisfies the assumption in \Cref{thm:stoppingtimes}. As a comparison, we add an incorrectly specified version of this observation period by letting $E_i$ be the time of the last report of this incident, in which case it no longer satisfies the conditions.

\paragraph{Simulator results} We compare five estimators: a ``naive'' estimator, that calculates $\frac{N_\theta^\obs(T)}{T}$, the ratio between the observed number of reports and the observation period; the MLE as derived in \Cref{eq:MLE} with correctly specified stopping times; the MLE with incorrectly specified observation ending times; a Poisson regression with correctly specified stopping times, and finally, a Poisson regression with incorrectly specified observation ending times. The regression methods were implemented using Scikit-learn \cite{scikit-learn}. We run each of the methods on the same simulated datasets and iterated 300 times. \Cref{tab:sim} summarizes our results: showing, for each distinct type with a differing incident rate, the estimates for $\lambda_\theta$ for each of the methods.
%; the poststratified MLE is by fitting a linear regression model with the MLE as dependent variable and the covariates as independent variable, and poststratify to get the estimate for each type

 These results indeed illustrate \Cref{prop:identi}, that attempting to recover reporting rates $\lambda_\theta$ from $N_\theta^\obs(T)$ is prone to bias: the naive method of counting the number of observed incidents conflates the incident rate with the reporting rate. Second, the stopping times assumption in \Cref{thm:stoppingtimes} is indeed crucial to a valid result, and in a correctly specified model setting, both the MLE in \Cref{eq:MLE} and the homogeneous Poisson regression recover the ground truth reporting rate $\lambda_\theta$, regardless of the incident rate. However, third, the regression approach is more data-efficient than the MLE, leading to tighter parameter estimates -- especially as the incident rate decreases, i.e., as the sample size in terms of the number of incidents decreases. This data efficiency is important in high dimensions, as in our real-world data application in the next section.

\section{Supplementary Empirical analysis}
\subsection{Two potential sources of bias}
\label{appendix:bias}

There are two potential sources of bias in the NYC data, due to ways in which the data differs from the model assumed in \Cref{sec:model}. Here, we evaluate these concerns and conclude that they likely do not substantially affect our measurements.

\paragraph{Repeat callers about the same incident.} One potential worry is that duplicate reports are a mirage: they are primarily generated by the same resident repeatedly calling about an incident until it is addressed. If that is the case, our method does not work: we rely on a Poisson rate assumption for the reporting behavior (that reporting behavior is memory-less and so that one report does not affect the likelihood of another for that incident), which is likely violated if the same person makes multiple reports about the same incident. In theory, such repeats should be minimal: NYC makes available a portal to check the status of past reports, so a reporter does not need to call again to remain up-to-date. In practice, however, our contacts at NYC DPR indicated that repeat calls occur.

To mitigate the effect of such repeat callers on our analysis, we obtain anonymized (hashed) caller information from NYC DPR for each report: if the caller provided it, their name, phone number, and/or email address. We then filter out the duplicate reports for each incident where either the phone number or emails match, or both the first and last names match. Our analyses are run on the resulting filtered dataset.
The above approach may not filter out all repeat callers: if callers choose not to leave their information, but call multiple times. Thus, we also run our main analyses on a filtered dataset where we additionally assume that any caller who did not leave their information is a repeat of a previous caller with no information. Our estimates are largely the same on this more conservative dataset, suggesting that repeat callers do not substantially bias our estimates. Results are given in \cref{sec:robustnesstables}.

\paragraph{Censored data: incidents that were not inspected}
The NYC DPR only marks duplicate reports corresponding to incidents that were inspected: for the service requests not connected to an inspected incident, we do not know which (if any) other reports also refer to the same incident. As our method relies on the rate of duplicate reports, we must discard service requests that were not inspected.

This censoring may limit the generalizability of our findings, from measuring the reporting rates of all incidents to measuring the reporting of incidents that tend to be inspected. This limitation to external validity may be acceptable: if the inspection decisions are correlated with incident importance (likely), then studying the heterogeneous reporting behavior for these incidents is a more important task than is studying that of minor incidents not deemed worth inspecting.

There is a second reason we believe that the censoring is relatively acceptable. In particular, we measure reporting rates as a function of incident type $\theta$, where the type includes characteristics such as report category and incident risk. If our models are correctly specified, and $\theta$ is rich enough to capture inspection decisions (there is no confounding), then this censoring does not affect our estimates.\footnote{We aim to model $\tilde{Pr}(Y | \theta)$, where $Y$ is reporting behavior. However, with just data on inspected incidents, we can only estimate $\tilde{Pr}(Y | \theta, \text{inspected})$. If $Y$ is independent of the inspection decision given $\theta$, then $\tilde{Pr}(Y | \theta, \text{inspected}) = \tilde{Pr}(Y | \theta)$. One potential source of bias is if, even conditional on $\theta$ (which includes the \textit{content} of the reports), NYC DPR is making decisions that strongly correlate with the \textit{number} of reports. Then, all of our rate estimates would be biased upwards, as we selectively observe data for incidents with many reports. However, as we primarily care about \textit{heterogeneous} reporting rates across types, such a bias matters to the extent that it heterogeneously affects different types of incidents or geographic locations. Furthermore, according to NYC DPR, the primary drivers of inspection decisions are the report characteristics, which are included in $\theta$. Nevertheless, an important direction for future work is directly addressing this censoring challenge.} While a seemingly strong assumption, we note that the $\theta$ we have available is the same data that the NYC DPR sees about a report through their portal when making an inspection decision; any confounding would have to come from another source.

Nevertheless, to the extent that the above (likely small) bias affects practice, it may be valuable for 311 systems to systematically tag duplicates for all reports and then apply our methods. (Relative to the missing data challenges that we centrally tackle in this work, i.e., incidents not reported and birth and death times, this duplicate censoring is cheaply addressable by city agencies). We note that the Chicago data \textit{does} mark duplicates even for open/uninspected incidents, and so this bias does not appear there.

\subsection{NYC data preprocessing}
\label{app:nycpreprocessing}

Before training models, we need to construct a dataset in which each row corresponds to an incident, and where we have the number of reports $\tilde M_i$ in an observation interval, the duration $\tilde T_i$ of that interval, and covariates $\theta$. We separated out an exploratory dataset of 8,000 unique incidents, on which we conducted covariate selection as detailed below.\footnote{While the exploratory data was used to filter variables for ultimate analysis, and to develop and fine-tune our models, we note that we did not hold out a separate test set at the outset of the project; it was not clear how to cluster assign reports to test and train before developing our empirical strategy, and there were initial (ultimately resolved) data errors on how reports were tagged to unique incidents. Thus, our overall approach was selected and developed using the NYC DPR data on which we ultimately report results, but not using the Chicago dataset. The exploratory dataset was ultimately composed of 4463 unique incidents, after filtering.}

We filter out the reports corresponding to the uninspected service requests (as we do not have duplicate information for these) and then use the provided incident label to group all service requests for the same incident. Then, we remove repeat caller reports, comparing each caller to previous (ordered by time) callers for the same incident. Next, we must construct a valid observation interval for each incident.

\paragraph{Constructing an observation interval $(S_i, E_i]$} As outlined in \Cref{thm:stoppingtimes}, we must be careful in how we choose an observation interval $(S_i, E_i]$ in which we count reports -- we need that the interval is inside the incident lifetime, i.e., we must end the interval before the incident is addressed, $E_i \leq t_i + T_i$. Both endpoints of the interval must also satisfy the conditions outlined in \Cref{thm:stoppingtimes}. As discussed above, the best choice to start the observation period $S_i$ is the time of the first report, but choosing the observation end is a design choice.

We make the following choice. Let $t_i^{\text{INSP}}$ be the inspection time of incident $i$, and, $t_i^{\text{WO}}$ be the time that a work order is placed for incident $i$, if applicable. Then, $E_i$ of each incident $i$ is:
\begin{equation}
	E_i = \min \left\{100\text{ days} + S_i,\  t_i^{\text{INSP}},\  t_i^{\text{WO}}\right\}. 	\label{eq:duration}
\end{equation}
The maximum duration of 100 is a design choice, for which we perform robustness checks (with 30 and 200 days); a maximum mitigates -- for incidents not inspected for a long time -- the inclusion of a time period in which an issue might have been resolved before an inspection, which would bias our estimates downward. \Cref{eq:duration} requires that inspection and work order times are stopping times; that they do not depend on the future trivially holds, and it is likely that they do not depend on the reporting {rate} except through the type $\theta$ and the sample path number of reports received up to that time.\footnote{We observe largely the same dashboard data that the inspector does when making inspection decisions, and so there is minimal unobserved confounding; the exception is that we do not process report free-form text, though in our conversations with NYC DPR these hold secondary importance after the structured fields.} (It is not a problem that incidents with more reports are inspected sooner).

\vspace{1em}

\noindent Appendix \Cref{fig:numreportshist} shows the histogram of the number of reports per incident during the observation interval; Appendix \Cref{fig:durationhist} shows the distribution of durations; and \Cref{tab:Summary} shows how the average duration differs by Borough and report category. The heterogeneity in duration length (due to the speed of being inspected or worked on) demonstrates the value of \Cref{thm:stoppingtimes}, which allows us to maximally utilize the data without introducing bias.  For example, compared to `Prune' incidents, `Hazard' incidents tend to have more reports on average and shorter duration: residents have a higher reporting rate for hazardous incidents, and these incidents tend to be addressed more quickly. Suppose we had to use a fixed duration $D$, instead of an inspection/work order dependent time. If $D$ is large (e.g., $D \approx 15$ days), then we bias our estimates downward, as we're including time periods after an incident has already been addressed -- and the bias heterogeneously affects incident types, since some incident types are typically addressed more quickly than are other types. On the other hand, a much shorter duration would substantially limit the data. Finally, Appendix \Cref{fig:daysfirstcomplaint} shows the average number of days after the first report that the $\ell$th duplicate report was submitted for an incident, for incidents with at least $k \geq \ell$ reports. The plots are largely linear (i.e., the average delay between the first and second report is the same as that between the third and fourth report), consistent with reporting rates being approximately homogeneous Poisson within the interval.

\paragraph{Covariate selection and processing} Next, we select the covariates that compose type $\theta$. The data given to us by NYC DPR includes a set of \textit{report} covariates (e.g., report \textit{Category}),\footnote{Occaisonally, different reports about the same incident disagree on the report covariates. We select the first report characteristics in those cases.} \textit{inspection} results (e.g., condition of the tree at inspection time), and \textit{tree} characteristics (e.g., the diameter of the tree at breast height, tree species). We augment this data with \textit{socioeconomic characteristics} as follows. Most of the reports in our data contain latitude-longitude coordinates for each inspection (and thus incident), using which we identify which of the over 2000 census tracts in New York City the incident is in, through an FCC API.\footnote{FCC Area API, https://geo.fcc.gov/api/census/} We then join this information with 2020 Census data obtained from the IPUMS NHGIS \cite{manson2020ipums}, which include socioeconomic characteristics such as race/ethnicity, education, income, and population density for each census tract.

Next, we perform covariate selection using the exploratory dataset. We remove report and inspection variables that are highly collinear, have low variance, or with a high number of missing values. Conversations with NYC DPR also played a role in the selection. Finally, we log transform several variables and standardize all data. Appendix \Cref{tab:desc} contains the covariates we use.

\smallskip

Starting with the dataset discussed above, we filter out the incidents for which any of the covariates are missing and those with short-logged reporting periods (Duration less than 0.1 days). We are left with a dataset of \num{81638} incidents on which we conduct our main analyses; for each incident $i$, we have the duration of observance $e_i -s_i$, the number of total reports $\tilde M_i$, and all the various geographic and demographic covariates associated with it.
\subsection{Additional information}

% \todo{add 1 sentence referencing each table/figure and what it includes}

In this section, we provide some additional information about our dataset and results. \Cref{fig:numreportshist} shows the histogram of the number of reports per incident during the observation interval; \Cref{fig:durationhist} shows the distribution of durations. \Cref{tab:desc} provides a description of the covariates selected; \Cref{fig:reginccoef} shows the relationship between the number of unique incidents observed versus the census tract fixed effect; \Cref{fig:posteriors} shows the posterior distribution of the number of reports as estimated by the basic Poisson regression model, and the zero-inflated Poisson regression model, with reference to the observed distribution; \Cref{tab:coeffull} and \Cref{tab:coeffullborough} lists the full information of the coefficients for census tract socioeconomic covariates as estimated alone in a regression alongside the incident-specific covariates and the borough fixed effects. \Cref{tab:nycmultidemo} lists the information of the coefficients for a subset of census tract socioeconomic covariates as estimated together in a regression alongside the incident-specific covariates.

\begin{figure}[tb]
	\centering
	%	\begin{subfigure}{.5\textwidth}
	% 		\centering
	\subfloat[][Number of reports per incident.]{
		\includegraphics[width=.45\textwidth]{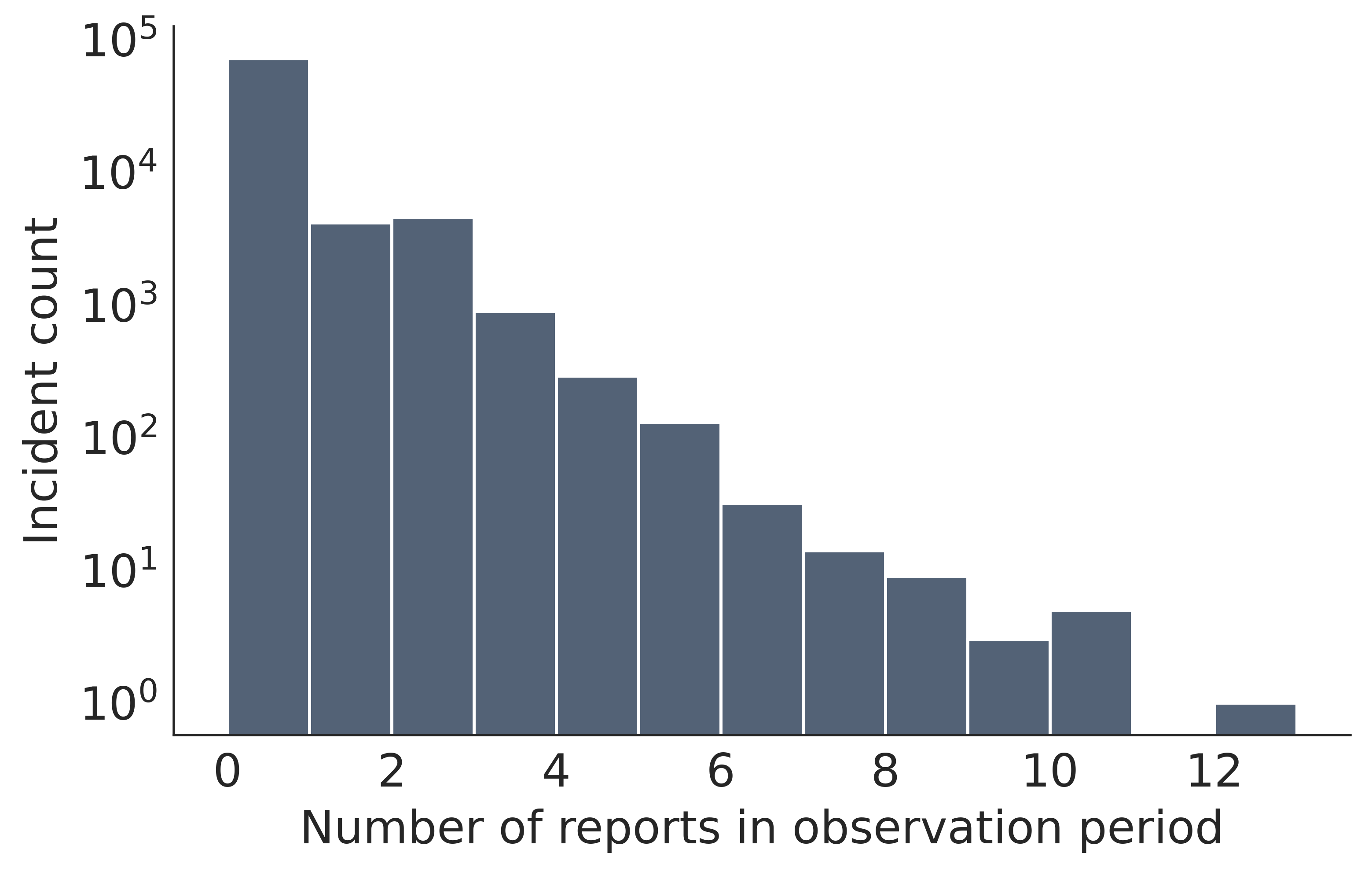}
		% 		\caption{Census tract fixed effects}
		\label{fig:numreportshist}
		% 		\end{subfigure}
	}
	\hfill
	\subfloat[][Length of observation period.]{
		\includegraphics[width=.45\textwidth]{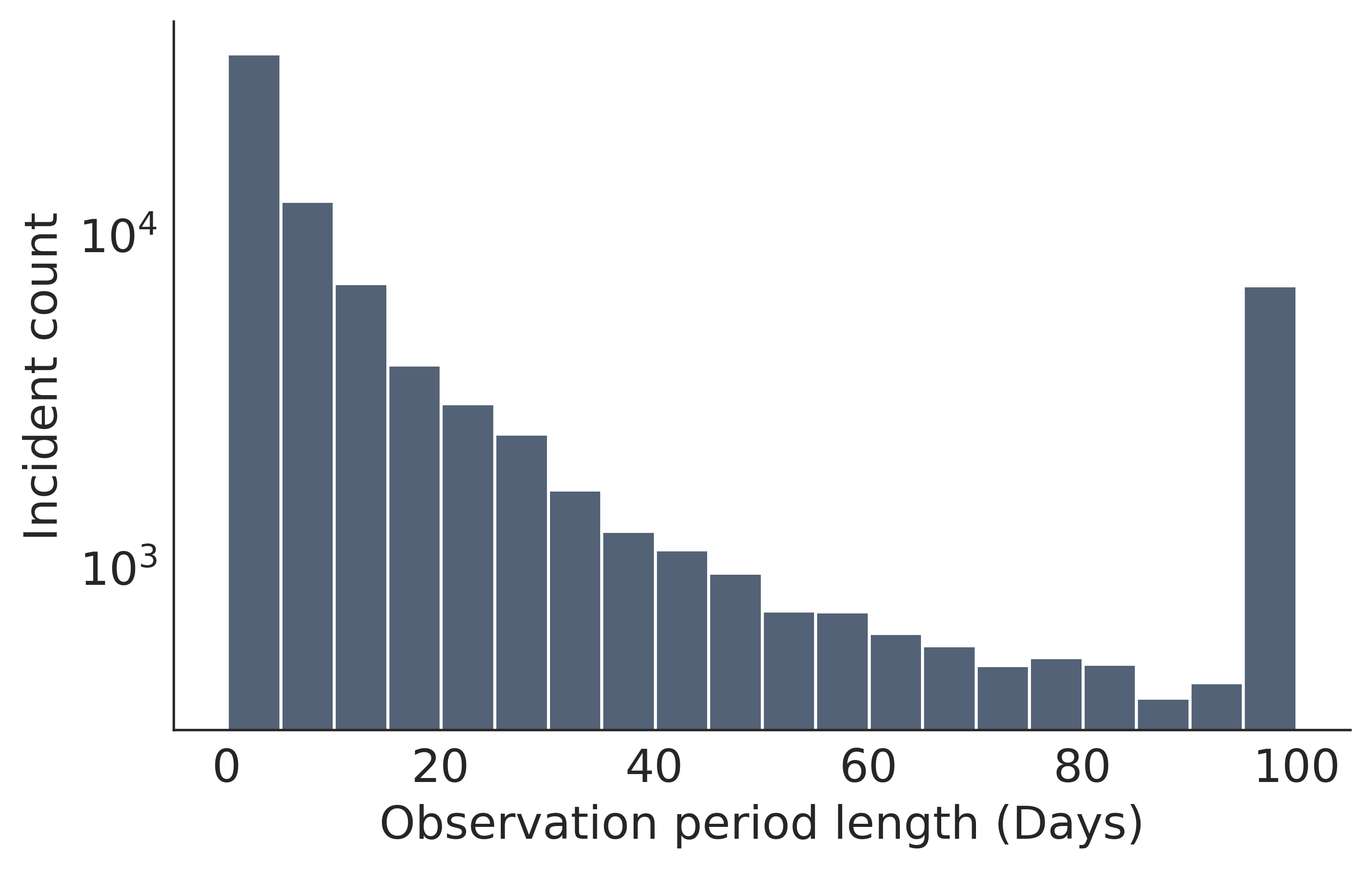}
		% 		\caption{Census tract fixed effects}
		\label{fig:durationhist}
		% 		\end{subfigure}
	}
	\caption{Distribution of number of reports and length of observation for each unique incident in the aggregated dataset. For most incidents, there are no reports after the first report (at least not in the observation period). There is a peak at 100 days for the observation period, due to our configuration in \Cref{eq:duration}, where we truncate longer periods to 100 days.}
\end{figure}

\begin{table}[tbh]
\caption{Description of covariates in the aggregated dataset. }
\label{tab:desc}
\centering
\resizebox{\textwidth}{!}{%
\begin{tabular}{l|l}
\multicolumn{1}{c|}{\textbf{Covariate}} & \multicolumn{1}{c}{\textbf{Description}}                                            \\ \hline
\textbf{Incident Global ID}             & An identifier unique to each incident.                                              \\ \hline
\textbf{Duration}                       & The observation duration as defined in \Cref{eq:duration}                                        \\ \hline
\textbf{Number Reports} &
  \begin{tabular}[c]{@{}l@{}}Number of reports on this incident in the observation duration, after filtering \\ out repeat callers\end{tabular} \\ \hline
  % \textbf{Created Month}                       & The month that the first report about each incident was created.                                        \\ \hline
\textbf{INSPCondition} &
  \begin{tabular}[c]{@{}l@{}l}The inspection outcome regarding the condition of the tree, indicates \\ whether the tree is dead, in good to excellent conditions, or in fair conditions.\end{tabular} \\ \hline
\textbf{INSP\_RiskAssessment}           & 
\begin{tabular}[c]{@{}l@{}l@{}}The inspection outcome regarding how dangerous the reported incident is, \\as determined by inspector. Ranges from 0 to 12, with 11 and 12 \\being category ``A" incident and prioritized for work orders.\end{tabular} \\ \hline
\textbf{\begin{tabular}[c]{@{}l@{}}Tree Diameter at \\ Breast Height (TDBH)\end{tabular}} &
  Main characteristic of the tree describing how large the tree trunk is. Measured in inches. \\ \hline
\textbf{Borough}                        & Indicating which borough in NYC this incident is located.                           \\ \hline
\textbf{Category}                       & The incident category as reported.                                                  \\ \hline
\textbf{Median Age}                     & Median age in the census tract.                                                     \\ \hline
\textbf{Fraction Hispanic}              & Fraction of residents that identify as Hispanic in the census tract.                \\ \hline
\textbf{Fraction white}                 & Fraction of residents that identify as white in the census tract.                   \\ \hline
\textbf{Fraction Black}                 & Fraction of residents that identify as Black in the census tract.                   \\ \hline
\textbf{Fraction no high school degree}              & Fraction of residents that have not graduated from high school in the census tract. \\ \hline
\textbf{Fraction college degree}          & Fraction of residents that have graduated from college in the census tract.         \\ \hline
\textbf{Fraction poverty}            & Fraction of residents that are identified to be in poverty in the census tract.     \\ \hline
\textbf{Fraction renter}                & Fraction of residents that rent their current residence in the census tract.        \\ \hline
\textbf{Fraction family}           & Fraction of family household in the census tract.         \\ \hline
\textbf{Median household value}     & Median value of household in the census tract. \\ \hline
\textbf{Income per capita}                     & Income per capita of residents in the census tract.                                    \\ \hline
\textbf{Density}                        & Population density in the census tract.
\end{tabular}%
}

\end{table}

 \begin{figure}[tb]
	\centering
	\subfloat[][Reported incidents vs \\number of trees]{
		\includegraphics[width=.4\textwidth]{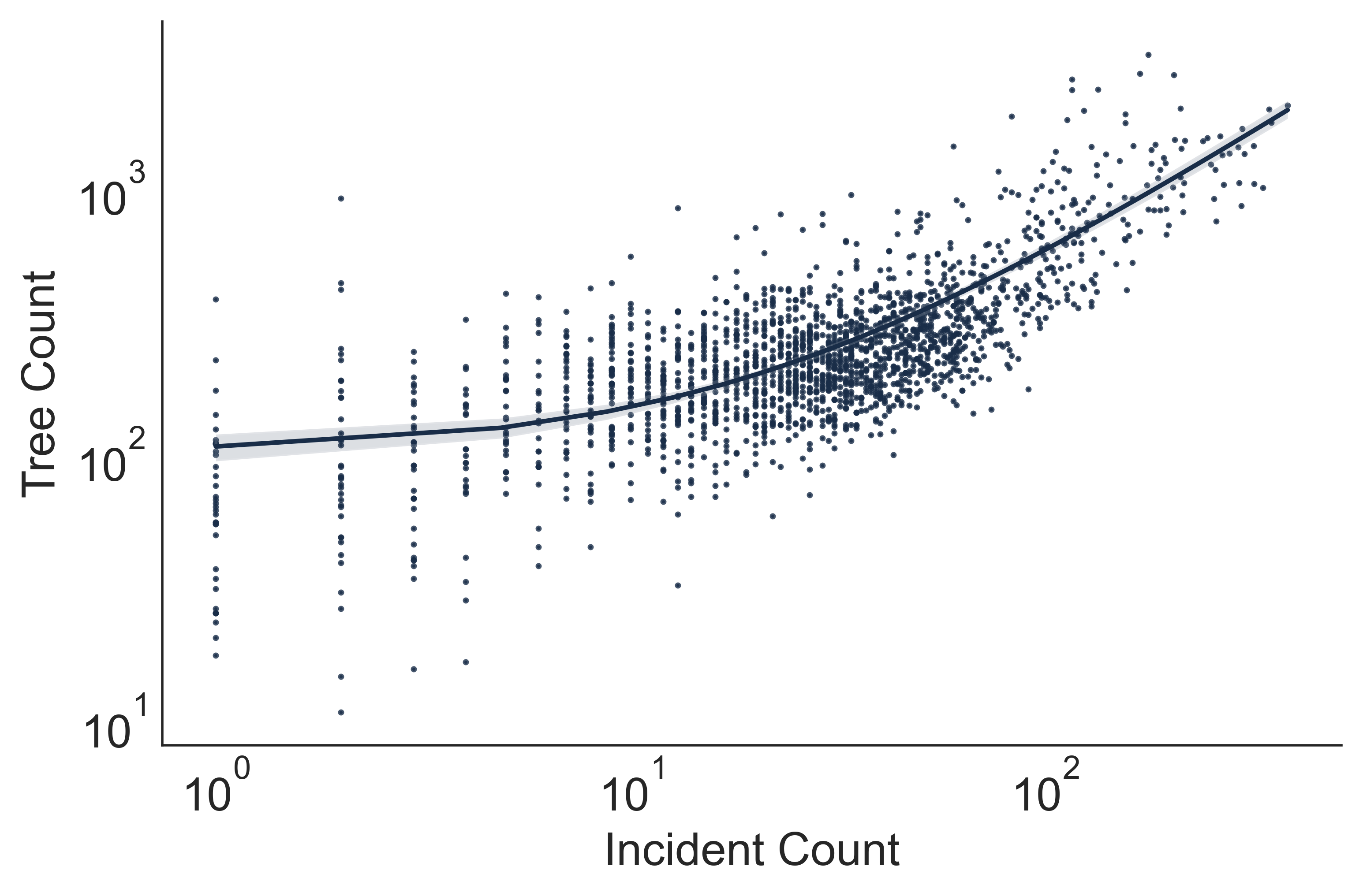}
		\label{fig:treesincidents}
	}
	\subfloat[][Incidents/Tree vs \\Census Tract Coefficient]{
	\includegraphics[width=.4\textwidth]{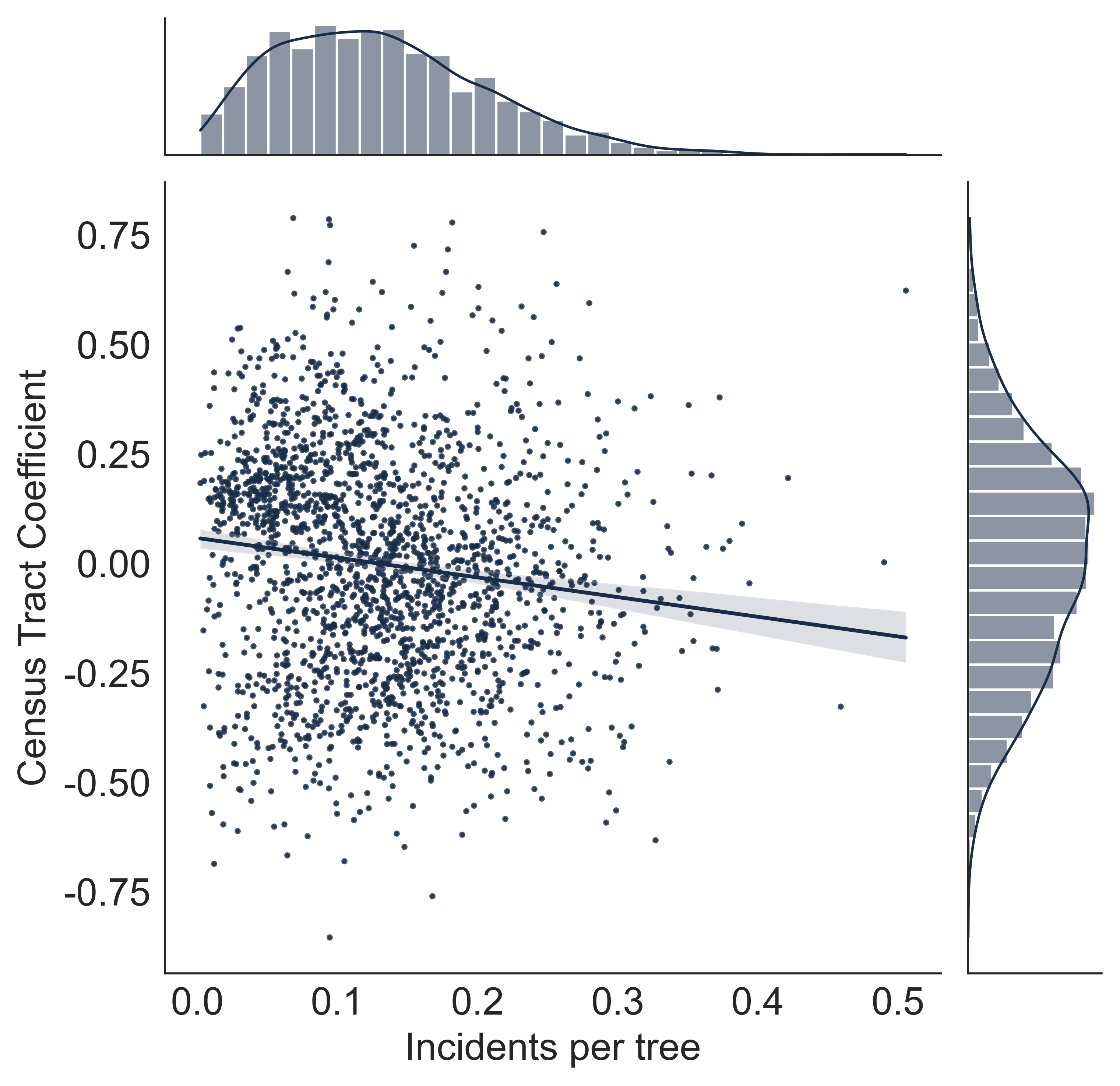}
	\label{fig:reginccoef}
}
\hfill
	\subfloat[][Reported \textbf{Hazard} incidents vs \\number of trees]{
		\includegraphics[width=.4\textwidth]{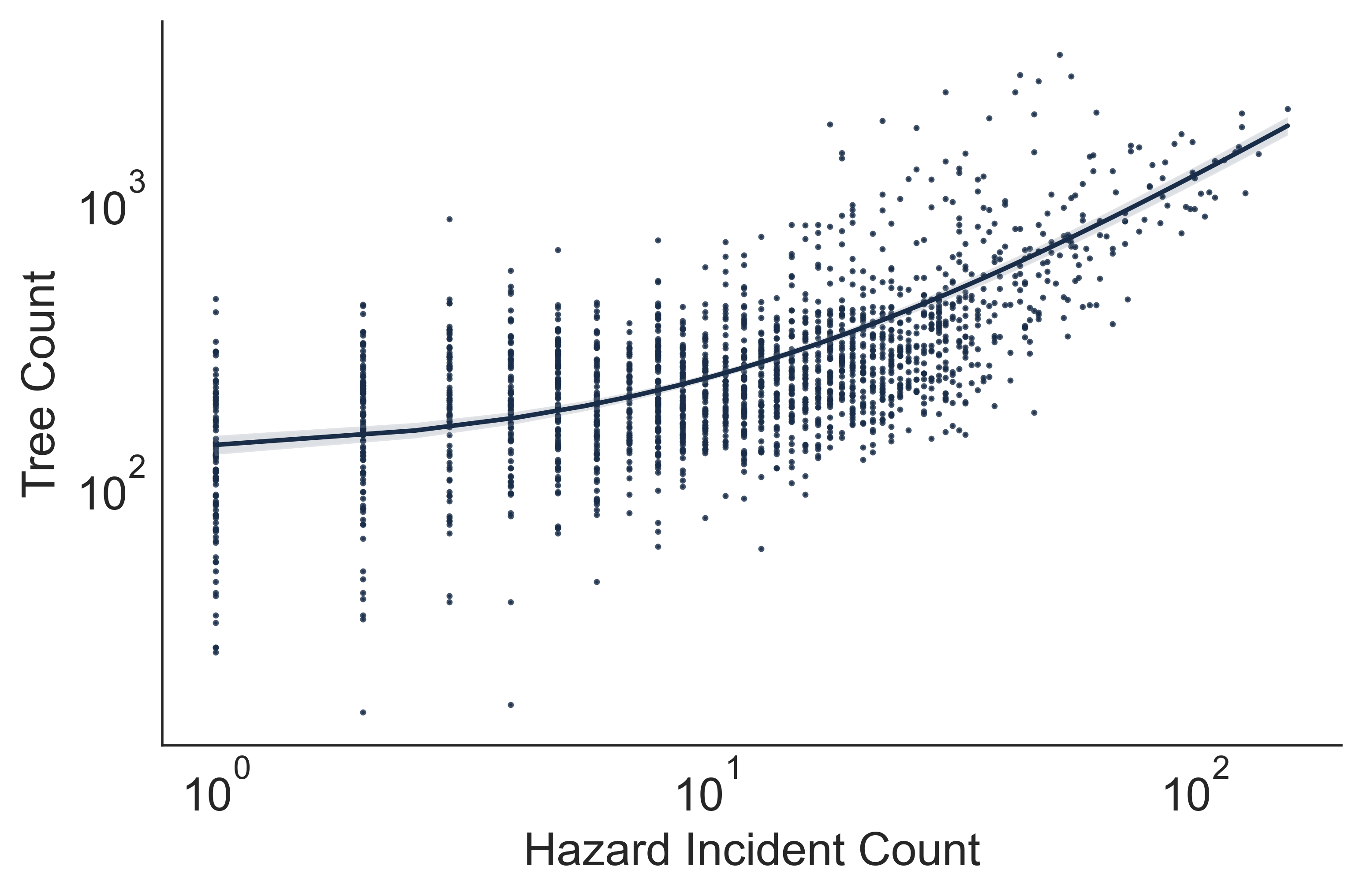}
		\label{fig:treesincidentshazard}
	}
	\subfloat[][\textbf{Hazard} Incidents/Tree vs \\Census Tract Coefficient]{
	\includegraphics[width=.4\textwidth]{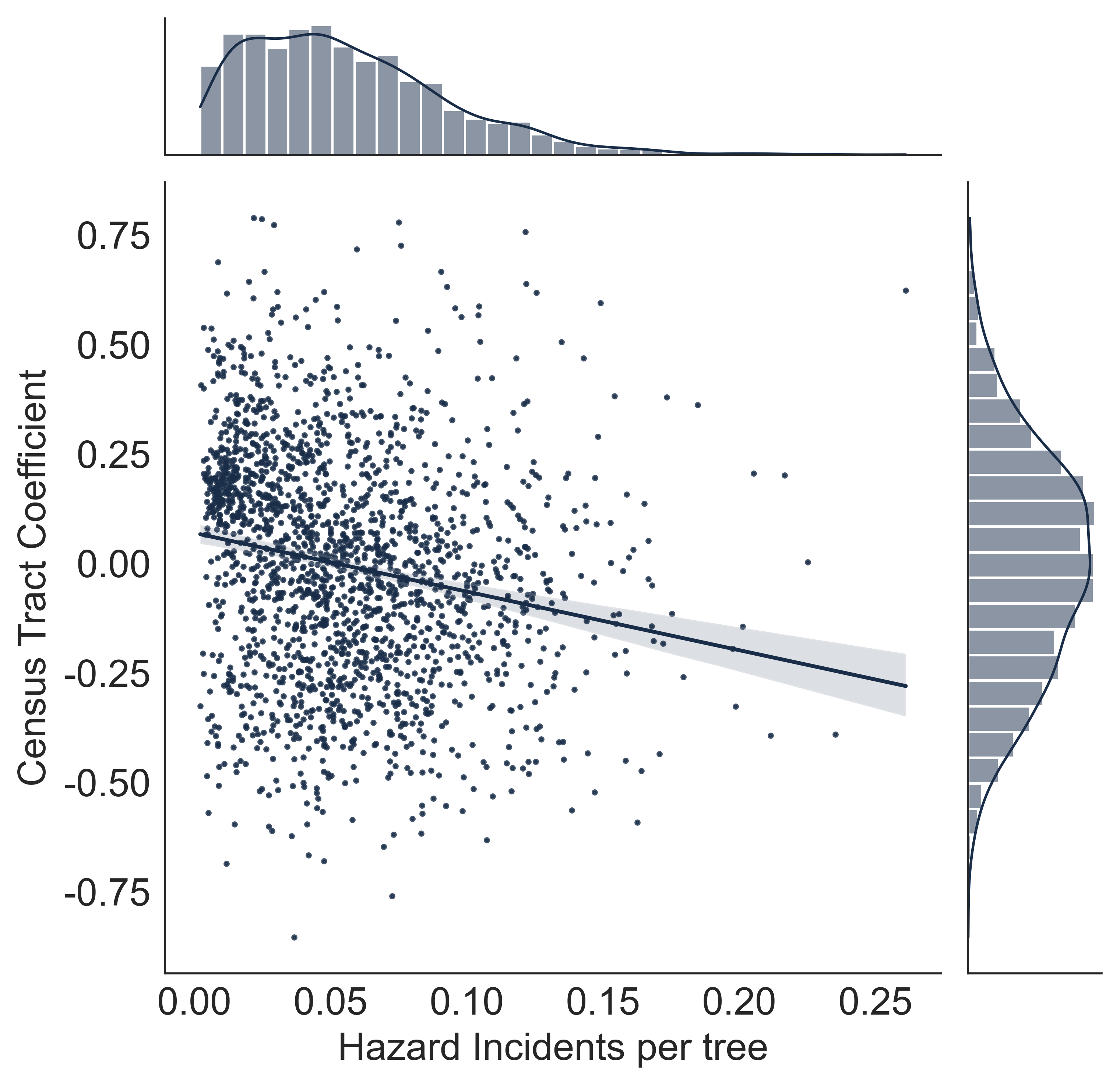}
	\label{fig:reginccoefhazard}
}
\hfill

	\caption{(a) For each census tract, the number of trees according to the 2015 New York Street Tree Census versus the number of reported incidents in that census tract. As expected, more trees result in more reported incidents. (b) The relationship between reported incidents per tree and the census tract fixed effect. (c) and (d) reproduce these plots while restricting to only \textbf{Hazard} incidents. The value, ``incidents per tree'' is an attempt to normalize the number of incidents we observe with the number we expect to observe -- and thus the ratio is a measure of \textit{reporting} rates as is done in prior work. We observe a slightly negative relationship between this measure and the one we develop, the census tract coefficient in the Poisson regression. While it is possible that one can construct better proxies for how many incidents we expect to observe than the raw counts of trees, the relationship suggests that our method's results can differ substantially from those of prior work. We prefer our measure, as it automatically controls for `legitimate' incident-level characteristics (such as risk) that may correlate with geography but are not captured with the number of trees---without needing to construct an estimate for the number of expected incidents for each such type. In practice, there can be many types of incidents; see, e.g., \Cref{tab:chicagobasic}. The Street Tree data is available here: \url{https://data.cityofnewyork.us/Environment/2015-Street-Tree-Census-Tree-Data/pi5s-9p35}.}

\end{figure}

 \begin{figure}[tb]
	\centering
%	\subfloat[][Reported incidents vs trees]{
		\includegraphics[width=.45\textwidth]{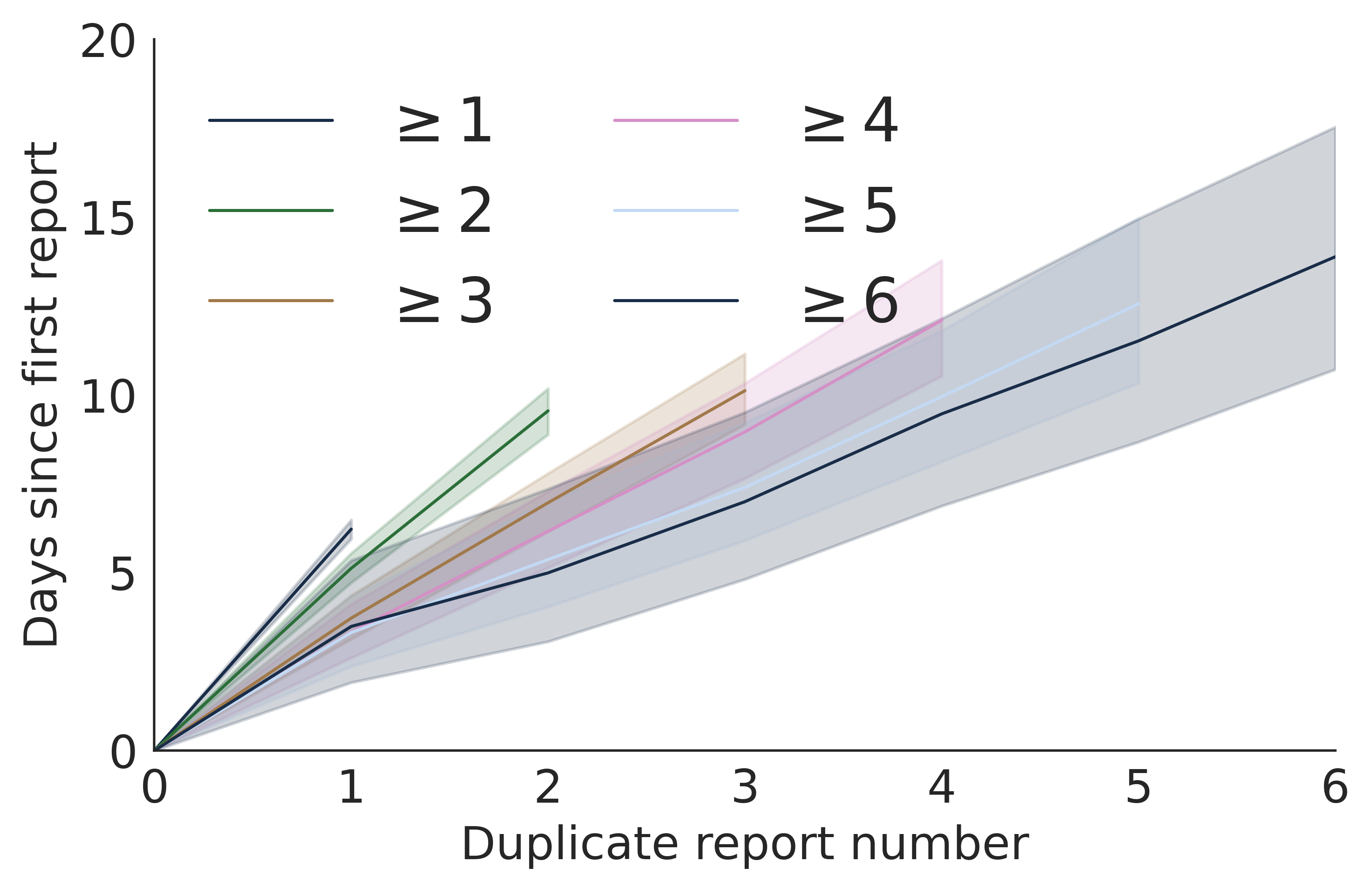}
%	}
%	\hfill
%	\subfloat[][Incidents/Tree vs Census Tract Coefficient]{
%		\includegraphics[width=.45\textwidth]{plots/incidents_per_tree_vs_coefficient}
%		\label{fig:reginccoef}
%	}
	\caption{Conditional on having at least $k$ duplicate reports, on average how many days after the $0th$ report was the $\ell$th report, where $\ell \leq k$? The linear nature of each plot is consistent with a homogeneous Poisson process within the given period. That incidents with more reports also receive reports \textit{faster} is consistent with those incidents being more severe in nature and having a higher reporting rate $\lambda_\theta$. }

			\label{fig:daysfirstcomplaint}
\end{figure}

 \begin{figure}[tb]
	\centering
	% \subfloat{
	% 	\includegraphics[width=.4\textwidth]{plots/basic_histograms.png}
	% 	\label{fig:basicposterior}
	% 	% 		\end{subfigure}
	% }
 % \subfloat{
	% 	\includegraphics[width=.4\textwidth]{plots/basic_histograms_nat.png}
	% 	\label{fig:basicposterior_nat}
	% }
	% \hfill
	% \subfloat[][Zero inflated Poisson regression, log scale]{
	% 			\includegraphics[width=.4\textwidth]{plots/basic_zeroinflated_histograms.png}
	% 	\label{fig:zifposterior}
	% }
 % \subfloat[][Zero inflated Poisson regression]{
	% 			\includegraphics[width=.4\textwidth]{plots/basic_zeroinflated_histograms_nat.png}
	% 	\label{fig:zifposterior_nat}
	% }
 % \hfill
 % \subfloat[][Category-dependent Zero inflated Poisson regression, log scale]{
	% 			\includegraphics[width=.4\textwidth]{plots/basic_zeroinflated_bycategory_histograms.png}
	% 	\label{fig:zifposterior_bc}
	% }
 % \subfloat[][Category-dependent Zero inflated Poisson regression]{
	% 			\includegraphics[width=.4\textwidth]{plots/basic_zeroinflated_bycategory_histograms_nat.png}
	% 	\label{fig:zifposterior_bcnat}
	% }
 % \hfill
 \begin{subfigure}[b]{\textwidth}
    \centering
    \includegraphics[width=.4\textwidth]{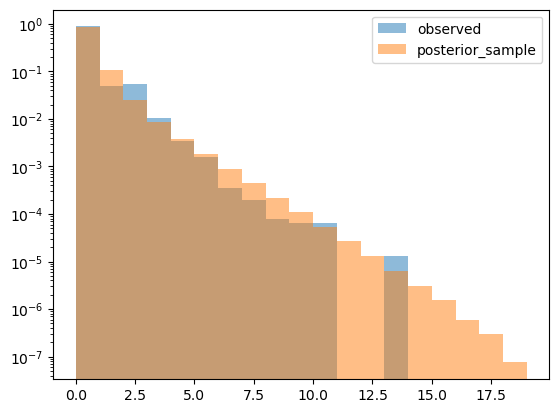}
    \includegraphics[width=.4\textwidth]{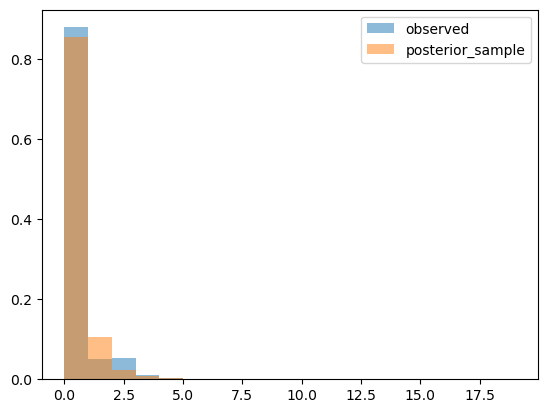}
    \hfill
    \caption{Standard Poisson regression}
    \label{fig:basicposterior}
 \end{subfigure}
  \begin{subfigure}[b]{\textwidth}
    \centering
    \includegraphics[width=.4\textwidth]{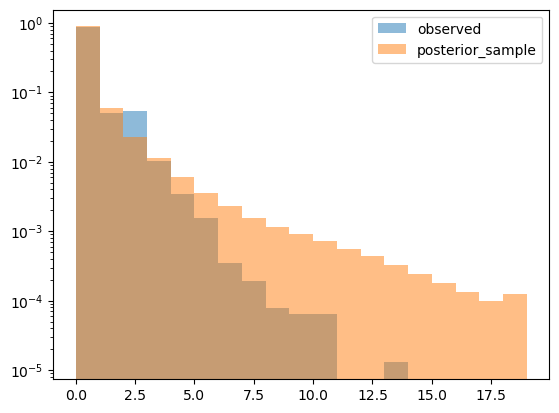}
    \includegraphics[width=.4\textwidth]{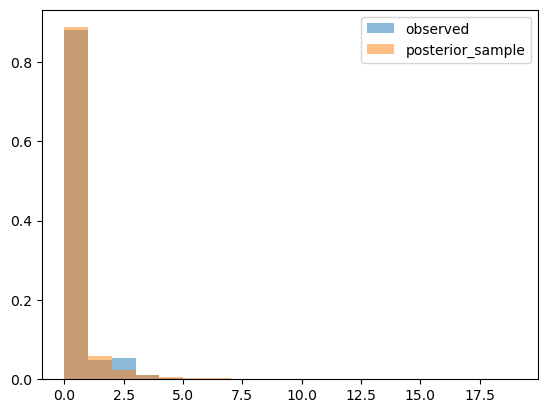}
    \hfill
    \caption{Zero inflated Poisson regression}
    \label{fig:zifposterior}
 \end{subfigure}
  \begin{subfigure}[b]{\textwidth}
    \centering
    \includegraphics[width=.4\textwidth]{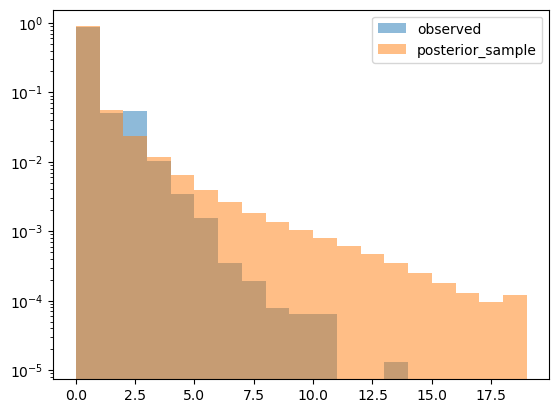}
    \includegraphics[width=.4\textwidth]{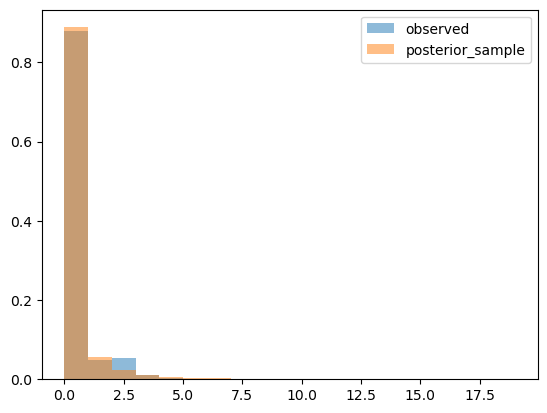}
    \hfill
    \caption{Zero inflated Poisson regression, with Category-dependent zero inflation coefficient}
    \label{fig:zifposterior_bycat}
 \end{subfigure}
	\caption{Comparison among posterior distributions sampled from the Standard Poisson regression model and two Zero-inflated variants. Left-hand plots are in log space for the $y$ axis, and right-hand plots are in regular probability space, to highlight the different parts of the distribution. Note that the Zero-inflated model has a better fit to the observed distribution since most of the density of this distribution is on reports 1, 2, and 3. Using Category-dependent zero inflation coefficients further improves the fit, but the gain is not substantial.}
	\label{fig:posteriors}

\end{figure}

\begin{table}
\centering
\caption{Census Tract socioeconomic coefficients in NYC, estimated alone in a regression alongside the incident-specific covariates. Full table corresponding to \Cref{tab:censuscoefficients}. }\label{tab:coeffull}
\begin{tabular}{lrrrrr}
\toprule
                               &   Mean &  StdDev &   2.5\% &  97.5\% \\
\midrule
                    Median age & -0.033 &   0.008 & -0.051 & -0.018 \\
             Fraction Hispanic &  0.030 &   0.009 &  0.011 &  0.046 \\
                Fraction white &  0.057 &   0.008 &  0.040 &  0.072 \\
                Fraction Black & -0.039 &   0.009 & -0.059 & -0.024 \\
Fraction no high school degree & -0.031 &   0.008 & -0.049 & -0.016 \\
       Fraction college degree &  0.043 &   0.009 &  0.024 &  0.059 \\
              Fraction poverty & -0.010 &   0.009 & -0.027 &  0.007 \\
               Fraction renter &  0.054 &   0.009 &  0.035 &  0.070 \\
               Fraction family & -0.081 &   0.009 & -0.099 & -0.065 \\
       Log(Median house value) &  0.065 &   0.009 &  0.045 &  0.081 \\
        Log(Income per capita) &  0.048 &   0.009 &  0.031 &  0.068 \\
                  Log(Density) &  0.077 &   0.010 &  0.056 &  0.099 \\
\bottomrule
\end{tabular}
\end{table}

\begin{table}
\centering
\caption{Census Tract socioeconomic coefficients in NYC, estimated alone in a regression alongside the incident-specific covariates and the \textbf{borough fixed effects}.}
\label{tab:coeffullborough}
\begin{tabular}{lrrrrr}
\toprule
                              &   Mean &  StdDev &   2.5\% &  97.5\% \\
\midrule
                    Median age & -0.009 &   0.008 & -0.025 &  0.005 \\
             Fraction Hispanic &  0.053 &   0.009 &  0.035 &  0.069 \\
                Fraction white &  0.068 &   0.008 &  0.052 &  0.083 \\
                Fraction Black & -0.052 &   0.009 & -0.070 & -0.037 \\
Fraction no high school degree & -0.039 &   0.008 & -0.056 & -0.024 \\
       Fraction college degree &  0.035 &   0.008 &  0.019 &  0.049 \\
              Fraction poverty & -0.020 &   0.008 & -0.036 & -0.006 \\
               Fraction renter &  0.022 &   0.008 &  0.004 &  0.038 \\
               Fraction family & -0.075 &   0.009 & -0.092 & -0.060 \\
       Log(Median house value) &  0.023 &   0.008 &  0.006 &  0.038 \\
        Log(Income per capita) &  0.056 &   0.008 &  0.038 &  0.072 \\
                  Log(Density) &  0.032 &   0.009 &  0.015 &  0.048 \\
\bottomrule
\end{tabular}
\end{table}

\begin{table}
\centering
\caption{Census Tract socioeconomic coefficients in NYC, estimated together in a regression alongside the incident-specific covariates. Covariates with high levels of collinearity are dropped to maintain interpretability.}
\label{tab:nycmultidemonodensity}
\begin{tabular}{lrrrrr}
\toprule
                        &   Mean &  StdDev &   2.5\% &  97.5\% \\
\midrule
             Median age & -0.014 &   0.009 & -0.034 & -0.003 \\
         Fraction white &  0.058 &   0.009 &  0.037 &  0.076 \\
Fraction college degree & -0.047 &   0.013 & -0.073 &  -0.023 \\
        Fraction renter &  0.042 &   0.010 &  0.021 &  0.060 \\
 Log(Income per capita) &  0.086 &   0.014 &  0.057 &  0.110 \\
\bottomrule
\end{tabular}
\end{table}

\begin{table}
\centering
\caption{Census Tract socioeconomic coefficients in NYC, estimated together in a regression alongside the incident-specific covariates. Compared with \Cref{tab:nycmultidemonodensity}, we additionally include for (log) population density. With the exception of the fraction of residents with college degrees which becomes insignificantly associated with the reporting rate, most socioeconomic variables are still significantly associated with the reporting rate in the same direction.}
\label{tab:nycmultidemo}
\begin{tabular}{lrrrrr}
\toprule
                        &   Mean &  StdDev &   2.5\% &  97.5\% \\
\midrule
             Median age & -0.021 &   0.009 & -0.038 & -0.003 \\
         Fraction white &  0.048 &   0.010 &  0.028 &  0.066 \\
Fraction college degree & -0.011 &   0.013 & -0.037 &  0.013 \\
        Fraction renter &  0.038 &   0.012 &  0.015 &  0.061 \\
 Log(Income per capita) &  0.074 &   0.015 &  0.041 &  0.103 \\
           Log(Density) &  0.087 &   0.012 &  0.063 &  0.108 \\
\bottomrule
\end{tabular}
\end{table}

\FloatBarrier
\subsection{Robustness checks -- \texorpdfstring{\Cref{tab:basicinfrefcoef}}{Base coefficients} with other specifications}
\label{sec:robustnesstables}
In this section, we present evidence for the robustness of our results in \Cref{tab:basicinfrefcoef} by reproducing its results with different configurations in \Cref{tab:standardpoisson} through \Cref{tab:zifremoveallrepeat}. Additionally, we show results when the risk assessment scores are binned according to their priority level in \Cref{tab:binnedrisk}.
\begin{table}
\centering
\caption{Regression coefficients for Standard Poisson regression with incident-level covariates and Borough fixed effects for Max Duration 100 days, Default repeat caller removal.}
\label{tab:standardpoisson}
\begin{tabular}{lrrrrrr}
\toprule
{} &   Mean &  StdDev &   2.5\% &  97.5\% &  R\_hat \\
\midrule
Intercept                           & -4.577 &   0.023 & -4.629 & -4.539 &    1.0 \\
INSPCondition[T.Dead]                & -0.334 &   0.029 & -0.396 & -0.285 &    1.0 \\
INSPCondition[T.Excellent\_Good]      & -0.428 &   0.023 & -0.472 & -0.382 &    1.0 \\
INSPCondition[T.Fair]                & -0.296 &   0.022 & -0.344 & -0.254 &    1.0 \\
INSP\_RiskAssessment                 &  0.286 &   0.010 &  0.265 &  0.305 &    1.0 \\
Log(Tree Diameter at Breast Height) & -0.020 &   0.009 & -0.039 & -0.003 &    1.0 \\
Borough[Bronx]                      &  0.103 &   0.022 &  0.061 &  0.144 &    1.0 \\
Borough[Brooklyn]                   & -0.148 &   0.016 & -0.185 & -0.117 &    1.0 \\
Borough[Manhattan]                  & -0.081 &   0.038 & -0.163 & -0.016 &    1.0 \\
Borough[Queens]                     & -0.112 &   0.015 & -0.142 & -0.080 &    1.0 \\
Borough[Staten Island]              &  0.237 &   0.026 &  0.183 &  0.284 &    1.0 \\
Category[Hazard]                    &  1.448 &   0.015 &  1.416 &  1.473 &    1.0 \\
Category[Illegal Tree Damage]       &  0.009 &   0.028 & -0.045 &  0.059 &    1.0 \\
Category[Prune]                     & -0.083 &   0.024 & -0.136 & -0.043 &    1.0 \\
Category[Remove Tree]               &  0.001 &   0.021 & -0.042 &  0.043 &    1.0 \\
Category[Root/Sewer/Sidewalk]       & -1.375 &   0.029 & -1.432 & -1.325 &    1.0 \\
\bottomrule
\end{tabular}
\end{table}

\begin{table}
\centering
\caption{Regression coefficients for Zero-inflated Poisson regression with incident-level covariates and Borough fixed effects for Max Duration 30 days, Default repeat caller removal.}
\begin{tabular}{lrrrrrr}
\toprule
                                {} &   Mean &  StdDev &   2.5\% &  97.5\% &  R\_hat \\
\midrule
Intercept                           & -2.679 &   0.031 & -2.738 & -2.617 &    1.0 \\
Zero Inflation fraction                 &  0.728 &   0.003 &  0.722 &  0.735 &    1.0 \\
INSPCondition[T.Dead]                & -0.226 &   0.042 & -0.312 & -0.144 &    1.0 \\
INSPCondition[T.Excellent\_Good]      & -0.368 &   0.029 & -0.422 & -0.311 &    1.0 \\
INSPCondition[T.Fair]                & -0.178 &   0.027 & -0.230 & -0.125 &    1.0 \\
INSP\_RiskAssessment                 &  0.204 &   0.012 &  0.178 &  0.227 &    1.0 \\
Log(Tree Diameter at Breast Height) & -0.009 &   0.009 & -0.026 &  0.008 &    1.0 \\

Borough[Bronx]                      & -0.041 &   0.027 & -0.099 &  0.012 &    1.0 \\
Borough[Brooklyn]                   & -0.188 &   0.019 & -0.226 & -0.155 &    1.0 \\
Borough[Manhattan]                  &  0.297 &   0.051 &  0.191 &  0.388 &    1.0 \\
Borough[Queens]                     & -0.240 &   0.020 & -0.279 & -0.203 &    1.0 \\
Borough[Staten Island]              &  0.172 &   0.035 &  0.097 &  0.238 &    1.0 \\
Category[Hazard]                    &  1.336 &   0.019 &  1.295 &  1.371 &    1.0 \\
Category[Illegal Tree Damage]       &  0.257 &   0.039 &  0.173 &  0.323 &    1.0 \\
Category[Prune]                     & -0.113 &   0.034 & -0.186 & -0.050 &    1.0 \\
Category[Remove Tree]               & -0.030 &   0.025 & -0.084 &  0.016 &    1.0 \\
Category[Root/Sewer/Sidewalk]       & -1.449 &   0.040 & -1.532 & -1.378 &    1.0 \\
\bottomrule
\end{tabular}
\end{table}

\begin{table}
\centering
\caption{Regression coefficients for Zero-inflated Poisson regression with incident-level covariates and Borough fixed effects for Max Duration 200 days, Default repeat caller removal.}
\begin{tabular}{lrrrrrr}
\toprule
                                {} &   Mean &  StdDev &   2.5\% &  97.5\% &  R\_hat \\
\midrule
Intercept                           & -3.495 &   0.029 & -3.548 & -3.440 &    1.0 \\
Zero Inflation fraction                 &  0.626 &   0.004 &  0.619 &  0.634 &    1.0 \\
INSPCondition[T.Dead]                & -0.462 &   0.034 & -0.536 & -0.401 &    1.0 \\
INSPCondition[T.Excellent\_Good]      & -0.205 &   0.027 & -0.254 & -0.153 &    1.0 \\
INSPCondition[T.Fair]                & -0.123 &   0.026 & -0.179 & -0.077 &    1.0 \\
INSP\_RiskAssessment                 &  0.238 &   0.011 &  0.216 &  0.258 &    1.0 \\
Log(Tree Diameter at Breast Height) & -0.042 &   0.008 & -0.059 & -0.025 &    1.0 \\
Borough[Bronx]                      & -0.139 &   0.027 & -0.191 & -0.093 &    1.0 \\
Borough[Brooklyn]                   & -0.393 &   0.019 & -0.433 & -0.357 &    1.0 \\
Borough[Manhattan]                  &  0.512 &   0.050 &  0.405 &  0.596 &    1.0 \\
Borough[Queens]                     & -0.244 &   0.019 & -0.278 & -0.208 &    1.0 \\
Borough[Staten Island]              &  0.264 &   0.032 &  0.196 &  0.326 &    1.0 \\
Category[Hazard]                    &  1.482 &   0.018 &  1.445 &  1.515 &    1.0 \\
Category[Illegal Tree Damage]       &  0.178 &   0.034 &  0.105 &  0.243 &    1.0 \\
Category[Prune]                     & -0.087 &   0.027 & -0.141 & -0.034 &    1.0 \\
Category[Remove Tree]               &  0.086 &   0.021 &  0.040 &  0.123 &    1.0 \\
Category[Root/Sewer/Sidewalk]       & -1.659 &   0.035 & -1.730 & -1.598 &    1.0 \\
\bottomrule
\end{tabular}
\end{table}

\begin{table}
\centering
\caption{Regression coefficients for Zero-inflated Poisson regression with incident-level covariates and Borough fixed effects for Max Duration 100 days, Remove all repeat callers and missing caller information.}
\label{tab:zifremoveallrepeat}
\begin{tabular}{lrrrrrr}
\toprule
 {} &   Mean &  StdDev &   2.5\% &  97.5\% &  R\_hat \\
\midrule
Intercept                           & -3.267 &   0.029 & -3.326 & -3.214 &    1.0 \\
Zero Inflation fraction                 &  0.652 &   0.004 &  0.642 &  0.659 &    1.0 \\
INSPCondtion[T.Dead]                & -0.334 &   0.034 & -0.401 & -0.270 &    1.0 \\
INSPCondtion[T.Excellent\_Good]      & -0.308 &   0.027 & -0.359 & -0.256 &    1.0 \\
INSPCondtion[T.Fair]                & -0.186 &   0.027 & -0.241 & -0.133 &    1.0 \\
INSP\_RiskAssessment                 &  0.232 &   0.012 &  0.209 &  0.256 &    1.0 \\
Log(Tree Diameter at Breast Height) & -0.028 &   0.008 & -0.045 & -0.013 &    1.0 \\
Borough[Bronx]                      & -0.059 &   0.029 & -0.116 & -0.006 &    1.0 \\
Borough[Brooklyn]                   & -0.378 &   0.021 & -0.425 & -0.339 &    1.0 \\
Borough[Manhattan]                  &  0.431 &   0.057 &  0.311 &  0.539 &    1.0 \\
Borough[Queens]                     & -0.242 &   0.021 & -0.286 & -0.202 &    1.0 \\
Borough[Staten Island]              &  0.248 &   0.031 &  0.185 &  0.307 &    1.0 \\
Category[Hazard]                    &  1.413 &   0.017 &  1.378 &  1.447 &    1.0 \\
Category[Illegal Tree Damage]       &  0.237 &   0.035 &  0.157 &  0.302 &    1.0 \\
Category[Prune]                     & -0.106 &   0.028 & -0.164 & -0.056 &    1.0 \\
Category[Remove Tree]               &  0.025 &   0.022 & -0.021 &  0.062 &    1.0 \\
Category[Root/Sewer/Sidewalk]       & -1.569 &   0.036 & -1.643 & -1.499 &    1.0 \\
\bottomrule
\end{tabular}
\end{table}

\begin{table}
\centering
\caption{Regression coefficients for Zero-inflated Poisson regression with incident-level covariates and Borough fixed effects for Max Duration 100 days, Default repeat caller removal. Risk assessment scores binned according to levels of prioritization.}
\label{tab:binnedrisk}
\begin{tabular}{lrrrrrr}
\toprule
 {} &   Mean &  StdDev &   2.5\% &  97.5\% &  R\_hat \\
\midrule
Intercept                           & -3.177	&0.205	&-3.528	&-2.872	&1.0 \\
Zero Inflation fraction                 &  0.656	&0.004	&0.650	&0.662	&1.0 \\
INSPCondition[T.Dead]                & -0.339	&0.033	&-0.390	&-0.285 &    1.0 \\
INSPCondition[T.Excellent\_Good]      & -0.410	&0.024	&-0.451	&-0.373 &    1.0 \\
INSPCondition[T.Fair]                & -0.188	&0.025	&-0.227	&-0.147 &    1.0 \\
Risk assessment category A  &0.871	&0.206	&0.555	&1.224 &1.0\\
Risk assessment category B  &0.594	&0.183	&0.290	&0.916 &1.0\\
Risk assessment category C  & -0.297	&0.189	&-0.586	&0.008 &1.0\\
Risk assessment category D  &-0.146	&0.172	&-0.400	&0.158 &1.0\\
Risk assessment category E  &-1.206	&0.834	&-2.689	&-0.018 &1.0\\
Risk assessment category Unknown & 0.184	&0.175	&-0.071	&0.490 &1.0\\
Log(Tree Diameter at Breast Height) & -0.005	&0.010	&-0.021	&0.010 &    1.0 \\
Borough[Bronx]                      & 0.086	&0.114	&-0.122	&0.268 &    1.0 \\
Borough[Brooklyn]                   & -0.350	&0.127	&-0.553	&-0.155 &    1.0 \\
Borough[Manhattan]                  &  0.765	&0.181	&0.445	&1.048 &    1.0 \\
Borough[Queens]                     & -0.600	&0.085	&-0.737	&-0.461 &    1.0 \\
Borough[Staten Island]              &  0.099	&0.192	&-0.219	&0.416 &    1.0 \\
Category[Hazard]                    &  1.473	&0.018	&1.442	&1.503 &    1.0 \\
Category[Illegal Tree Damage]       &  0.220	&0.035	&0.159	&0.276 &    1.0 \\
Category[Prune]                     & -0.057	&0.029	&-0.105	&-0.012 &    1.0 \\
Category[Remove Tree]               &  0.052	&0.022	&0.016	&0.088 &    1.0 \\
Category[Root/Sewer/Sidewalk]       & -1.687	&0.040	&-1.749	&-1.624 &    1.0 \\
\bottomrule
\end{tabular}
\end{table}

\subsection{Calculation of mean reporting delay}
\label{app:exampledelaycalc}

In this section, we give the example calculation for the mean reporting delay of a Hazard, tree in Poor condition, risk assessment score 12 incident in Manhattan. The mean delay is calculated as:
\[
1/\exp(\underbrace{-3.229}_{\text{Intercept, tree in Poor condition}} + \underbrace{1.418}_{\text{Hazard}} + \underbrace{0.438}_{\text{Manhattan}} + \underbrace{\frac{12-6.4915}{2.1788}\times 0.240}_{\text{Standardized risk assessment score}} ) \approx 2.2
\]

There are a few points worth mentioning with this calculation. First, we take the exponential of the sum of these coefficients, in accordance with the specification of the Poisson regression we fit in \Cref{eq:poissonregression}; we further take the reciprocal, since the mean of an Exponential random variable is the reciprocal of its rate. Second, coefficient estimate for the dimension of `tree in Poor condition' is integrated into the estimate for the intercept term; for tree in other conditions, an additional appropriate coefficient estimate needs to be added to the exponent. Third, tree size and risk assessment scores are standardized in the train set, thus in the calculation, we need to do the same standardization process as when we obtained the train set. Since we are concerned with trees of average size, the tree size variable is standardized to 0 here and omitted. 

\subsection{Additional information on socioeconomic and spatial reporting inequities}
\label{app:socioeconomic}

In this section, we provide some additional information on socioeconomic and spatial reporting inequities omitted from \Cref{sec:results}.

\textbf{Calculation of cumulative association of socioeconomic variables.} We first give an example of how the cumulative association scores in \Cref{fig:spatialheterogproj} are calculated. Note that all coefficients given in \Cref{tab:nycmultidemonodensity} are on standardized covariates. (Standardization was done on all covariates on the data frame used for the Poisson regression, and average refers to the corresponding average on that dataset). Let us suppose one hypothetical census tract has the following socioeconomic profile: 
\begin{itemize}
    \item Median age: 1 standard deviation above average;
    \item Fraction of white residents: 1 standard deviation below average;
    \item Fraction of residents with a college degree: 0.5 standard deviations above average;
    \item Fraction of renters: 0.5 standard deviations below average;
    \item Log income per capita: 0.5 standard deviations above average.
\end{itemize}

Then, according to \Cref{tab:nycmultidemonodensity}, the cumulative association can be calculated as:
\[
1\times (-0.014) + (-1)\times 0.058 + 0.5 \times (-0.047) + (-0.5)\times 0.042 + 0.5\times 0.086 = -0.0735,
\]
where each additive part in the equation corresponds to each of the socioeconomic variables above.

\Cref{fig:spatialmapchicagoproj} is a reproduction of \Cref{fig:spatialheterogproj} using data from Chicago. Substantial spatial differences in reporting across different areas of the city are visible. \Cref{fig:spatialheterog} present the results when our model includes indicator variables for the over 2000 census tracts in New York City and the over 2000 census block groups in Chicago. These estimates capture more fine-grained spatial differences, beyond which may occur due to socioeconomic differences.

\begin{figure}[bt]
 	\centering
\includegraphics[width=.5\textwidth]{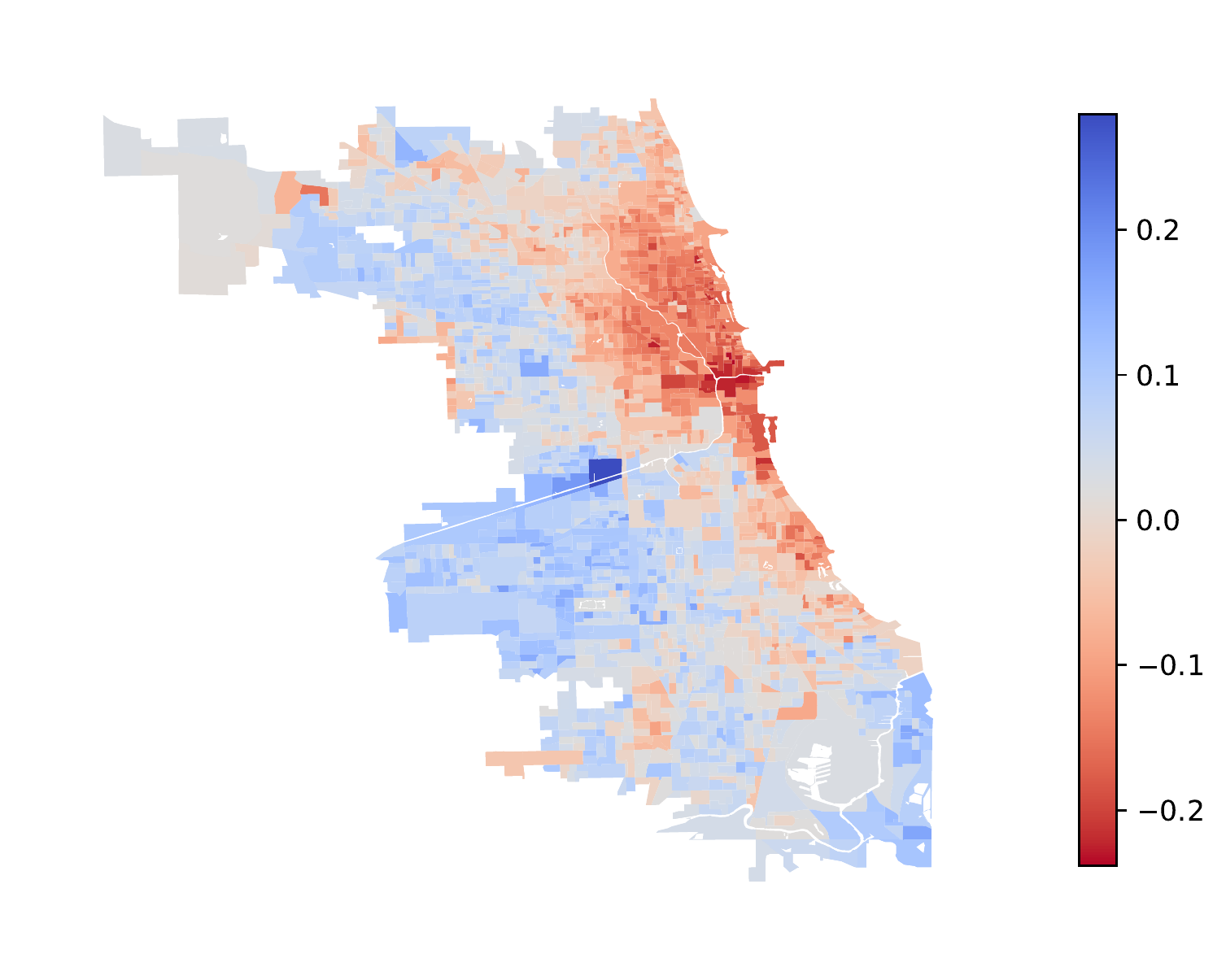}
 	\caption{Coefficients for each census block group in Chicago, representing the combined association of socioeconomic variables (log income per capita, fraction of white residents, fraction of renter, median age, and fraction of residents with college degrees) on reporting rates.}
 \label{fig:spatialmapchicagoproj}
 \end{figure}

  \begin{figure}[bt]
 	\centering
\subfloat[][Census tract coefficients in New York City]{
 		\includegraphics[width=.48\textwidth]{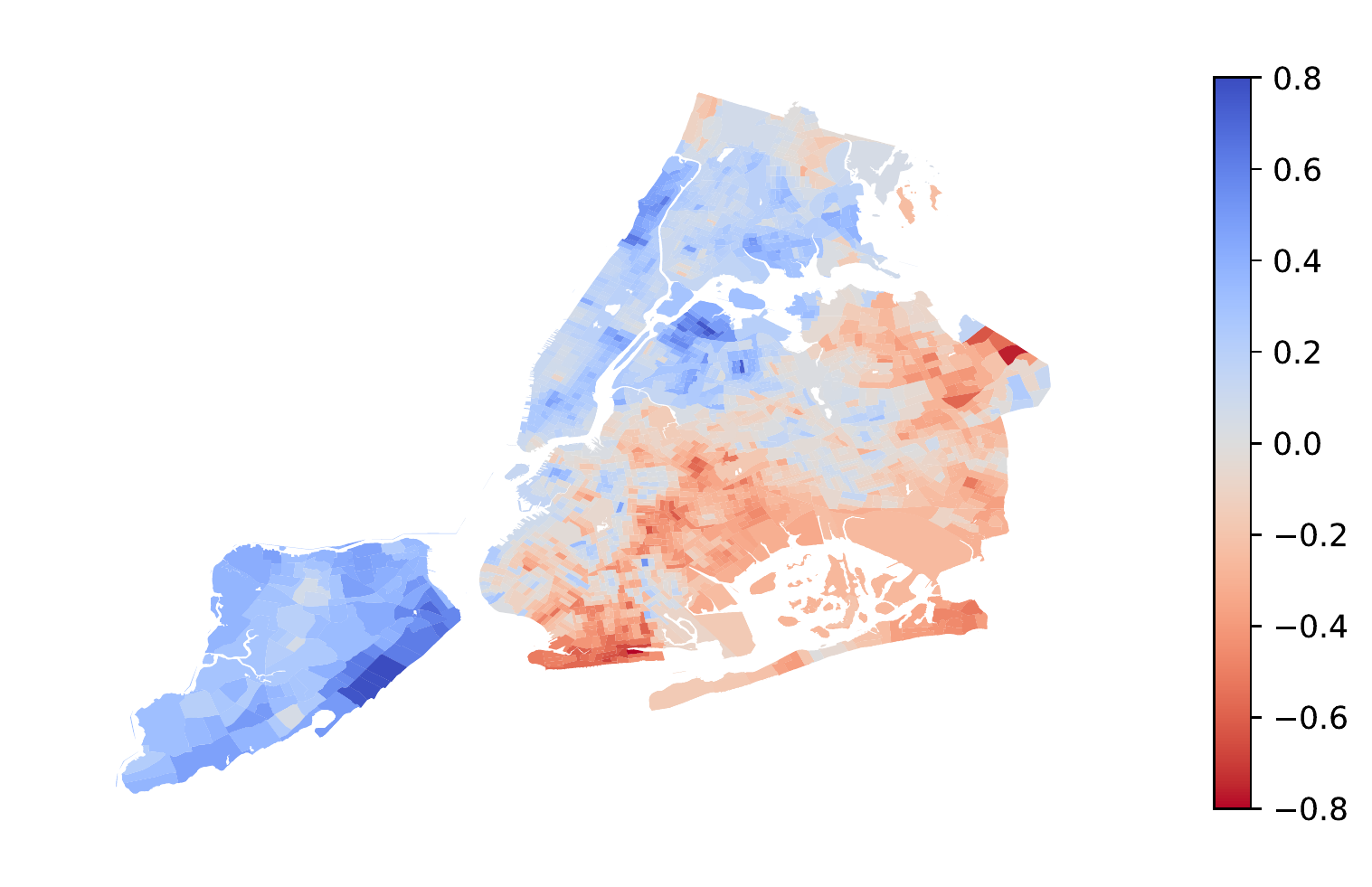}
 		\label{fig:spatialmap}
}
 	\hfill
	\subfloat[][Census block group coefficients in Chicago]{
        \vspace{-.3cm}
		\includegraphics[width=.43\textwidth]{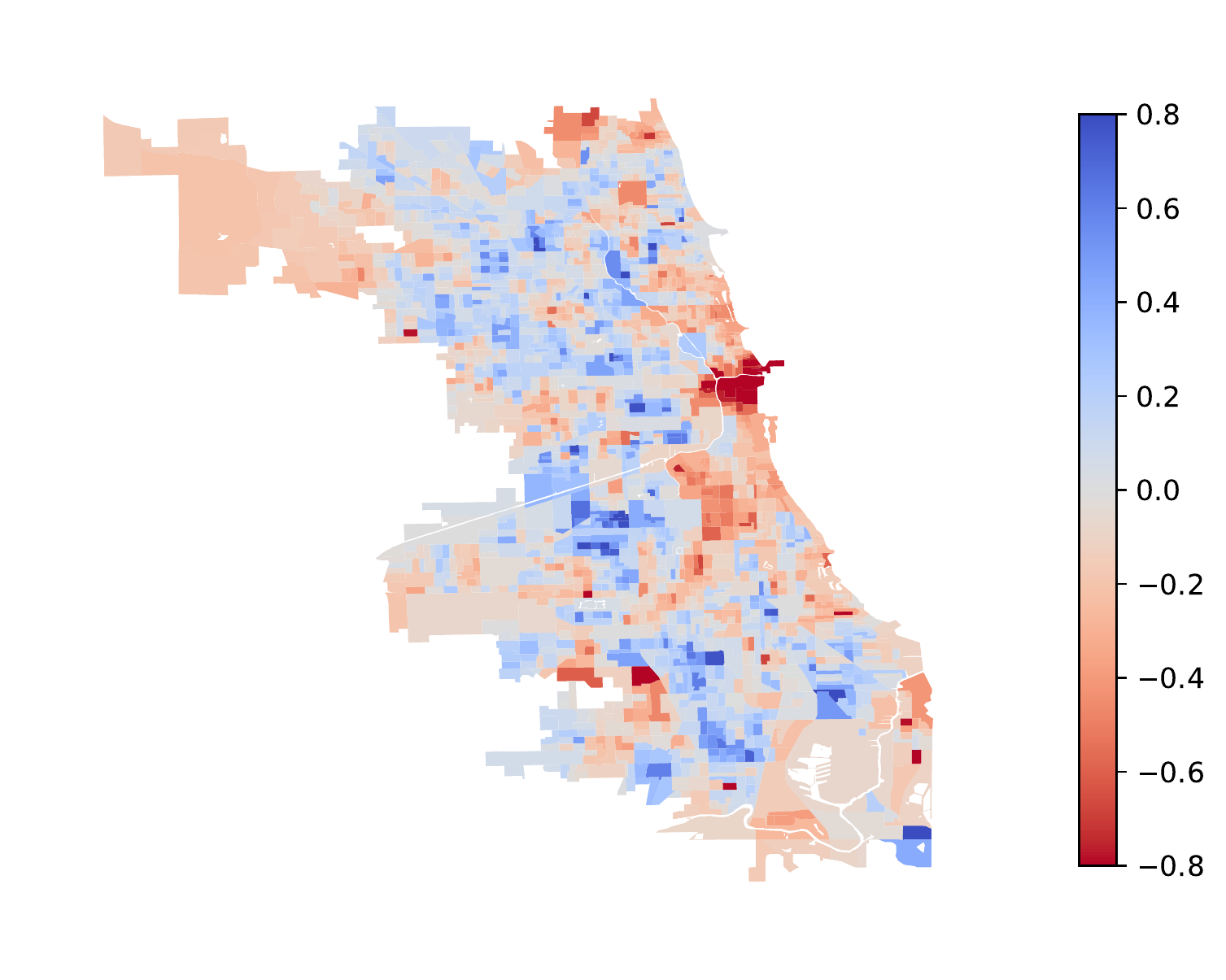}
		\label{fig:spatialmapchicago}
}
 	\caption{Coefficients on spatial variables in NYC and Chicago. These spatial coefficients are estimated using the ICAR spatial zero-inflated Poisson regression. More positive coefficients indicate higher reporting rates.}
 	\label{fig:spatialheterog}
 \end{figure}

\subsection{Replicating NYC results with public data}\label{app:nycpublic}

In the main text, we present results from NYC that were obtained from partially private data. In this subsection, we show that with all public data, we can nontheless reproduce all the results. In this subsection, we will refer to the data used in \Cref{sec:datapplication} as the ``full dataset'' and the data used here as the ``public dataset''

In \Cref{sec:datapplication}, we outlined that internal data were used to identify inspections and work orders associated with the service requests, and identify (anonymized) information about the caller. With the public dataset, we can still perform the first task, by joining the public service request data\footnote{\url{https://data.cityofnewyork.us/Environment/Forestry-Service-Requests/mu46-p9is}} with public inspection data\footnote{\url{https://data.cityofnewyork.us/Environment/Forestry-Inspections/4pt5-3vv4}}, public work order data\footnote{\url{https://data.cityofnewyork.us/Environment/Forestry-Inspections/4pt5-3vv4}}, and public risk assessment data\footnote{\url{https://data.cityofnewyork.us/Environment/Forestry-Risk-Assessments/259a-b6s7}}. During this process, we note that all of the covariates which we used in the previous analyses are retained. The second task, however, cannot be performed since public datasets do not contain any information about the caller. This means that we are unable to identify repeat callers and remove them in our analyses. This, however, does not change the final results qualitatively.

\paragraph{Summary of preprocessing} The aforementioned public datasets are joined using unique ID's on service requests and inspections. These datasets are all actively maintained and updated, and the version used here contains \num{569340} unique service requests made by the public within the period between 2/28/2015 and 9/01/2022. We then created an aggregated dataset as described in \Cref{sec:datapplication}. See Appendix \Cref{tab:publicsummary} for summary statistics, Appendix \Cref{fig:numreportshistpublic} and Appendix \Cref{fig:durationhistpublic} for distribution of number of reports per incident and length of observation inverval.

% \usepackage{graphicx}
\begin{table}[tb]
    \caption{Summary statistics from the NYC public dataset.}
    \label{tab:publicsummary}
    \centering
    \resizebox{.9\textwidth}{!}{%

    \begin{tabular}{r|r|rr|rrr}
    \multicolumn{1}{c|}{\textbf{}}             & \multicolumn{1}{c|}{\textbf{\begin{tabular}[c]{@{}c@{}}Service \\ requests\end{tabular}}} & \multicolumn{2}{c|}{\textbf{Inspections}}                                                                                                                          & \multicolumn{3}{c}{\textbf{\begin{tabular}[c]{@{}c@{}}Incidents\\ (from inspected reports)\end{tabular}}}                                                                                                                                             \\
    \multicolumn{1}{l|}{}                      & \multicolumn{1}{l|}{}                                                                     & \multicolumn{1}{c}{\begin{tabular}[c]{@{}c@{}}Inspected \\ SRs\end{tabular}} & \multicolumn{1}{c|}{\begin{tabular}[c]{@{}c@{}}Fraction\\ inspected\end{tabular}} & \multicolumn{1}{c}{\begin{tabular}[c]{@{}c@{}}Unique\\ incidents\end{tabular}} & \multicolumn{1}{c}{\begin{tabular}[c]{@{}c@{}}Avg. reports\\per incident\end{tabular}} & \multicolumn{1}{c}{\begin{tabular}[c]{@{}c@{}}Median Days\\to Inspection\end{tabular}} \\ \hline
    \multicolumn{1}{l|}{\textbf{Total number}} &
    
    \num{569340} &              \num{363676} &                0.64 &            \num{246744} &                  1.47 &                              6.53  \\
    \multicolumn{1}{l|}{\textbf{By Borough}}   &                                                                                           &                                                                              &                                                                                     &                                                                                &                                                                                &                                                                                     \\
    \textit{Queens} &    \num{235363} &              \num{152985} &                0.65 &            \num{107524} &                  1.42 &                              5.99 \\
    \textit{Brooklyn} &    \num{176030} &              \num{118782} &                0.67 &             \num{73067} &                  1.63 &                              8.19 \\
    \textit{Staten Island} &     \num{69436} &               \num{35415} &                0.51 &             \num{25710} &                  1.38 &                              5.24 \\
    \textit{Bronx} &     \num{54839} &               \num{38336} &                0.70 &             \num{26110} &                  1.47 &                              7.89 \\
    \textit{Manhattan} &     \num{33659} &               \num{18156} &                0.54 &             \num{14346} &                  1.27 &                              2.82 \\
    \multicolumn{1}{l|}{\textbf{By Category}}  &                                                                                           &                                                                              &                                                                                     &                                                                                &                                                                                &                                                                                     \\
    \textit{Hazard} &    \num{236593} &              \num{179477} &                0.76 &            \num{118135} &                  1.52 &                              2.80 \\
    \textit{Prune} &    \num{124725} &               \num{42929} &                0.34 &             \num{34305} &                  1.25 &                             17.97 \\
    \textit{Remove Tree} &    \num{105960} &               \num{83179} &                0.79 &             \num{63422} &                  1.31 &                              7.55 \\
    \textit{Root/Sewer/Sidewalk} &     \num{72900} &               \num{41837} &                0.57 &             \num{33450} &                  1.25 &                             38.07 \\
    \textit{Illegal Tree Damage} &     \num{27479} &               \num{15554}&                0.57 &             \num{12447} &                  1.25 &                             17.38 \\
    \textit{Other} &      1,683 &                 700 &                0.42 &               645 &                  1.09	 &                              5.70 \\
    \end{tabular}%
    }
    
    \end{table}

\begin{figure}[t]
	\centering
	%	\begin{subfigure}{.5\textwidth}
	% 		\centering
	\subfloat[][Number of reports per incident.]{
		\includegraphics[width=.45\textwidth]{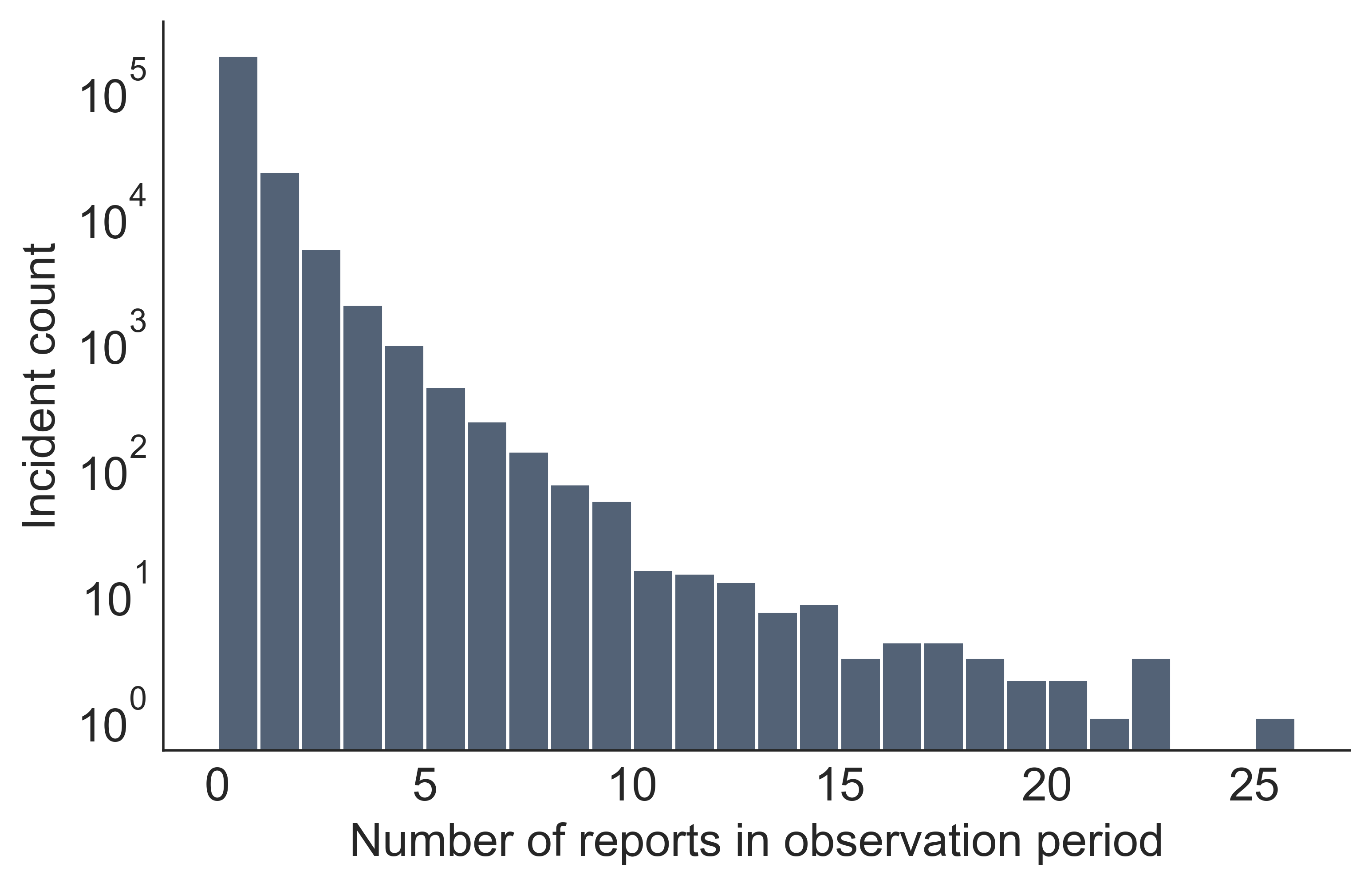}
		% 		\caption{Census tract fixed effects}
		\label{fig:numreportshistpublic}
		% 		\end{subfigure}
	}
	\hfill
	\subfloat[][Length of observation period.]{
		\includegraphics[width=.45\textwidth]{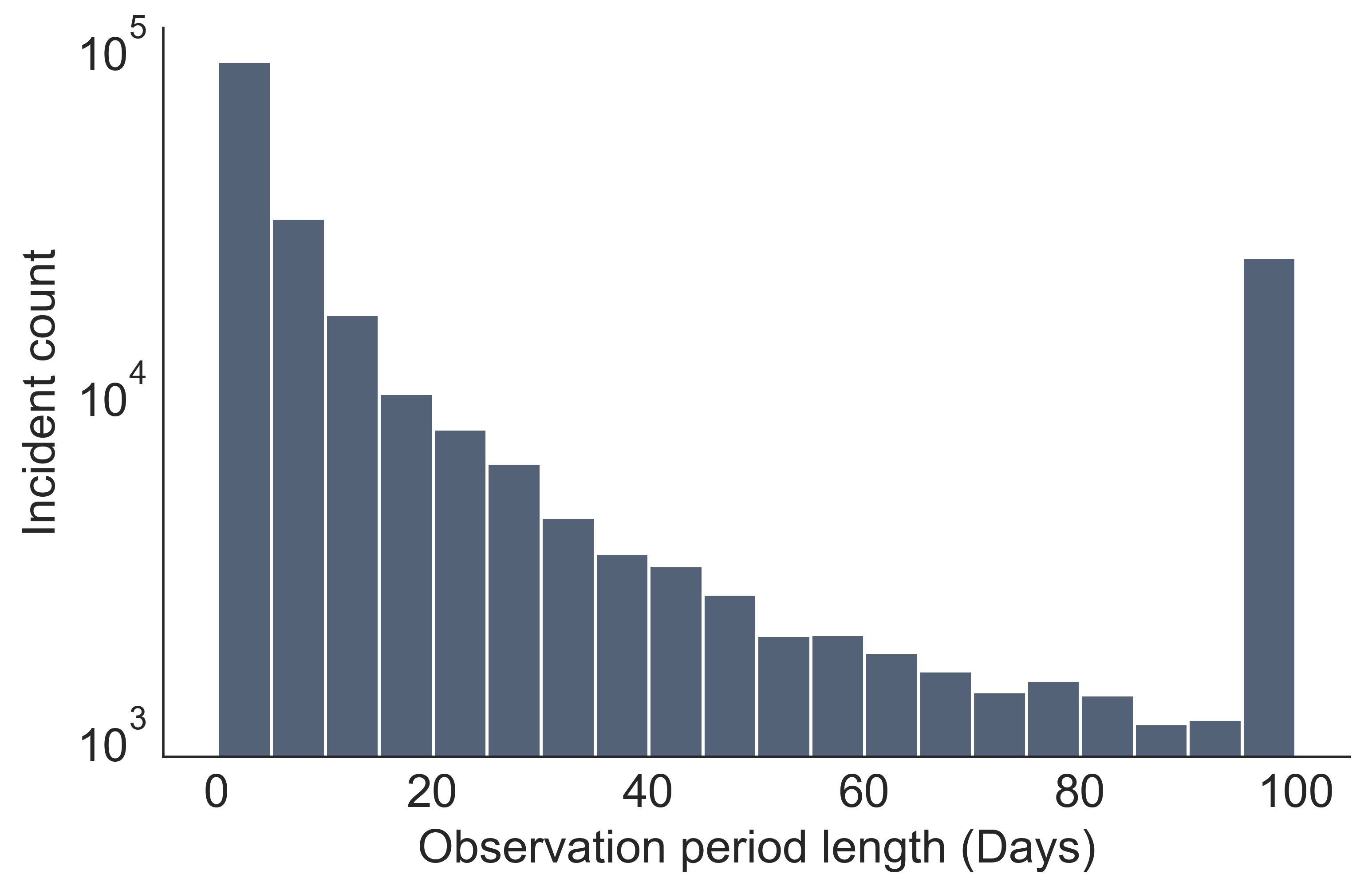}
		% 		\caption{Census tract fixed effects}
		\label{fig:durationhistpublic}
		% 		\end{subfigure}
	}
	\caption{Distribution of number of reports and length of observation for each unique incident in the aggregated dataset using public data.}
\end{figure}

\paragraph{Summary of results} We run all the analyses in \Cref{sec:datapplication} on the public data set. Appendix \Cref{tab:nycbasicpublic100} summarizes results with the \textbf{Base} variables, while Appendix \Cref{tab:nycbasicpublic30} and Appendix \Cref{tab:nycbasicpublic200} serve as robustness checks. Appendix \Cref{fig:spatialmappublic} shows the coefficients on \textbf{Spatial} variables. Appendix \Cref{tab:nycdemopublicnoborough}, \Cref{tab:nycdemopublic}, and \Cref{tab:nycmultidemopublic} show the coefficients on census tract \textbf{Socioeconomic} covariates.

\begin{table}
\centering
\caption{Regression coefficients for Base variables in NYC, for Max Duration 100 days, estimated with public data}
\label{tab:nycbasicpublic100}
\begin{tabular}{lrrrrrr}
\toprule
{} &   Mean &  StdDev &   2.5\% &  97.5\% &  R\_hat \\
\midrule
Intercept                           & -3.449 &   0.028 & -3.507 & -3.401 &    1.0 \\
Zero Inflation fraction                 &  0.539 &   0.005 &  0.529 &  0.547 &    1.0 \\
INSPCondition[T.Dead]                & -0.163 &   0.035 & -0.237 & -0.103 &    1.0 \\
INSPCondition[T.Excellent\_Good]      & -0.698 &   0.028 & -0.755 & -0.646 &    1.0 \\
INSPCondition[T.Fair]                & -0.436 &   0.025 & -0.492 & -0.387 &    1.0 \\
INSP\_RiskAssessment                 &  0.262 &   0.013 &  0.235 &  0.287 &    1.0 \\
Log(Tree Diameter at Breast Height) &  0.079 &   0.011 &  0.053 &  0.097 &    1.0 \\
Borough[Bronx]                      &  0.029 &   0.026 & -0.022 &  0.076 &    1.0 \\
Borough[Brooklyn]                   & -0.178 &   0.017 & -0.214 & -0.147 &    1.0 \\
Borough[Manhattan]                  &  0.635 &   0.039 &  0.558 &  0.701 &    1.0 \\
Borough[Queens]                     & -0.516 &   0.017 & -0.549 & -0.485 &    1.0 \\
Borough[Staten Island]              &  0.030 &   0.029 & -0.030 &  0.084 &    1.0 \\
Category[Hazard]                    &  1.583 &   0.019 &  1.547 &  1.621 &    1.0 \\
Category[Illegal Tree Damage]       &  0.039 &   0.038 & -0.033 &  0.111 &    1.0 \\
Category[Prune]                     & -0.108 &   0.028 & -0.166 & -0.056 &    1.0 \\
Category[Remove Tree]               & -0.173 &   0.023 & -0.216 & -0.131 &    1.0 \\
Category[Root/Sewer/Sidewalk]       & -1.341 &   0.035 & -1.413 & -1.278 &    1.0 \\
\bottomrule
\end{tabular}
\end{table}

\begin{table}
\centering
\caption{Regression coefficients for Base variables in NYC, for Max Duration 30 days, estimated with public data}
\label{tab:nycbasicpublic30}
\begin{tabular}{lrrrrrr}
\toprule
{} &   Mean &  StdDev &   2.5\% &  97.5\% &  R\_hat \\
\midrule
Intercept                           & -2.923 &   0.029 & -2.987 & -2.873 &    1.0 \\
Zero Inflation fraction                 &  0.618 &   0.004 &  0.610 &  0.626 &    1.0 \\
INSPCondition[T.Dead]                & -0.210 &   0.037 & -0.282 & -0.140 &    1.0 \\
INSPCondition[T.Excellent\_Good]      & -0.691 &   0.032 & -0.751 & -0.627 &    1.0 \\
INSPCondition[T.Fair]                & -0.429 &   0.029 & -0.482 & -0.369 &    1.0 \\
INSP\_RiskAssessment                 &  0.241 &   0.013 &  0.213 &  0.264 &    1.0 \\
Log(Tree Diameter at Breast Height) &  0.083 &   0.011 &  0.060 &  0.103 &    1.0 \\
Borough[Bronx]                      & -0.099 &   0.026 & -0.153 & -0.051 &    1.0 \\
Borough[Brooklyn]                   & -0.053 &   0.018 & -0.088 & -0.018 &    1.0 \\
Borough[Manhattan]                  &  0.536 &   0.038 &  0.455 &  0.604 &    1.0 \\
Borough[Queens]                     & -0.374 &   0.018 & -0.413 & -0.339 &    1.0 \\
Borough[Staten Island]              & -0.010 &   0.029 & -0.069 &  0.041 &    1.0 \\
Category[Hazard]                    &  1.458 &   0.019 &  1.415 &  1.492 &    1.0 \\
Category[Illegal Tree Damage]       &  0.252 &   0.035 &  0.175 &  0.311 &    1.0 \\
Category[Prune]                     & -0.191 &   0.033 & -0.256 & -0.126 &    1.0 \\
Category[Remove Tree]               & -0.241 &   0.026 & -0.298 & -0.193 &    1.0 \\
Category[Root/Sewer/Sidewalk]       & -1.277 &   0.040 & -1.364 & -1.211 &    1.0 \\
\bottomrule
\end{tabular}
\end{table}

\begin{table}
\centering
\caption{Regression coefficients for Base variables in NYC, for Max Duration 200 days, estimated with public data}
\label{tab:nycbasicpublic200}
\begin{tabular}{lrrrrrr}
\toprule
{} &   Mean &  StdDev &   2.5\% &  97.5\% &  R\_hat \\
\midrule
Intercept                           & -3.670 &   0.027 & -3.720 & -3.614 &    1.0 \\
Zero Inflation fraction                 &  0.508 &   0.005 &  0.499 &  0.517 &    1.0 \\
INSPCondition[T.Dead]                & -0.116 &   0.037 & -0.192 & -0.049 &    1.0 \\
INSPCondition[T.Excellent\_Good]      & -0.631 &   0.028 & -0.690 & -0.579 &    1.0 \\
INSPCondition[T.Fair]                & -0.389 &   0.027 & -0.447 & -0.340 &    1.0 \\
INSP\_RiskAssessment                 &  0.256 &   0.012 &  0.233 &  0.279 &    1.0 \\
Log(Tree Diameter at Breast Height) &  0.059 &   0.011 &  0.039 &  0.079 &    1.0 \\
Borough[Bronx]                      &  0.130 &   0.024 &  0.079 &  0.173 &    1.0 \\
Borough[Brooklyn]                   & -0.177 &   0.018 & -0.216 & -0.147 &    1.0 \\
Borough[Manhattan]                  &  0.585 &   0.039 &  0.495 &  0.657 &    1.0 \\
Borough[Queens]                     & -0.525 &   0.018 & -0.562 & -0.492 &    1.0 \\
Borough[Staten Island]              & -0.013 &   0.028 & -0.067 &  0.039 &    1.0 \\
Category[Hazard]                    &  1.595 &   0.017 &  1.562 &  1.624 &    1.0 \\
Category[Illegal Tree Damage]       &  0.183 &   0.033 &  0.111 &  0.242 &    1.0 \\
Category[Prune]                     & -0.209 &   0.023 & -0.260 & -0.167 &    1.0 \\
Category[Remove Tree]               & -0.148 &   0.025 & -0.209 & -0.107 &    1.0 \\
Category[Root/Sewer/Sidewalk]       & -1.422 &   0.032 & -1.482 & -1.360 &    1.0 \\
\bottomrule
\end{tabular}
\end{table}

\begin{figure}[tb]
	\centering
	\includegraphics[width = .5 \textwidth]{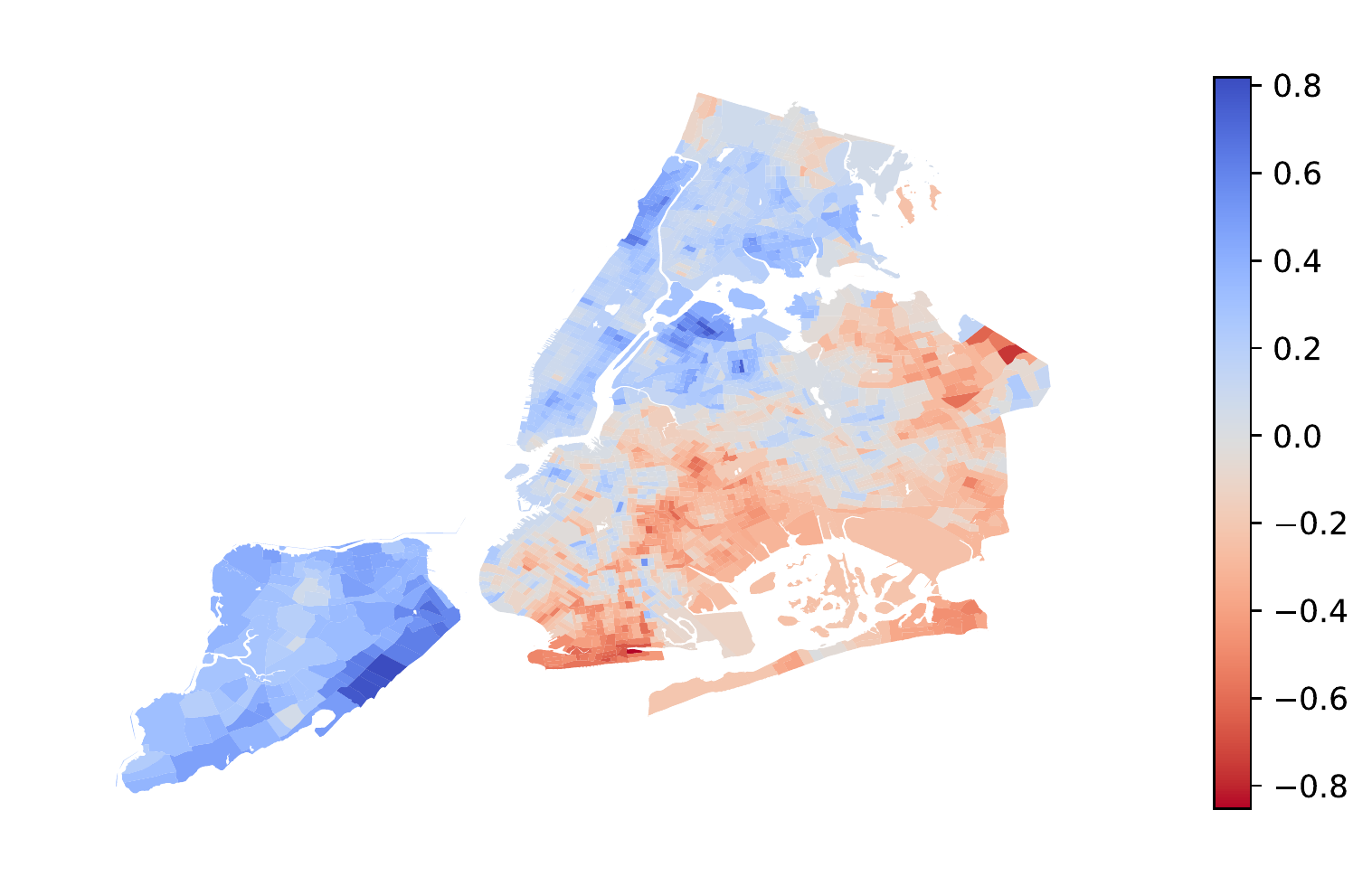}
	\caption{Coefficients on spatial variables in NYC, estimated with public data}
	\label{fig:spatialmappublic}
\end{figure}

\begin{table}
\centering
\caption{Census Tract Socioeconomic coefficients, estimated alone in a regression alongside
the incident-specific covariates, with all public data}
\label{tab:nycdemopublicnoborough}
\begin{tabular}{lrrrrr}
\toprule
                               &   Mean &  StdDev &   2.5\% &  97.5\% \\
\midrule
                    Median age & -0.061 &   0.008 & -0.078 & -0.045 \\
             Fraction Hispanic &  0.042 &   0.009 &  0.022 &  0.059 \\
                Fraction white &  0.109 &   0.008 &  0.092 &  0.124 \\
                Fraction Black & -0.065 &   0.008 & -0.083 & -0.050 \\
Fraction no high school degree & -0.007 &   0.009 & -0.026 &  0.011 \\
       Fraction college degree &  0.056 &   0.009 &  0.037 &  0.072 \\
              Fraction poverty &  0.030 &   0.008 &  0.014 &  0.046 \\
               Fraction renter &  0.126 &   0.010 &  0.104 &  0.143 \\
               Fraction family & -0.111 &   0.008 & -0.129 & -0.097 \\
       Log(Median house value) &  0.085 &   0.009 &  0.066 &  0.101 \\
        Log(Income per capita) &  0.050 &   0.010 &  0.030 &  0.067 \\
                  Log(Density) &  0.145 &   0.010 &  0.122 &  0.164 \\
\bottomrule
\end{tabular}
\end{table}

\begin{table}
\centering
\caption{Census Tract Socioeconomic coefficients, estimated alone in a regression alongside
the incident-specific covariates and the \textbf{borough fixed effects}, with all public data}
\label{tab:nycdemopublic}
\begin{tabular}{lrrrrr}
\toprule
                               &   Mean &  StdDev &   2.5\% &  97.5\% \\
\midrule
                    Median age & -0.032 &   0.008 & -0.050 & -0.018 \\
             Fraction Hispanic &  0.043 &   0.009 &  0.023 &  0.059 \\
                Fraction white &  0.069 &   0.009 &  0.050 &  0.087 \\
                Fraction Black & -0.075 &   0.008 & -0.092 & -0.060 \\
Fraction no high school degree &  0.005 &   0.009 & -0.014 &  0.023 \\
       Fraction college degree &  0.027 &   0.009 &  0.009 &  0.043 \\
              Fraction poverty & -0.014 &   0.009 & -0.033 &  0.002 \\
               Fraction renter &  0.078 &   0.010 &  0.057 &  0.098 \\
               Fraction family & -0.045 &   0.010 & -0.064 & -0.026 \\
       Log(Median house value) &  0.061 &   0.010 &  0.039 &  0.082 \\
        Log(Income per capita) &  0.015 &   0.009 & -0.005 &  0.032 \\
                  Log(Density) &  0.104 &   0.011 &  0.082 &  0.123 \\
\bottomrule
\end{tabular}
\end{table}

\begin{table}
\centering
\caption{Census Tract Socioeconomic coefficients, estimated together in a regression alongside
the incident-specific covariates, with all public data}
\label{tab:nycmultidemopublicnodensity}
\begin{tabular}{lrrrrr}
\toprule
                        &   Mean &  StdDev &   2.5\% &  97.5\% \\
\midrule
             Median age & -0.024 &   0.010 & -0.046 &  -0.005 \\
         Fraction white &  0.100 &   0.010 &  0.081 &  0.118 \\
Fraction college degree & -0.018 &   0.014 & -0.048 &  0.007 \\
        Fraction renter &  0.142 &   0.010 &  0.123 &  0.160 \\
 Log(Income per capita) &  0.073 &   0.014 &  0.047 &  0.103 \\
\bottomrule
\end{tabular}
\end{table}

\begin{table}
\centering
\caption{Census Tract Socioeconomic coefficients, estimated together in a regression alongside
the incident-specific covariates, with all public data. Compared with \Cref{tab:nycmultidemopublic}, we additionally include (log) population density, which does not substantially affect the results.}
\label{tab:nycmultidemopublic}
\begin{tabular}{lrrrrr}
\toprule
                        &   Mean &  StdDev &   2.5\% &  97.5\% \\
\midrule
             Median age & -0.010 &   0.010 & -0.030 &  0.009 \\
         Fraction white &  0.110 &   0.010 &  0.091 &  0.129 \\
Fraction college degree & -0.022 &   0.013 & -0.050 &  0.004 \\
        Fraction renter &  0.081 &   0.012 &  0.057 &  0.104 \\
 Log(Income per capita) &  0.069 &   0.015 &  0.039 &  0.095 \\
           Log(Density) &  0.116 &   0.011 &  0.092 &  0.137 \\
\bottomrule
\end{tabular}
\end{table}

\FloatBarrier
\subsection{Supplementary information on method validation}

\label{app:additionalvalidation}

In this section, we first provide more information on the validation of model results using storm-inferred ground truth (in \Cref{sec:validation_hurricanes}), in \Cref{app:hurricane}. We further present an additional validation technique in \Cref{app:oosmodel}. Finally, we validate our estimated cumulative association of socioeconomic variables with reporting rate against a measure of civic engagement -- participation in voting -- in \Cref{app:voterparticipation}. 

% \nikhil{I don't think this paragraph correctly reflects what is in main text and what is new in appendix?}

% \todo{still need to finish this section?}

\subsubsection{Supplementary information on validation using reports after hurricanes}
\label{app:hurricane}

In \Cref{sec:validation_hurricanes}, we validate our results using data immediately after Tropical Storm Isaias affected New York City, until 12 PM on 8/14/2020. Here we present results from similar analyses using data on Isaias with different end-time filters and from Tropical Storm Ida (\Cref{fig:repisaias}). Note that we must filter for incidents only immediately after the storm, because new incidents, e.g., 2 months after the event may not have been caused by the storm, as opposed to events immediately after the storm.

For Tropical Storm Ida, we filter for service requests between 12 PM on 9/01/2021 and 12 PM on various days, as specified in \Cref{fig:ida}, and define the true reporting delay as the time between the first report of an incident and 12 PM on 9/01/2021. We then estimate the reporting delays, similarly using coefficients on census tract and incident level variables. The individual incident level Pearson correlation between true reporting delays and model-estimated reporting delays is significant ($r = 0.13, p<.001$), while the means of them within each of the 30 bins defined on the model-estimated reporting delays exhibit even stronger levels of correlation, as shown in \Cref{fig:ida}.

For other storm events that affected New York City in recent years, no significant increase in the amount of tree-related service requests was observed as in \Cref{fig:histdelayhurricane}, and so we limit our analyses to these two hurricanes.

\begin{figure}[tbh]
	\centering
	\subfloat[][Filter for reports before 12 PM on 8/13/2020.]{
		\includegraphics[width=.44\textwidth]{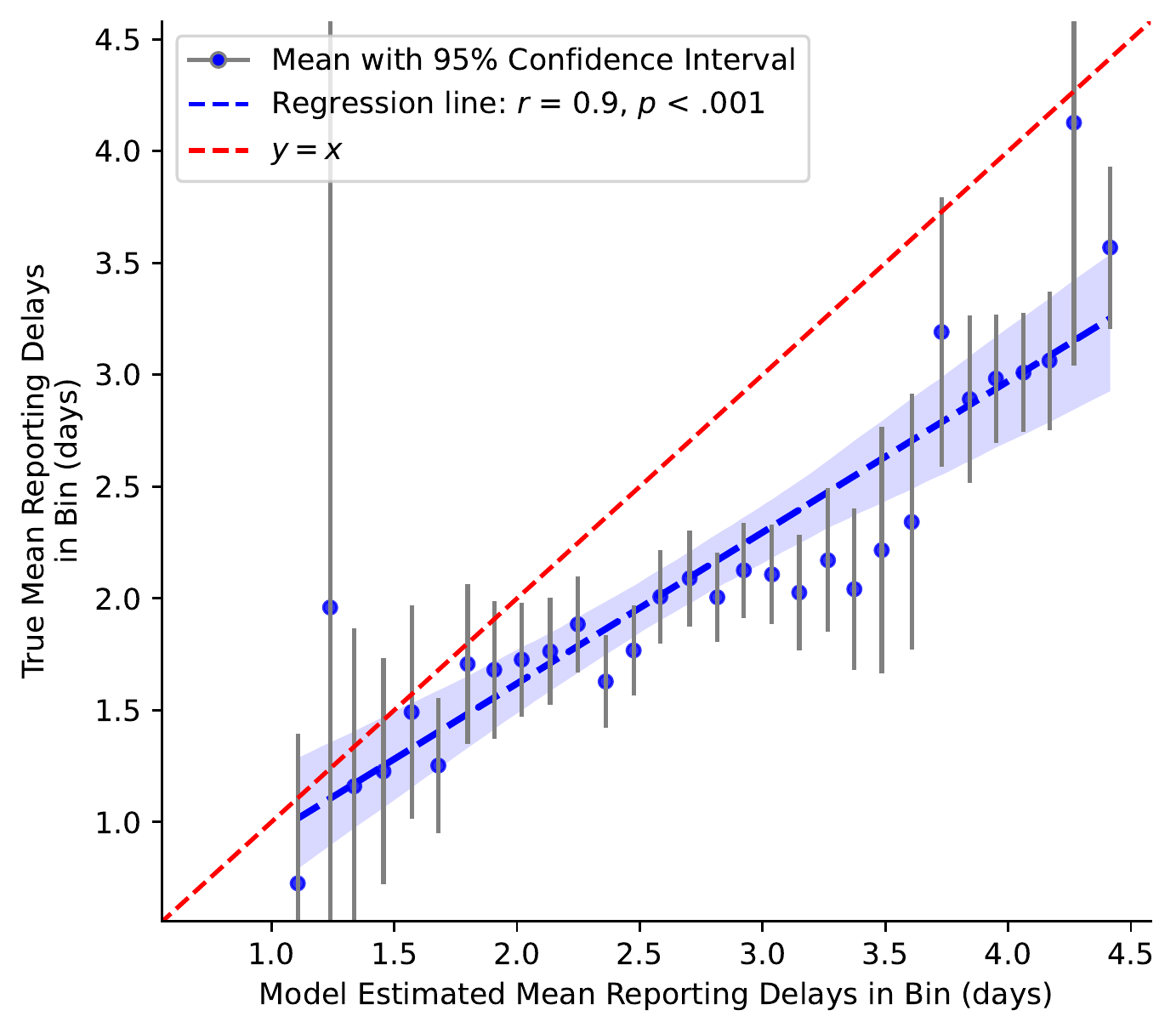}
		\label{fig:isaias8}
	}
\subfloat[][Filter for reports before 12 PM on 8/12/2020.]{
		\includegraphics[width=.44\textwidth]{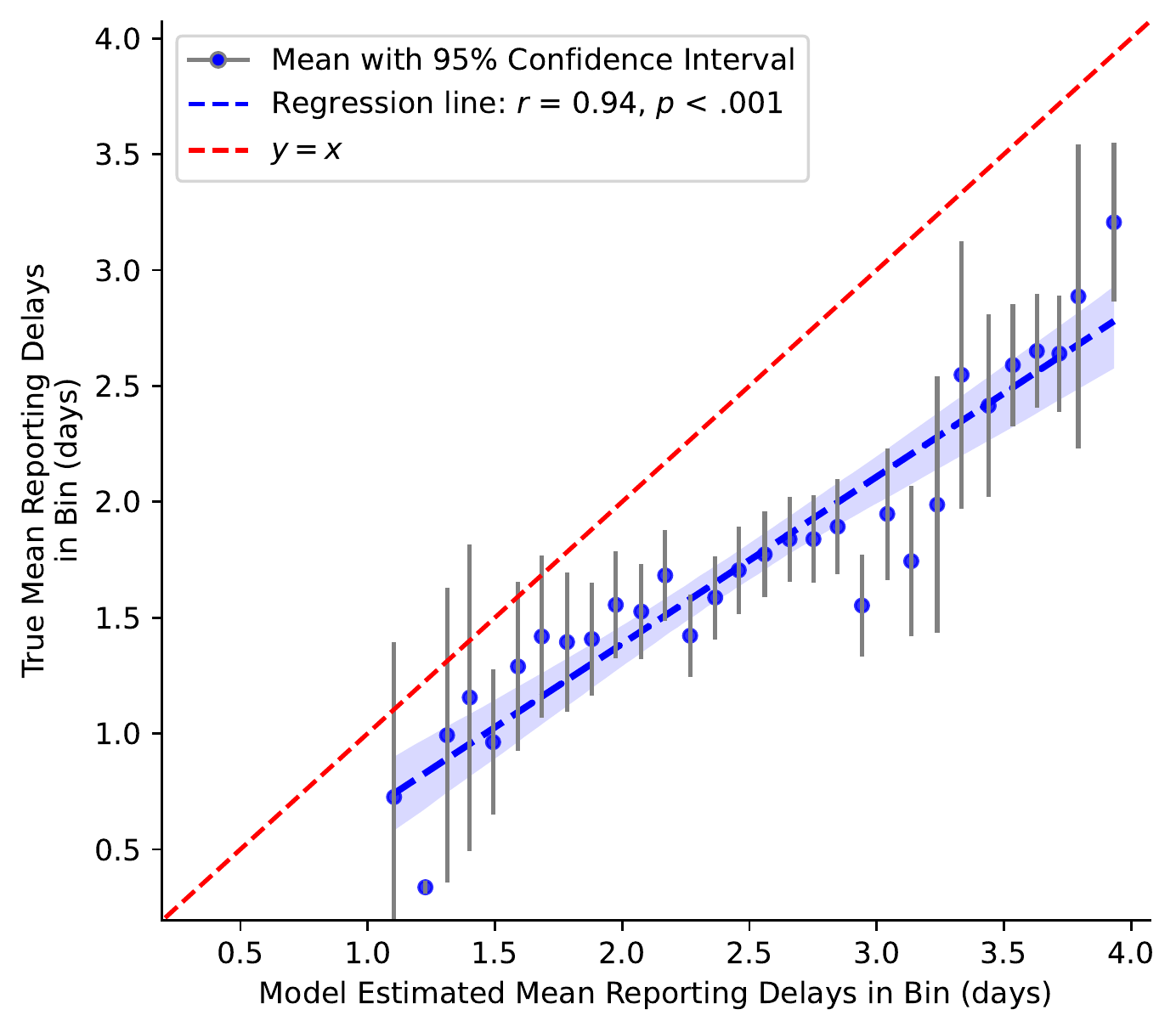}
		\label{fig:isaias7}
	}
 	\hfill
 \subfloat[][Filter for reports before 12 PM on 8/11/2020.]{
		\includegraphics[width=.44\textwidth]{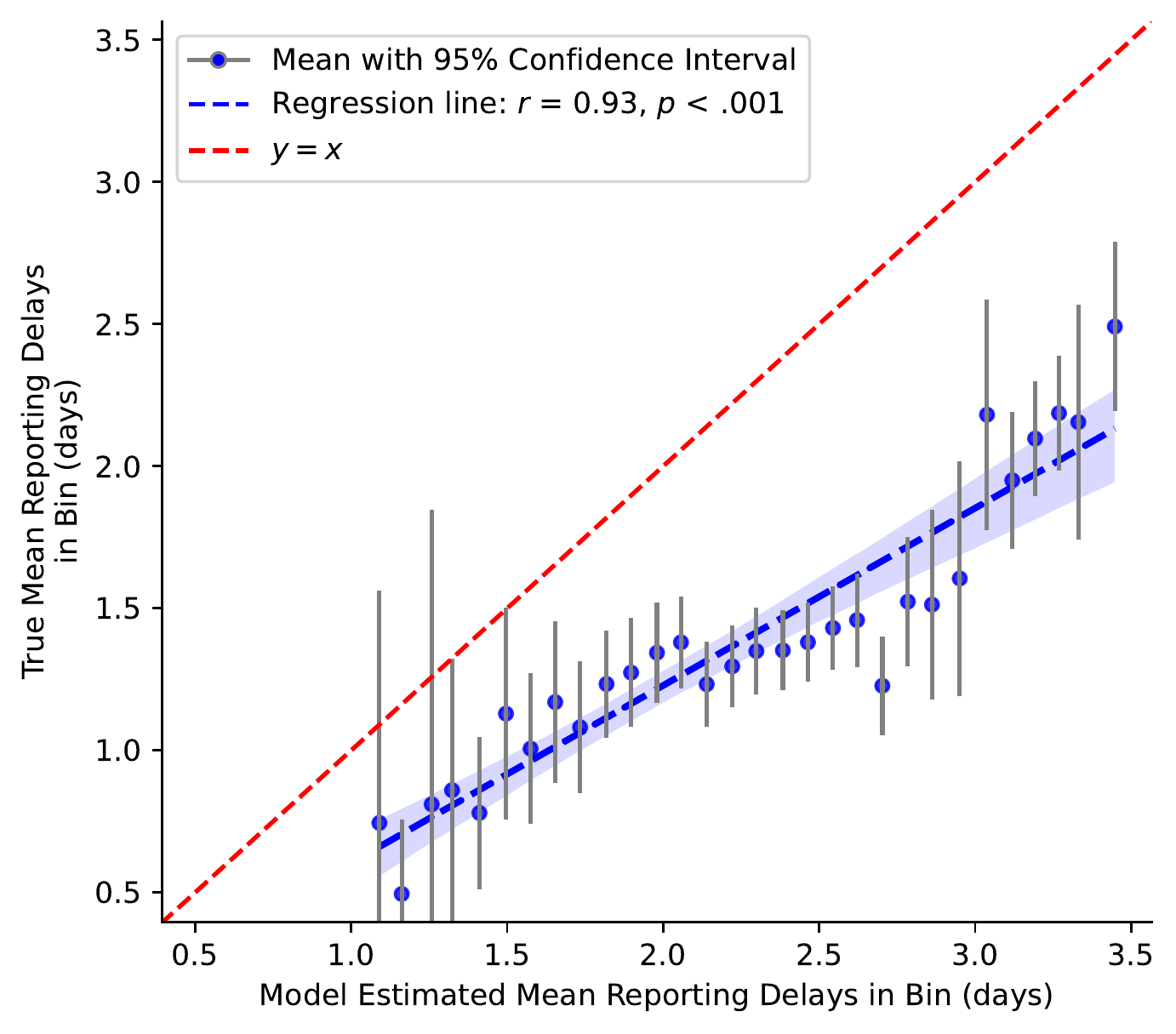}
		\label{fig:isaias6}
	}
 \subfloat[][Filter for reports before 12 PM on 8/10/2020.]{
		\includegraphics[width=.44\textwidth]{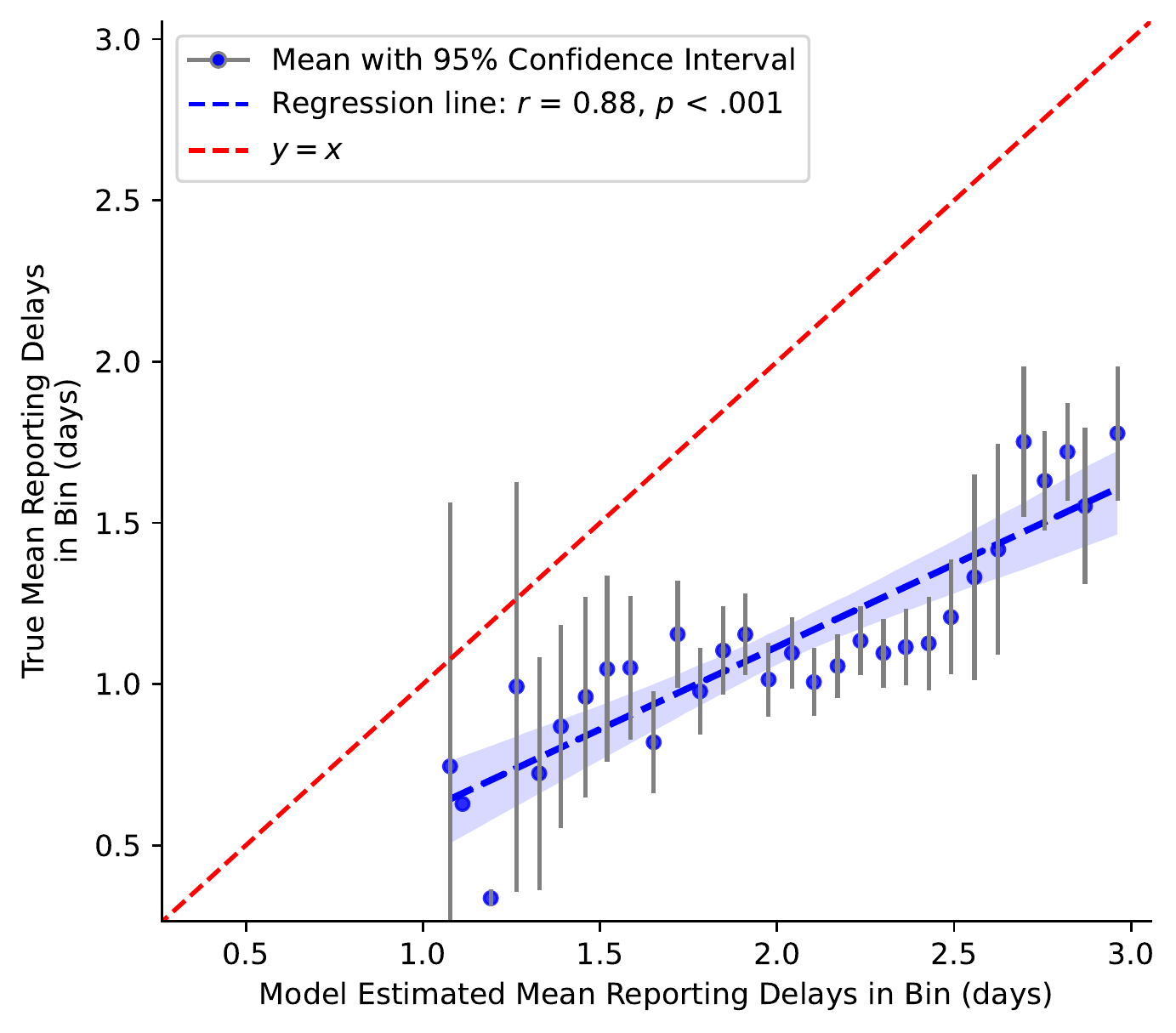}
		\label{fig:isaias5}
	}
 \hfill
	\caption{Replication of \Cref{fig:isaiasdelay} using data with different filtering for the end date. }
    \label{fig:repisaias}
\end{figure}

% For Tropical Storm Elsa, we filter for service requests between 5 AM on 07/09/2021 and end of day on 07/14/2021, and define the true reporting delay as the time between first report of an incident and 5 AM on 07/09/2021; for Tropical Storm Ida, we filter for service requests between 12 PM on 09/01/2021 and end of day on 09/10/2021, and define the true reporting delay as the time between first report of an incident and 12 PM on 09/01/2021. We then estimate the reporting delays, similarly using coefficients on census tract and incident level variables. The individual incident level Pearson correlation between true reporting delays and model-estimated reporting delays is significant for both Elsa ($r = 0.16, p<.001$) and Ida ($r = 0.13, p<.001$), while the means of them within each of the 10 bins defined on the model-estimated reporting delays exhibit even stronger levels of correlation, as shown in \Cref{fig:rephurricanedelay}.

\begin{figure}[tb]
	\centering
	%	\begin{subfigure}{.5\textwidth}
	% 		\centering
	\subfloat[][Filter for reports before 12 PM on 9/10/2021.]{
		\includegraphics[width=.45\textwidth]{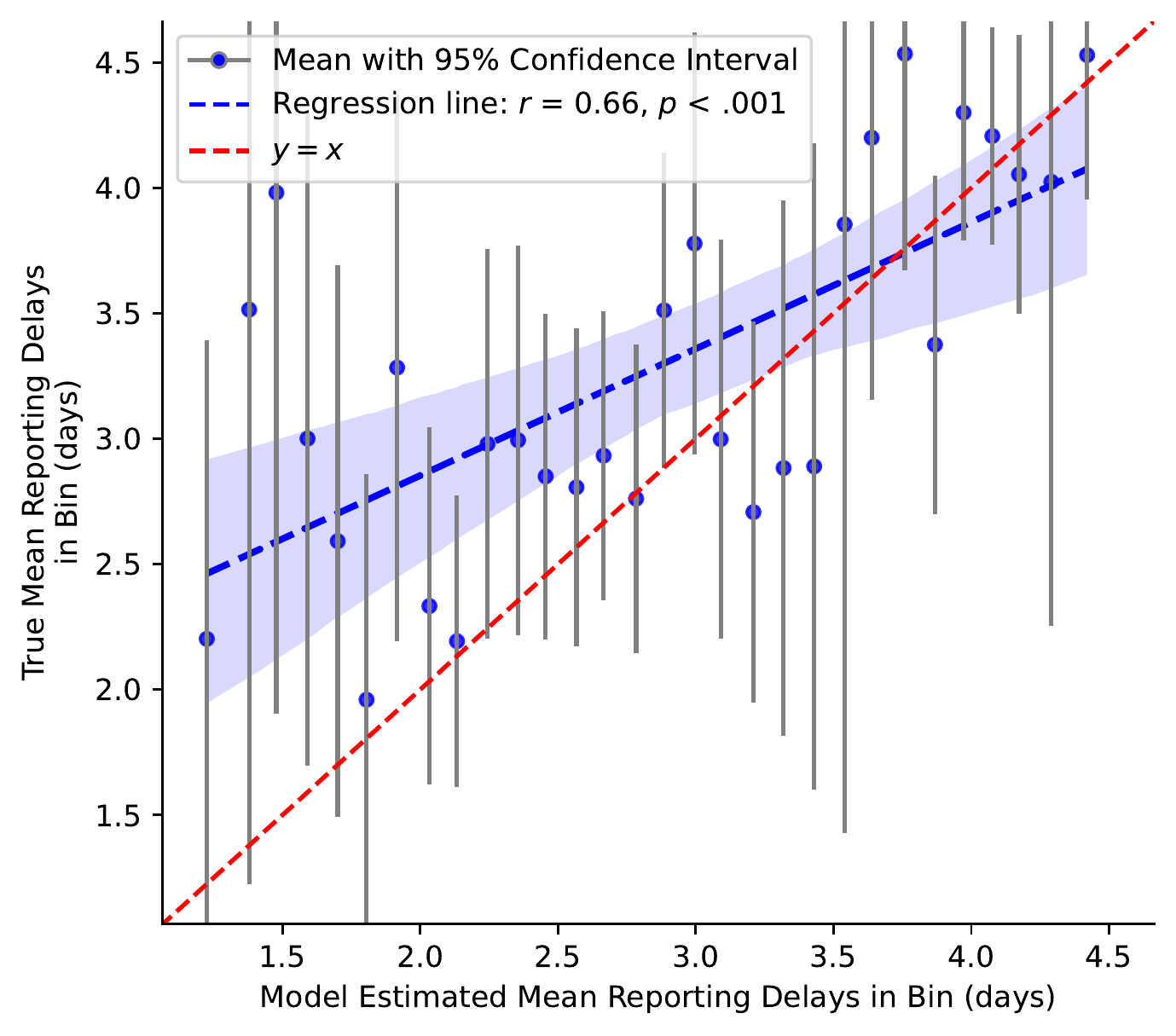}
		% 		\caption{Census tract fixed effects}
		\label{fig:ida9}
		% 		\end{subfigure}
	}
	\subfloat[][Filter for reports before 12 PM on 9/9/2021.]{
		\includegraphics[width=.45\textwidth]{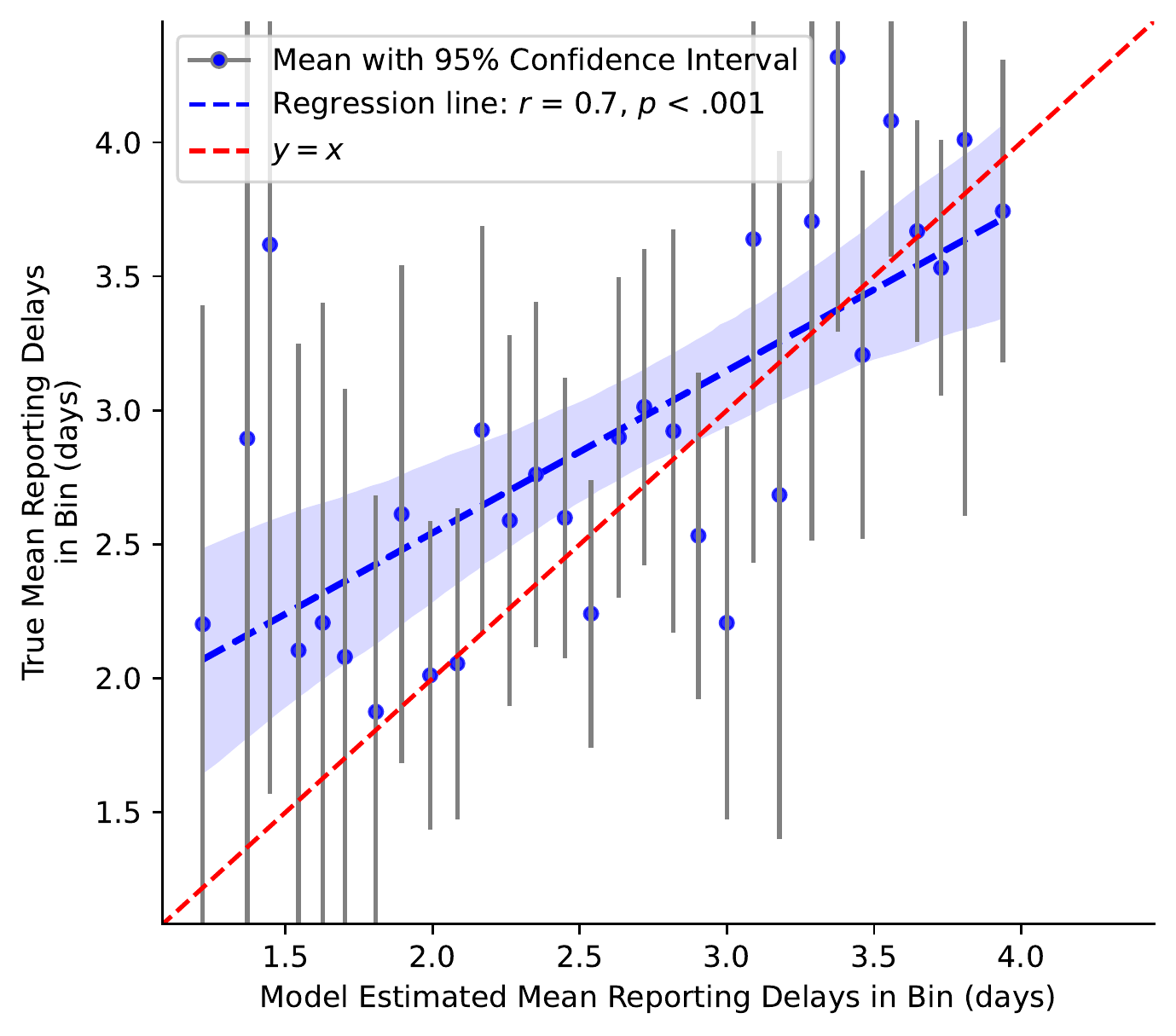}
		% 		\caption{Census tract fixed effects}
		\label{fig:ida8}
		% 		\end{subfigure}
	}
 	\hfill
  	\subfloat[][Filter for reports before 12 PM on 9/8/2021.]{
		\includegraphics[width=.45\textwidth]{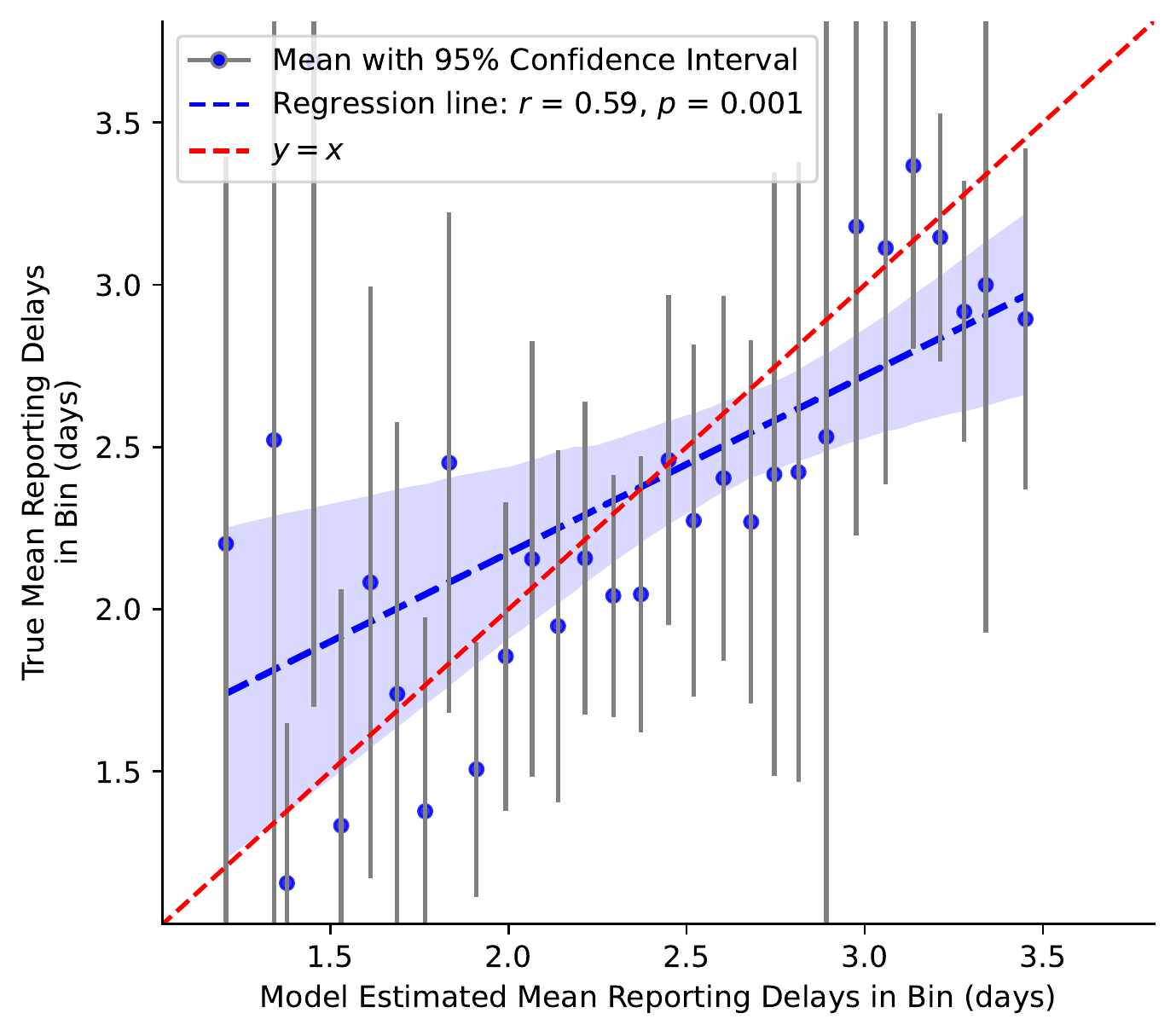}
		% 		\caption{Census tract fixed effects}
		\label{fig:ida7}
		% 		\end{subfigure}
	}
 	\subfloat[][Filter for reports before 12 PM on 9/7/2021.]{
		\includegraphics[width=.45\textwidth]{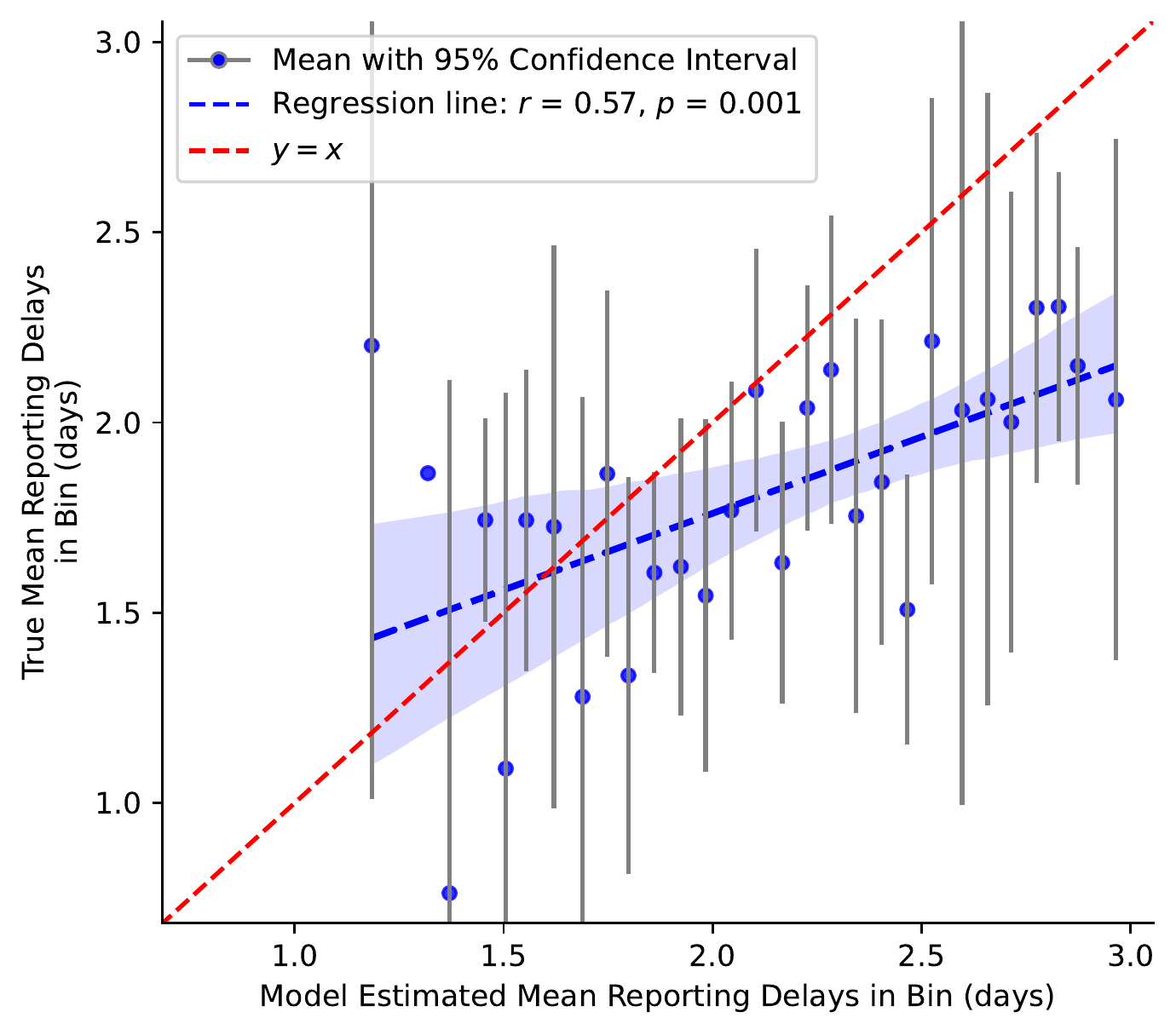}
		% 		\caption{Census tract fixed effects}
		\label{fig:ida6}
		% 		\end{subfigure}
	}
 \hfill
	\caption{Replication of \Cref{fig:isaiasdelay} using data after Tropical Storm Ida, with different end date specifications.}
    \label{fig:ida}
\end{figure}

\begin{figure}[tb]
	\centering
	%	\begin{subfigure}{.5\textwidth}
	% 		\centering
	\subfloat[][Histogram of reporting delay after Isaias.]{
		\includegraphics[width=.45\textwidth]{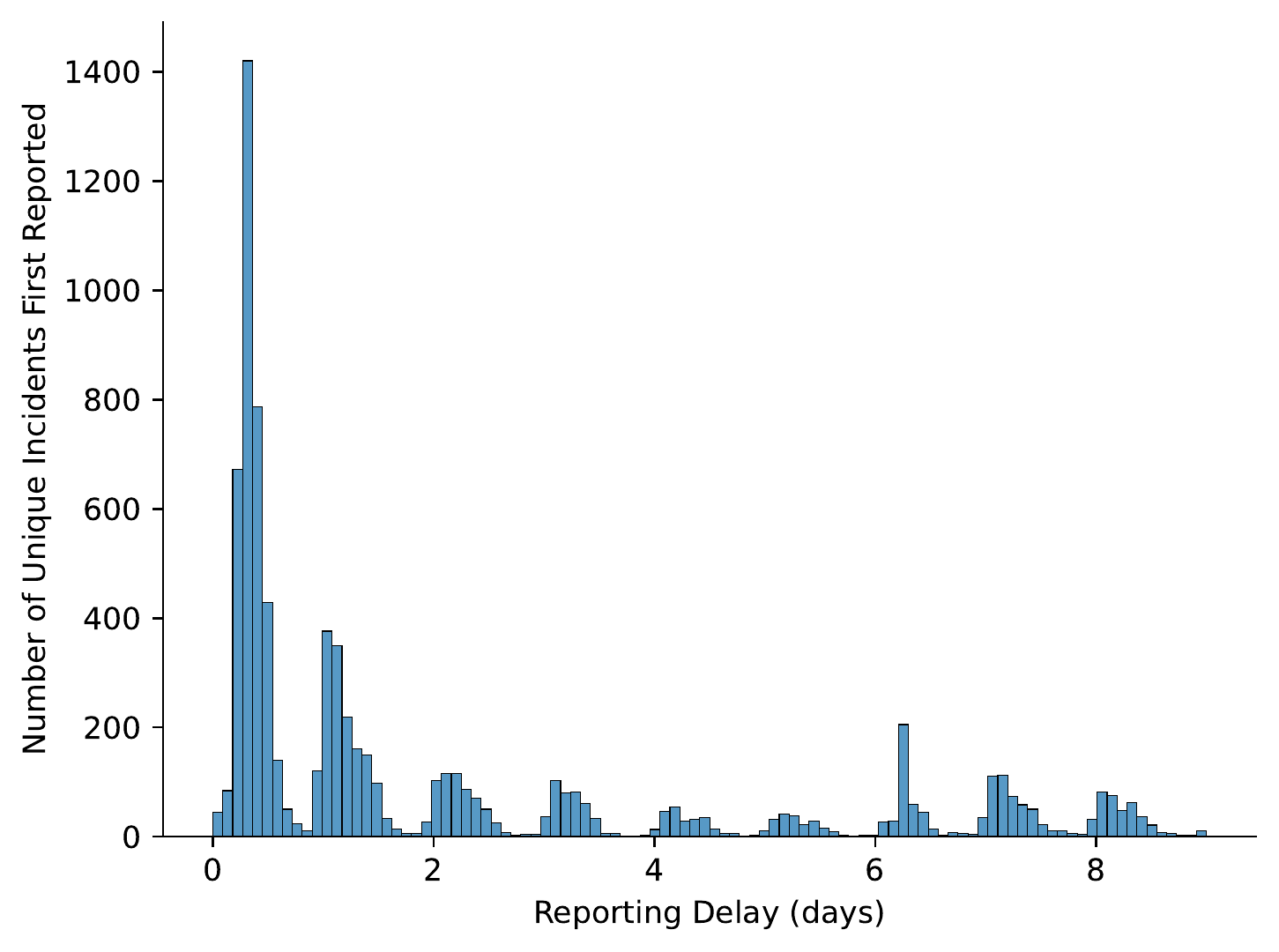}
		% 		\caption{Census tract fixed effects}
		\label{fig:isaiashist}
		% 		\end{subfigure}
	}
	\subfloat[][Histogram of reporting delay after Ida.]{
		\includegraphics[width=.45\textwidth]{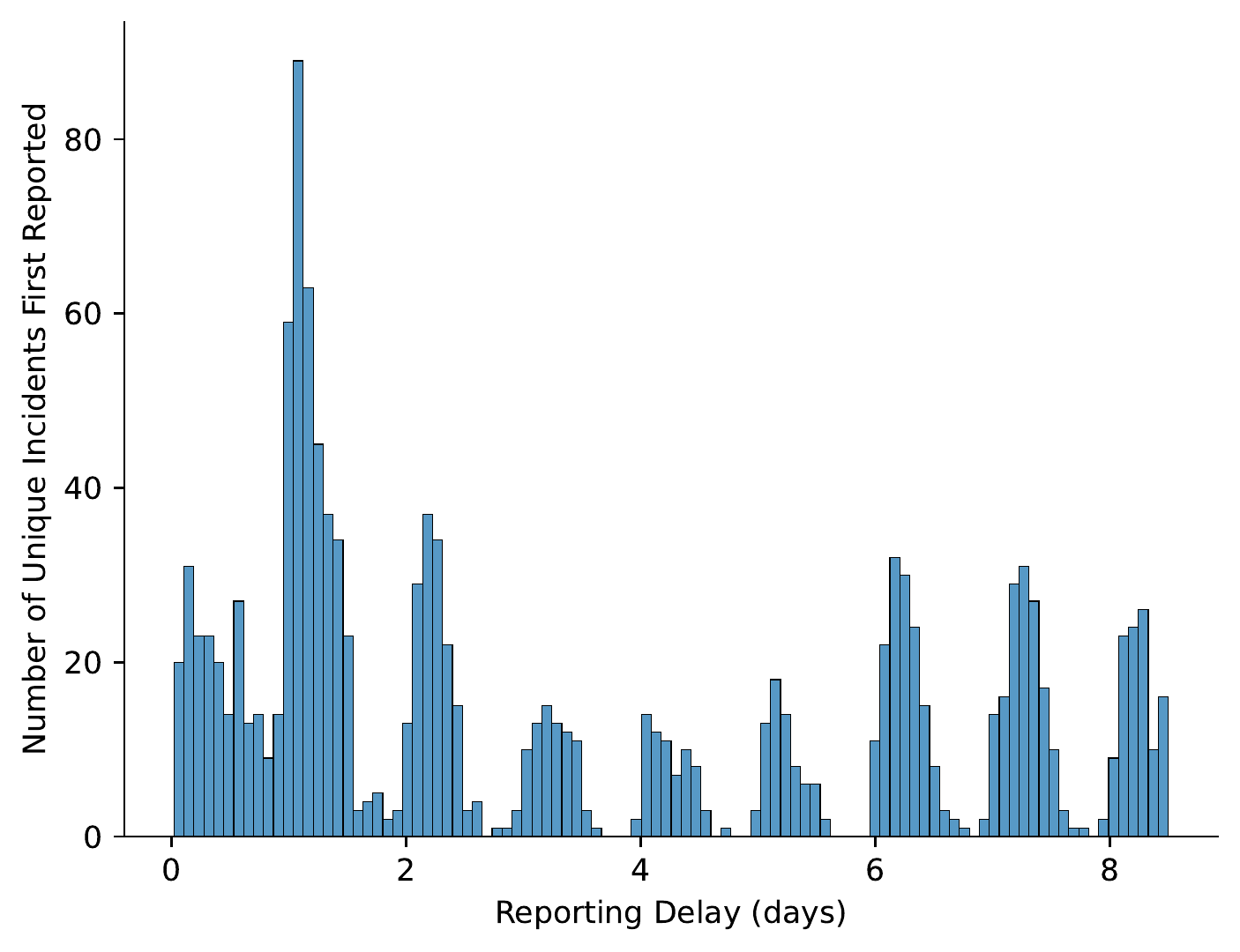}
		% 		\caption{Census tract fixed effects}
		\label{fig:idahist}
		% 		\end{subfigure}
	}
 	\hfill
	\caption{Histogram of reporting delays of unique incidents after Isaias and Ida. After each hurricane, we immediately observe a surge in reported incidents.}
    \label{fig:histdelayhurricane}
\end{figure}

\subsubsection{Supplementary information on using a training set starting from the second report}
\label{app:oosmodel}

In \Cref{sec:validation_hurricanes} and in \Cref{app:hurricane}, the predictions were generated using the results of the spatial model we presented in the main text. The training set used to obtain such results was constructed by defining the time of the first report as the start of our observation period, as detailed in \Cref{app:nycpreprocessing}.

However, our theoretical analysis also holds if we instead use as the Poisson interval the time starting after the \textit{second} report --  if we split a Poisson process on the time of the first jump, and start counting subsequent jumps at that time, the resulting counting process would still be a Poisson process with the same rate. Thus, we can estimate the Poisson process rate using our methods, starting with the second report time.\footnote{Formally, consider a Poisson process with rate $\lambda$, denoted as $\{X(t), t\ge 0\}$, with the first jump time denoted as $T_1$. Based on the memoryless property of Poisson process jump times, the process $\{X(t)-1, t\ge T_1\}$ is still a Poisson process with the same rate parameter (see e.g., \cite{resnick1992adventures}). Thus in our context, for an incident with type $\theta$, if we treat the first report as the start of a Poisson process, treat the second report as equivalent to the first report, and count duplicates starting at the third report, this reporting process is still a Poisson process with rate $\lambda_\theta$, but now the start of this process is observed.} However, crucially, we as researchers know for the incident the \textit{time between the first and second report}, which was not used to train the model (but which the model is trying to predict). Intuitively, this procedure evaluates our methods in a setting in which the ground truth is known.

Formalizing this idea, we further validate our empirical approach as follows. We first construct a training set from the same raw public data introduced in \Cref{app:nycpublic}, with the time of the second report as the start of our observation period, and the end of our observation period defined analogous to \Cref{eq:duration}, but with 20, 30, and 50 days as the maximum duration.\footnote{We note that in \Cref{thm:stoppingtimes}, we require the end of the observation period to be before incident death time. The maximum duration is thus a design choice to alleviate the effect of the incident being resolved before an inspection or a work order. Using the time of the second report as the start of the observation interval requires us to set the maximum duration at a smaller value compared to using the time of the first report, in order to have the equivalent effect. Furthermore, incidents with two or more reports are likely incidents with higher than average urgency, and thus more likely to be resolved than outside sources, which further justifies the choice of shorter maximum duration.} We then use this alternative training set to estimate the coefficients for the set of \textbf{Base} covariates, and then use these coefficients to estimate a reporting delay for each incident. The reason we only included the \textbf{Base} covariates is that the training set constructed in this way only retains 16,460 incidents, which makes it hard to estimate the effect of high-dimensional covariates, such as the census tract spatial variables. Finally, we compare these estimates to the observed time between the first and second reports for each incident.

On the individual incident level, our estimates and the true delays between the first and second reports are significantly correlated (Pearson correlation $r = 0.28, 0.28, 0.29$ for a maximum duration of 20 days, 30 days, and 50 days, respectively, all with $p<.001$). \Cref{fig:oosmodeldelays} showcase a comparison of the means when all incidents are binned according to their model-estimated delays. We note that though the choice of the maximum duration of observation interval affects the specific model predictions, all three choices produced reasonable estimates that are close to the true delays on the group level. (We use a smaller maximum duration as these incidents are on average higher risk, and we only start counting after the second report). 

\begin{figure}[tb]
	\centering
	%	\begin{subfigure}{.5\textwidth}
	% 		\centering
	\subfloat[][Maximum duration 20 days.]{
		\includegraphics[width=.32\textwidth]{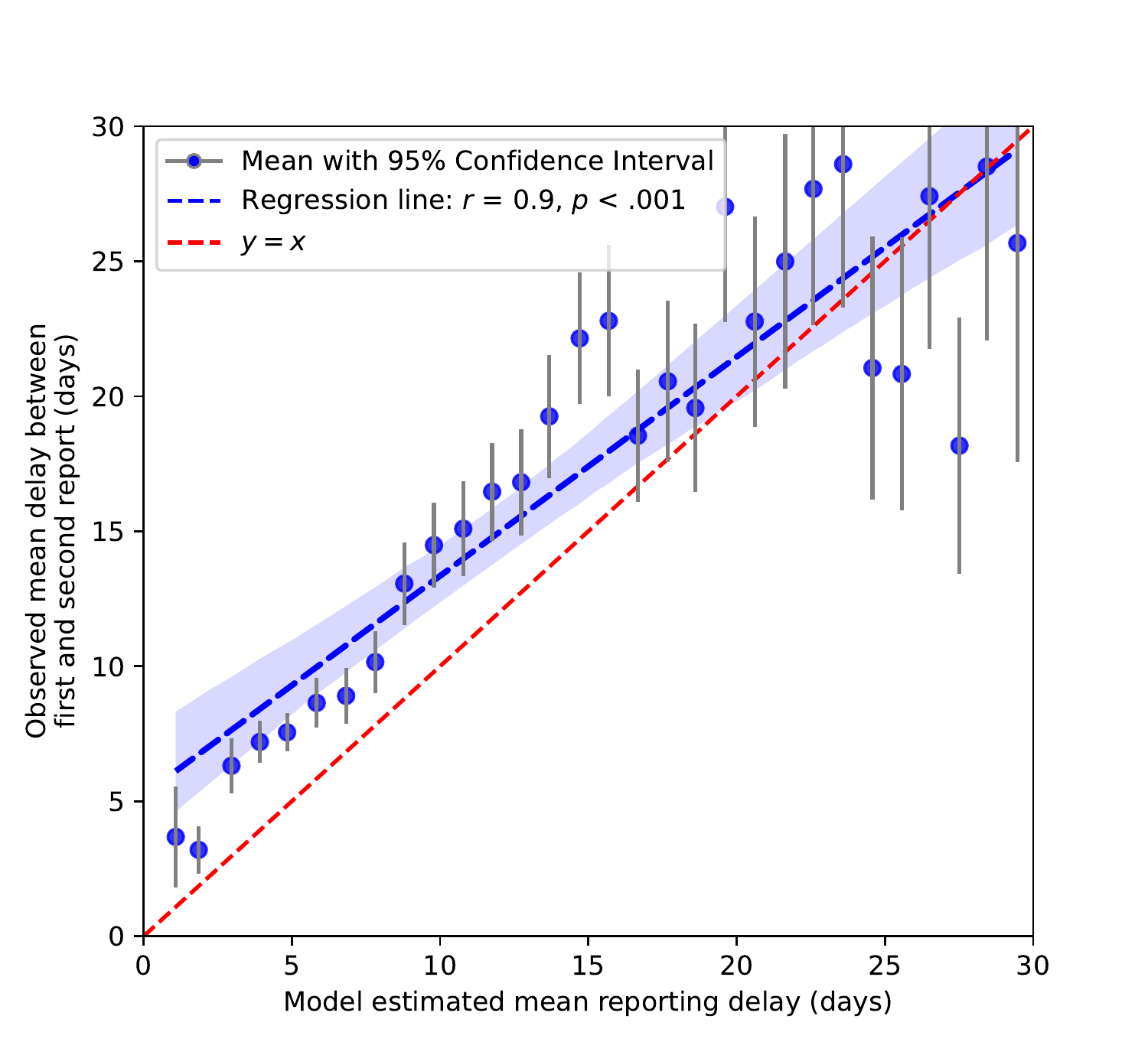}
		% 		\caption{Census tract fixed effects}
		\label{fig:20daysoosmodel}
		% 		\end{subfigure}
	}
	\subfloat[][Maximum duration 30 days.]{
		\includegraphics[width=.32\textwidth]{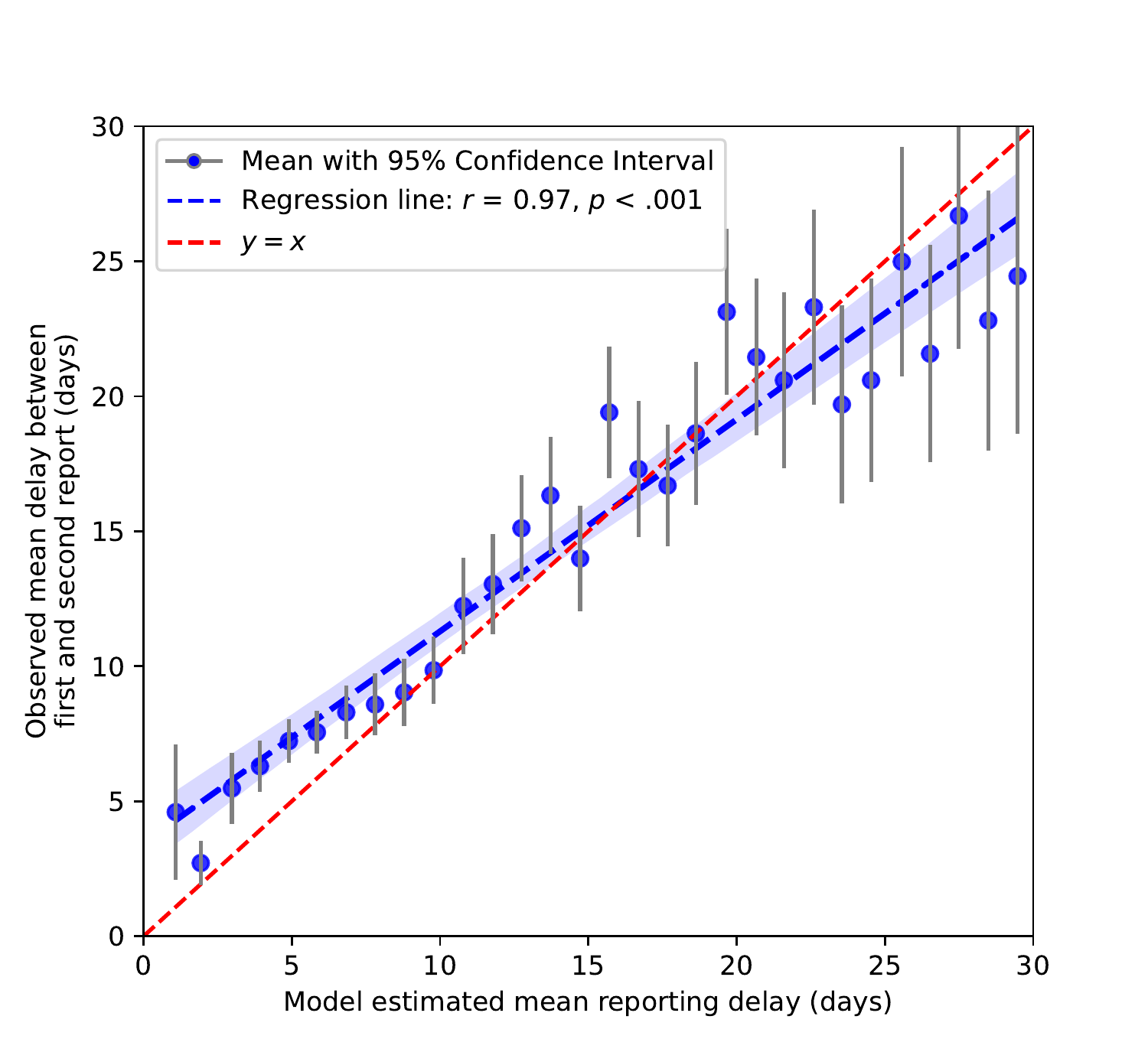}
		% 		\caption{Census tract fixed effects}
		\label{fig:30daysoosmodel}
		% 		\end{subfigure}
	}
 \subfloat[][Maximum duration 50 days.]{
		\includegraphics[width=.32\textwidth]{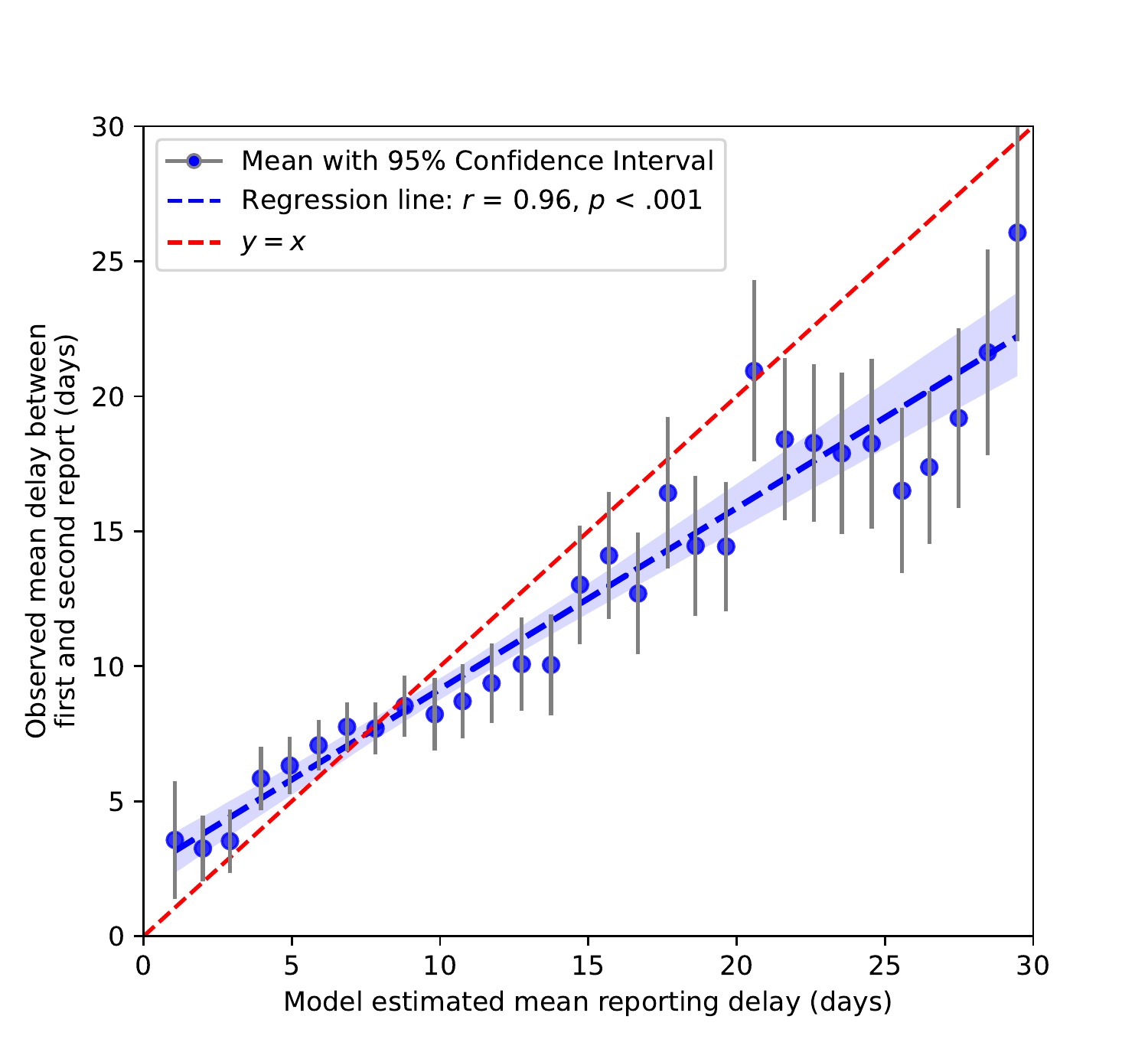}
		% 		\caption{Census tract fixed effects}
		\label{fig:50daysoosmodel}
		% 		\end{subfigure}
	}
 	\hfill
	\caption{Comparison of observed delays between first and second reports and model-estimated reporting delays, which are estimated using a model trained with data constructed by defining the start of the observation interval as the time of the second report, and the end as in \Cref{eq:duration}, but with maximum duration set as 20 days, 30 days, and 50 days, respectively. All incidents are then categorized into 30 bins based on their model-estimated reporting delays, and we calculate the means of true and estimated delays within each bin.}
    \label{fig:oosmodeldelays}
\end{figure}

\subsubsection{Supplementary validation of census tract socioeconomic coefficients estimates}
\label{app:voterparticipation}

The disparities in reporting behavior along socioeconomic variables highlight \textit{individual-level} behavioral heterogeneity in resident crowdsourcing, and civic engagement at large. The level of civic engagement can be measured by many other means, chief among which is participation in political voting. In this section, we validate our coefficient estimates on socioeconomic variables using this idea. 

Voter-level public data on participation in voting is generally scarce in New York; however, the NYC Campaign Finance Board published a dataset containing participation scores of more than 4 million voters calculated based on their voting history from 2008 up to 2018, along with the census tract in which they resided in 2010.\footnote{\url{https://data.cityofnewyork.us/City-Government/Voter-Analysis-2008-2018/psx2-aqx3}} We calculate cumulative association scores (the same analysis as in the new main text \Cref{fig:spatialheterogproj}\footnote{Since we are focusing on individual-level behavior, for that figure and this analysis we do not use the density variable, as that does not reflect individual-level characteristics, and only use the remaining five: Median age, Fraction white, Fraction college degree, Fraction renter, log of per capita income}) for each 2010 census tract, using the corresponding census data. We compare these scores with the average voter participation score in each tract. At the census tract level, our model-estimated mean engagement score is significantly correlated with the mean voter participation score (Pearson $r = 0.33$, $p<.001$). \Cref{fig:voterparticipation} shows the full relationship, between binned cumulative association scores and voter participation rates. A significant positive correlation shows that our estimated disparities in 311 system usage relate to other forms of civic participation and representation.

\begin{figure}
    \centering
    \includegraphics[width = 0.5\textwidth]{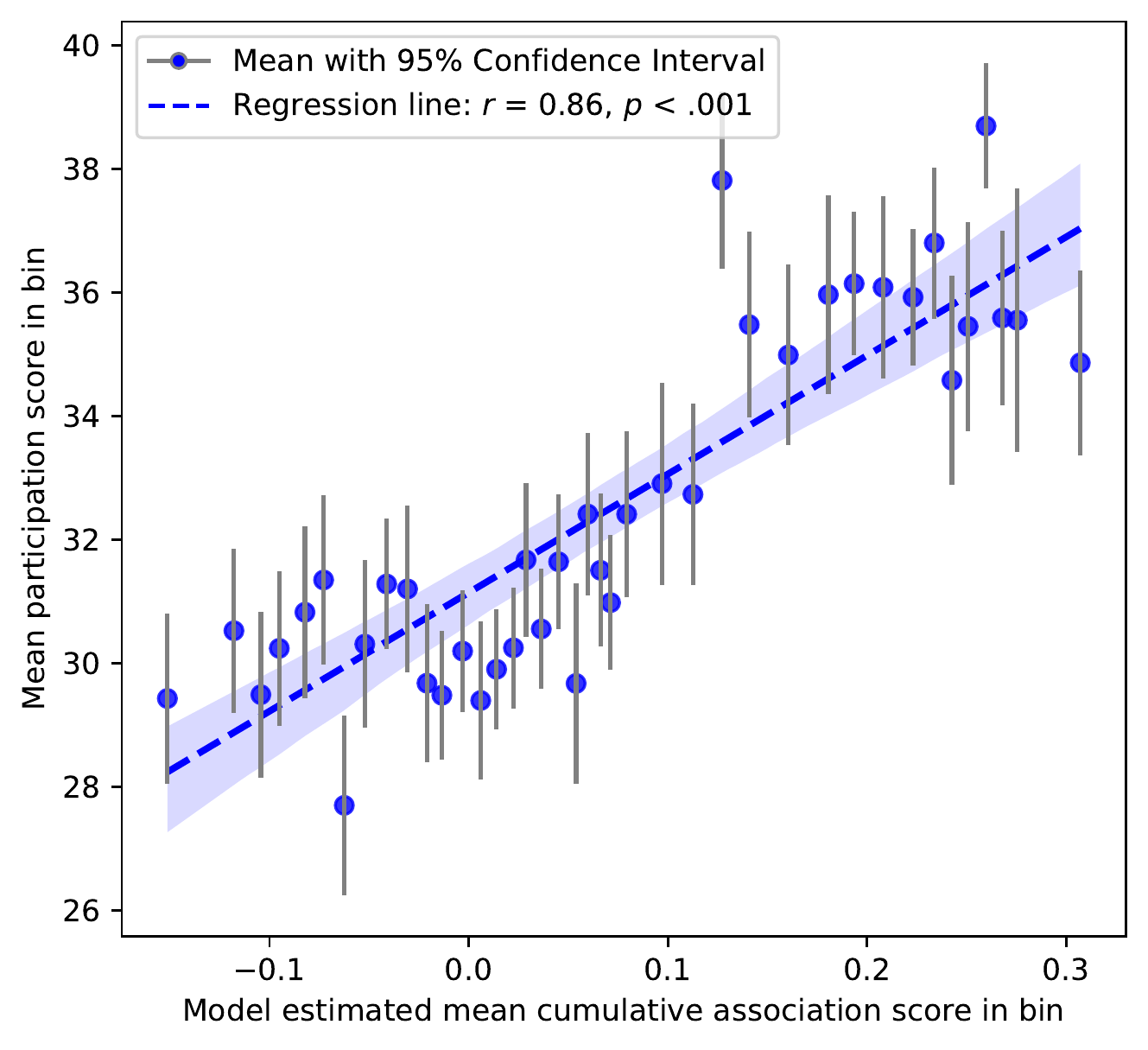}
    \caption{Relationship between voter participation rates in a census tract, and our model estimated cumulative association of socioeconomic variables with reporting rate. A significant positive correlation between these two measures affirms that our model estimates are uncovering inequities in reporting as a type of civic engagement.}
    \label{fig:voterparticipation}
\end{figure}

\FloatBarrier
 \FloatBarrier
\subsection{Additional information for comparative delay analysis}
\label{app:alternativeimpactanalyis}

Similar to \Cref{fig:delaysall_raw}, we present results with alternative analysis choices: imputing missing values (with infinite delays), other risk levels, analyzing by request category instead of risk prioritization level, and analyzing all incidents instead of the subset that was inspected or worked on; results are qualitatively similar.

For each unique incident, inspection delays are measured as the time between the first inspection (if it was inspected) and the first service request for that incident; work order delays are measured as the time between actually finishing the work order (if it was completed) and first inspection; reporting delays are \textit{estimated} -- for each incident, we estimate the mean of an Exponential random variable with reporting rate our model estimates for an incident of the given characteristics. In this plot, we only consider incidents that are inspected, and ignore missing values; in each Borough over 89\% of such high-risk inspected incidents eventually have a completed work order.

Similar to \Cref{fig:delaysall_equity}, we present results with other risk prioritization groups.

% A but imputed with missing
  \begin{figure}[bth]
  \vspace{-0.2cm}
 	\centering
 		\includegraphics[width=.48\textwidth]{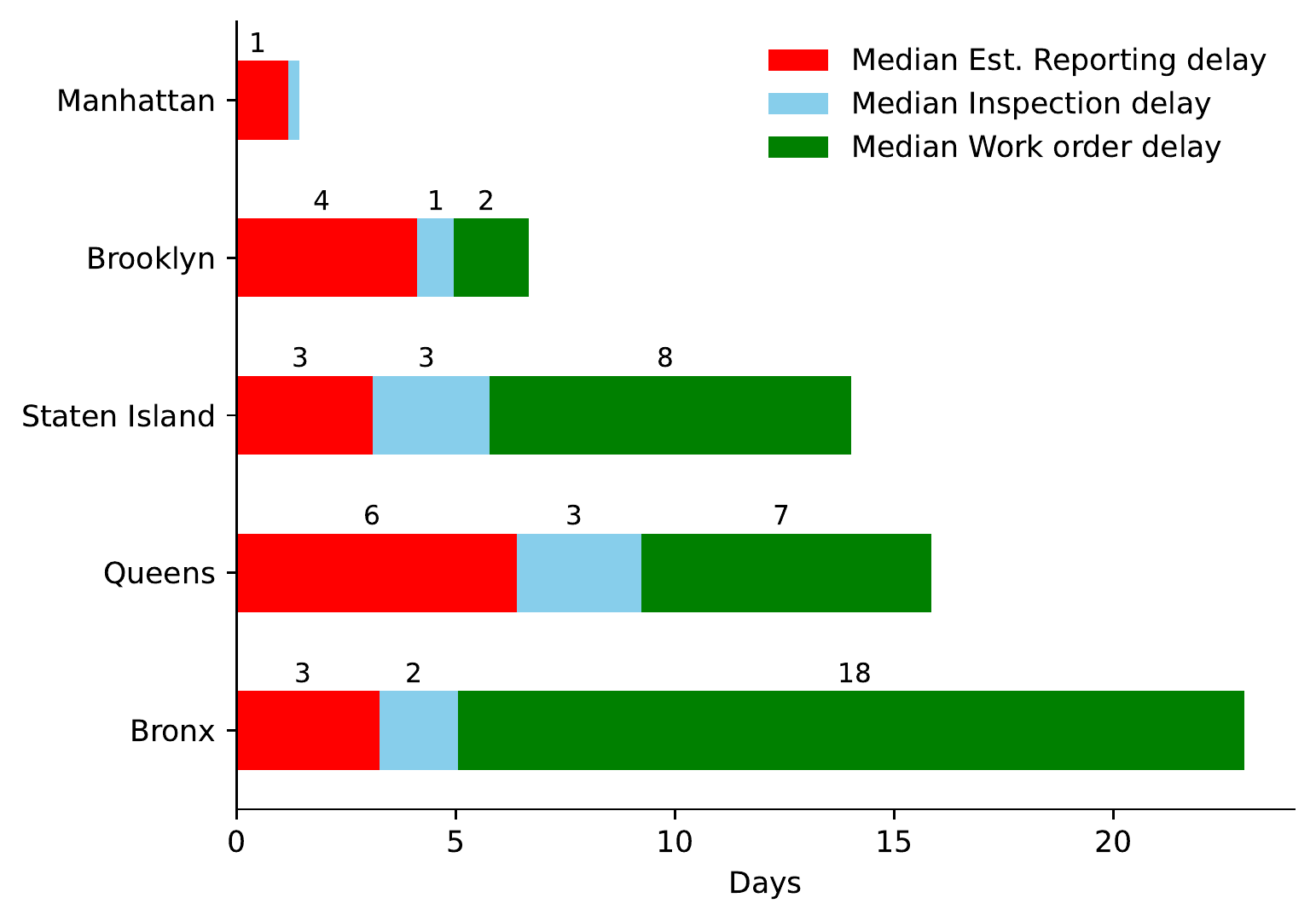}
 	\caption{Same as \Cref{fig:delaysall} (highest risk prioritization \textit{A}) except with incomplete work orders imputed as having infinite delays}
  \label{fig:riskAImputed}
 \end{figure}

%all other risk levels, not imputed

  \begin{figure}[bth]
  \vspace{-0.2cm}
 	\centering
   \subfloat[][Risk prioritization \textit{B}]{
 		\includegraphics[width=.45\textwidth]{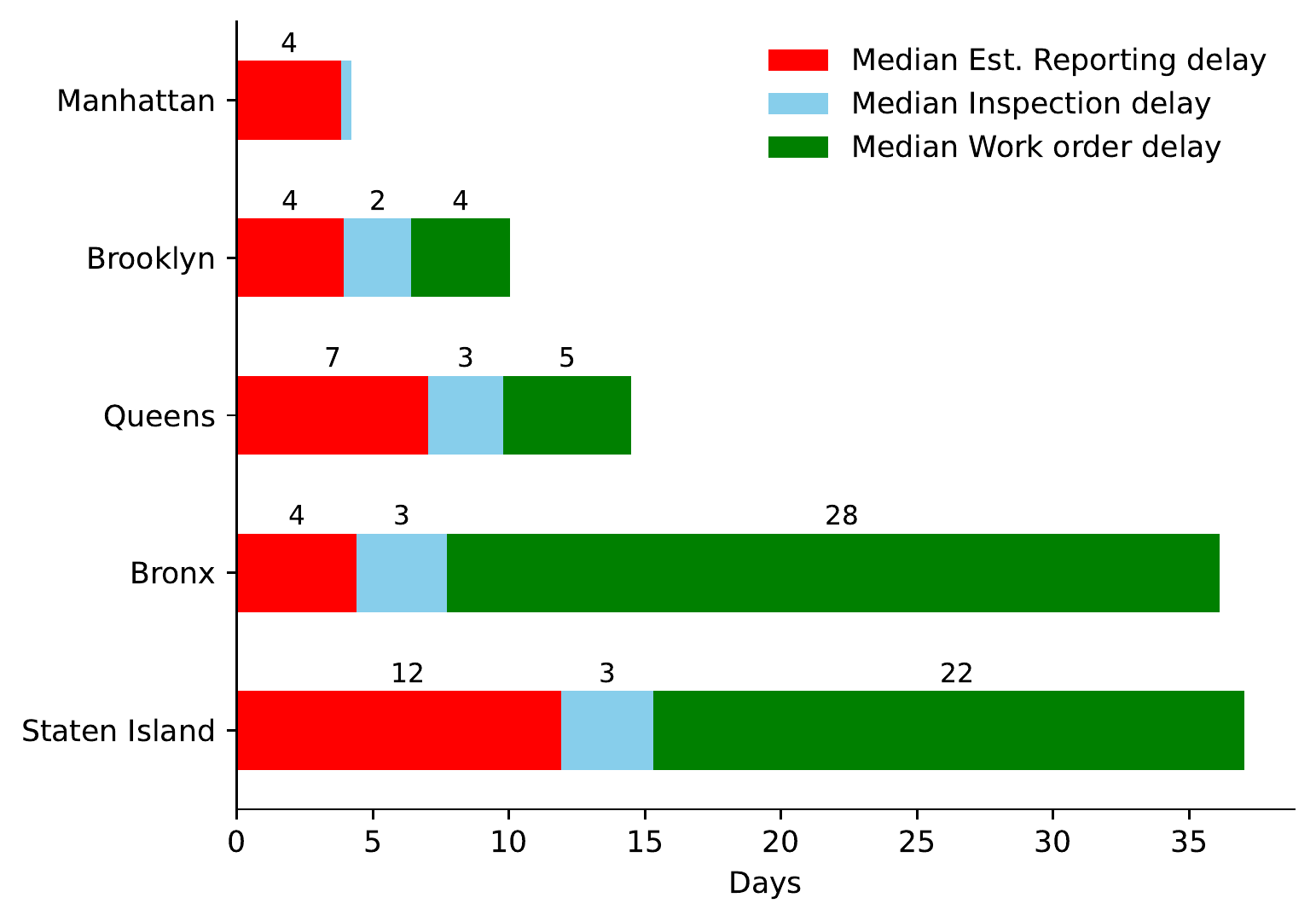}
   }
      \subfloat[][Risk prioritization \textit{C}]{
 		\includegraphics[width=.45\textwidth]{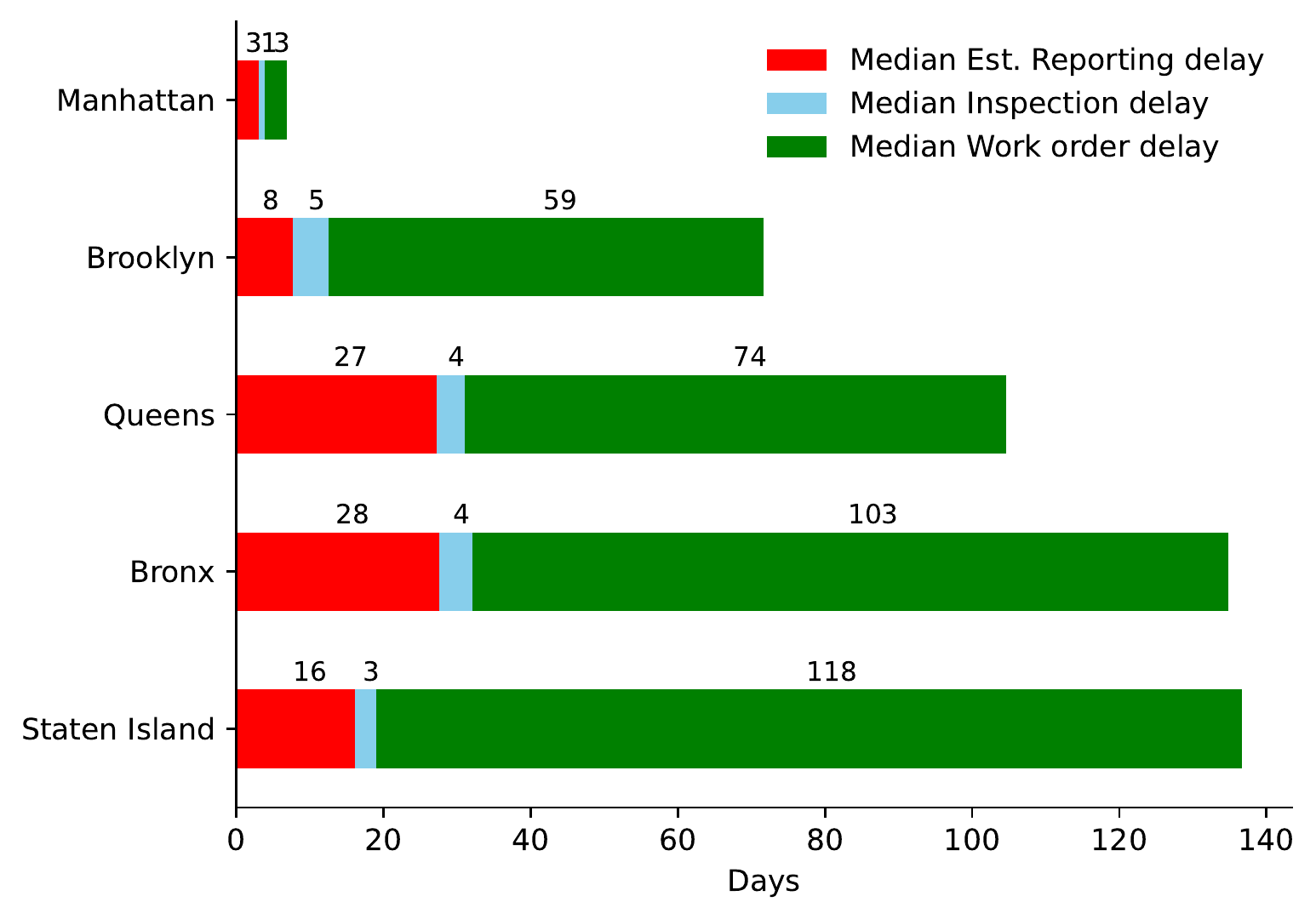}
   }
   %    \subfloat[][Risk prioritization \textit{D}]{
 		% \includegraphics[width=.3\textwidth]{plots/impact_analysis/publicdatarisk_risk_notimputed_D_Borough.pdf}
   % }
 	\caption{Same as \Cref{fig:delaysall} (highest risk prioritization, not imputed missing values) except for other risk prioritization levels. Risk prioritization group D is omitted due to a low percentage of incidents that eventually got worked on.}
 \end{figure}

% By hazards, also looking at things addressed -- and not imputing; showing things addressed

  \begin{figure}[bth]
  \vspace{-0.2cm}
 	\centering
   \subfloat[][\textit{Hazard} delays]{
 		\includegraphics[width=.32\textwidth]{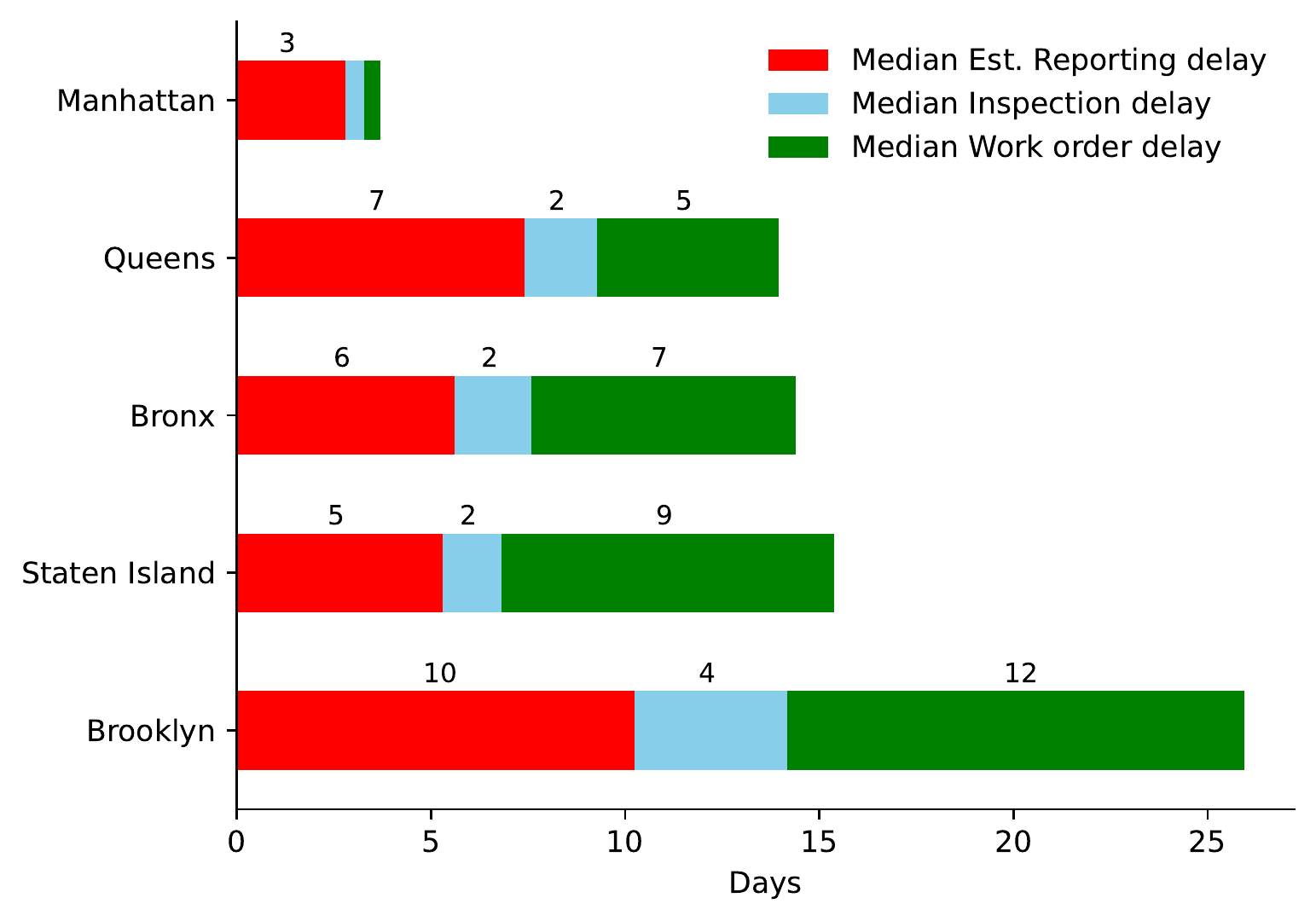}
   }
      \subfloat[][\textit{Hazard} fraction addressed]{
 		\includegraphics[width=.32\textwidth]{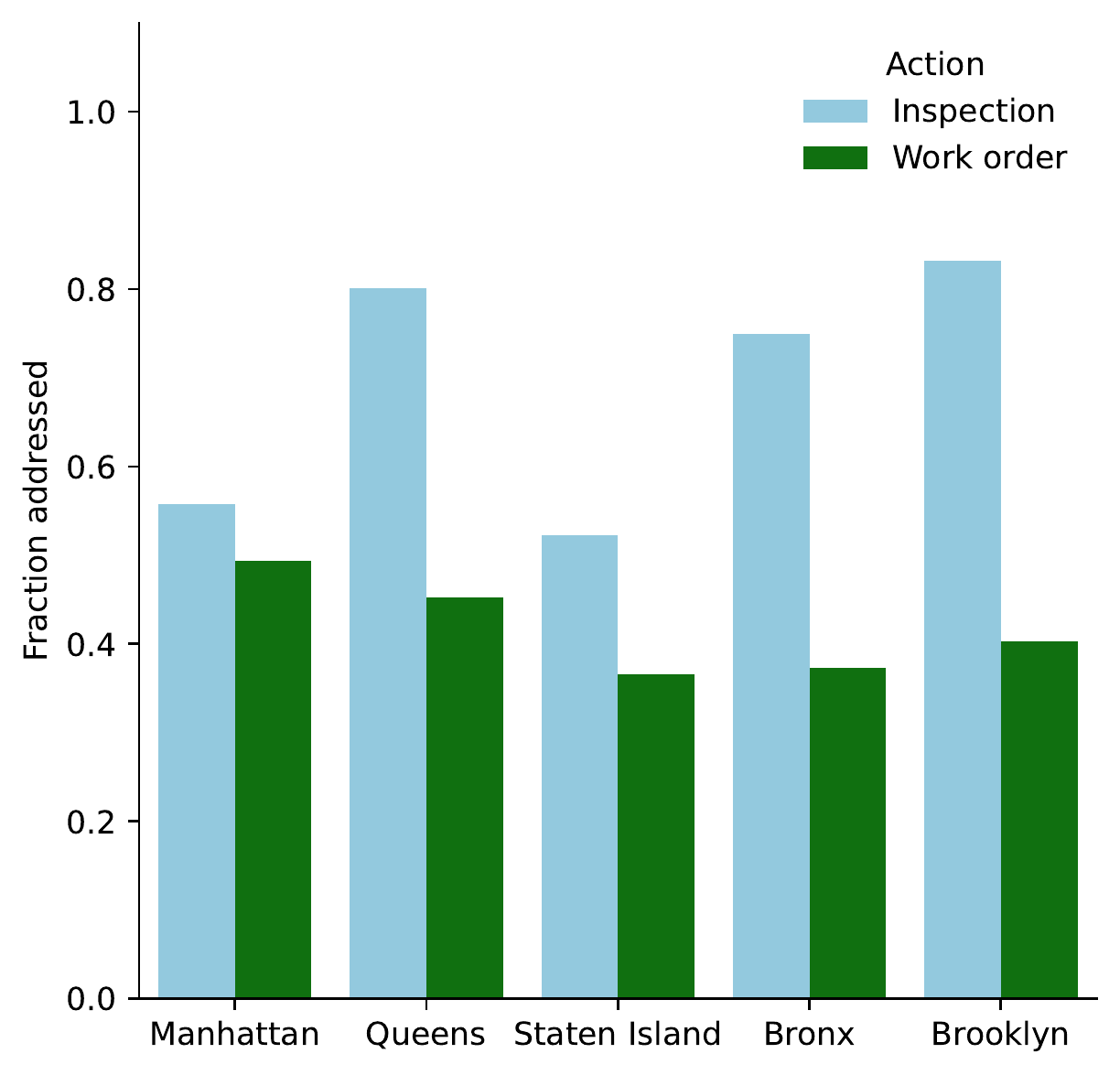}
   }
 	\caption{We repeat the analysis in \Cref{sec:conclusion} except, instead of separating out by \textit{Risk prioritization level}, we separately analyze each \textit{report-time Service Request Category} (here, we show the \textit{Hazard} category. Since we have this feature for \textit{all service requests} (not just those inspected, as in the case where we use the risk prioritization covariate), we can (1) estimate reporting delays for \textit{all} incidents after training a model using just report-time incidents (but on the inspected set, due to labeling of duplicates; and (2) we can also report the fraction of incidents of this category that were inspected and worked on, respectively, as we do in (b). We get qualitatively similar results as in \Cref{sec:conclusion}: the ordering of Borough end-to-end delays (as well as their individual parts) is about the same. The same Boroughs with higher work order delays just counting the incidents that were actually worked on (Staten Island, Bronx, Brooklyn) also have lower fractions of incidents that were addressed. Figure (a) does not impute missing values, since (b) reports the fraction inspected and worked on, respectively. }
  \label{fig:appendixHazarddelays}
 \end{figure}

% same addressed plot for all other types as well

   \begin{figure}[bth]
  \vspace{-0.2cm}
 	\centering
   \subfloat[][\textit{Illegal tree damage} delays]{
 		\includegraphics[width=.3\textwidth]{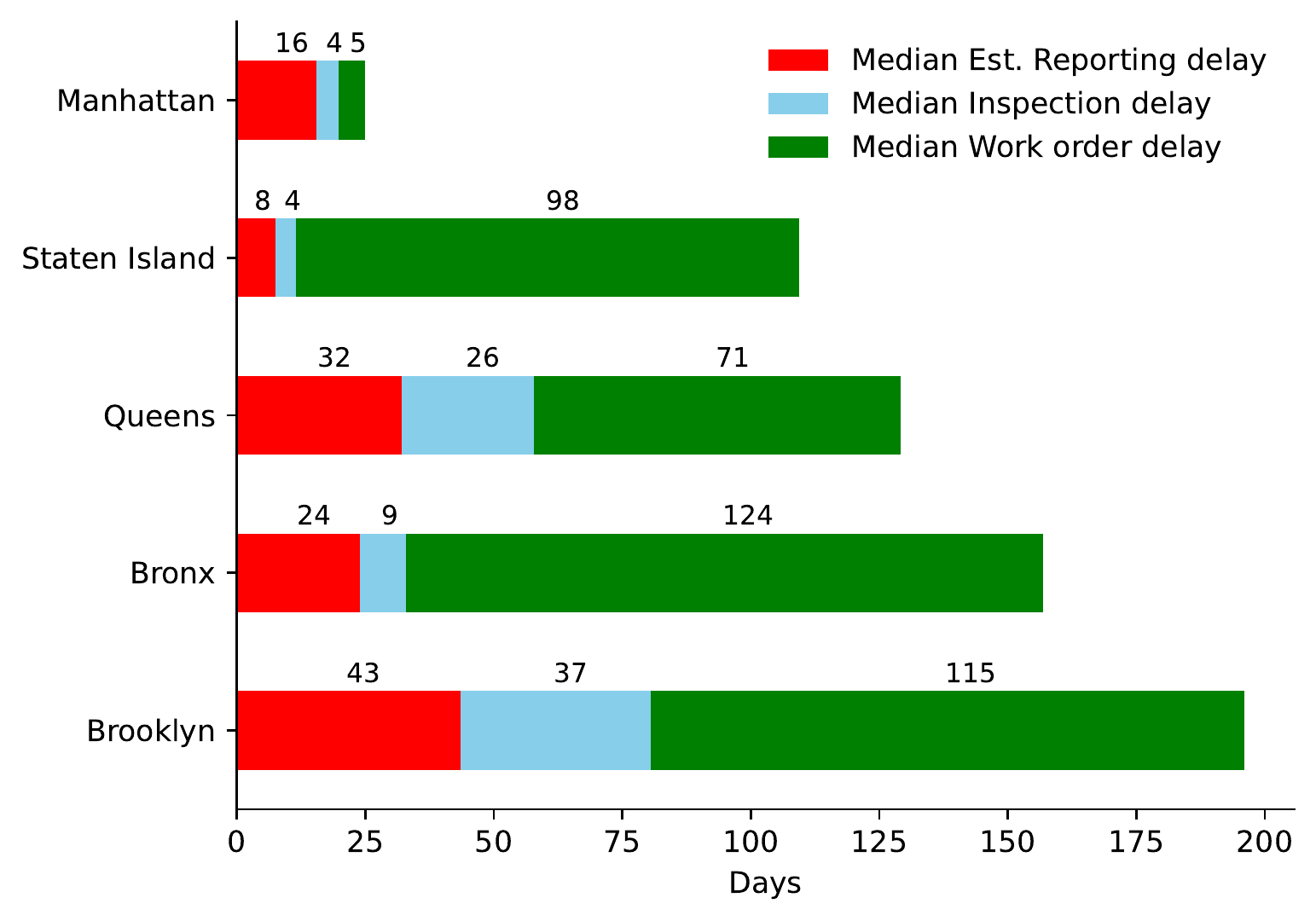}
   }
      \subfloat[][\textit{Illegal tree damage} fraction addressed]{
 		\includegraphics[width=.3\textwidth]{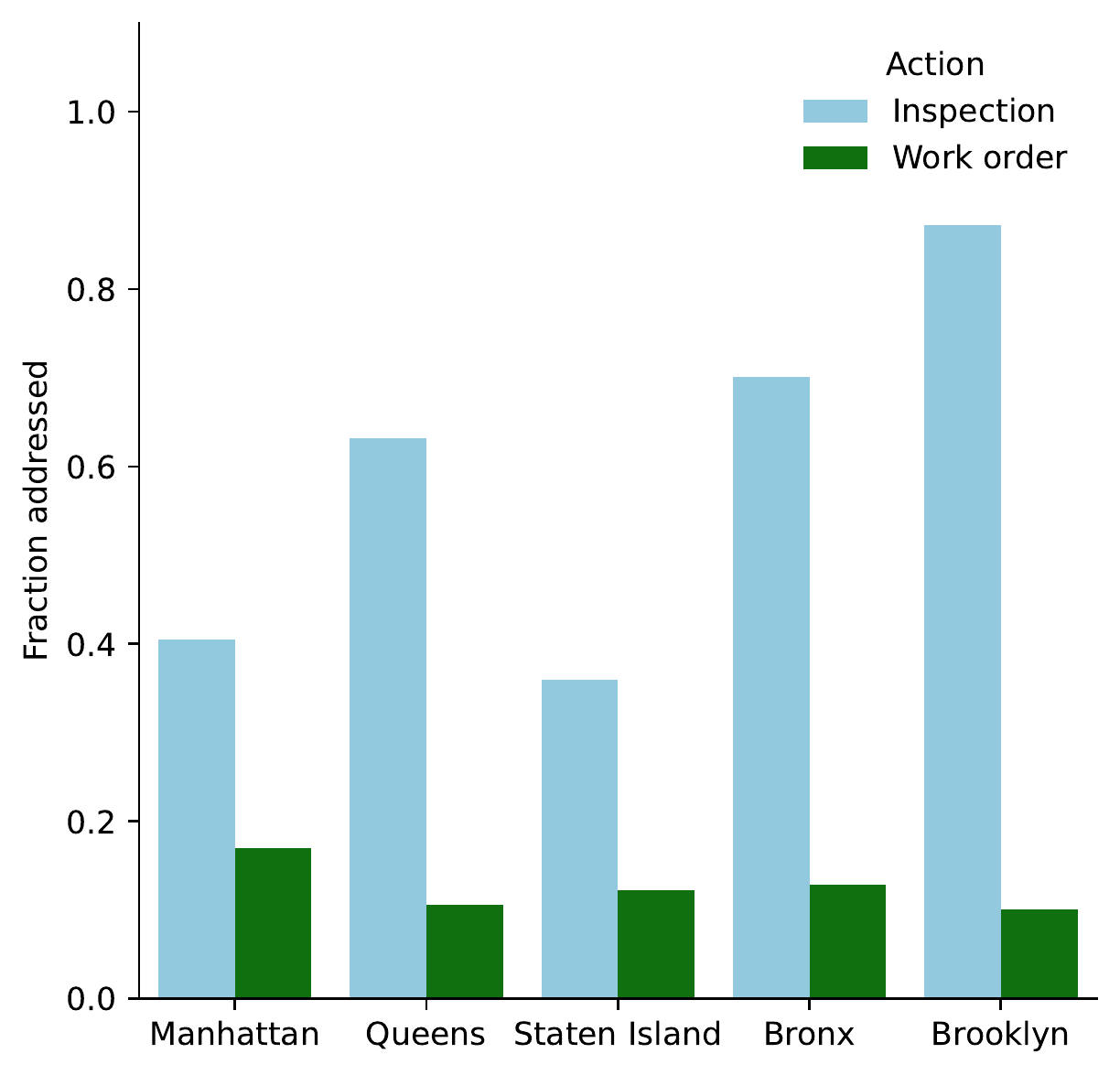}
   }

      \subfloat[][\textit{Prune} delays]{
 		\includegraphics[width=.3\textwidth]{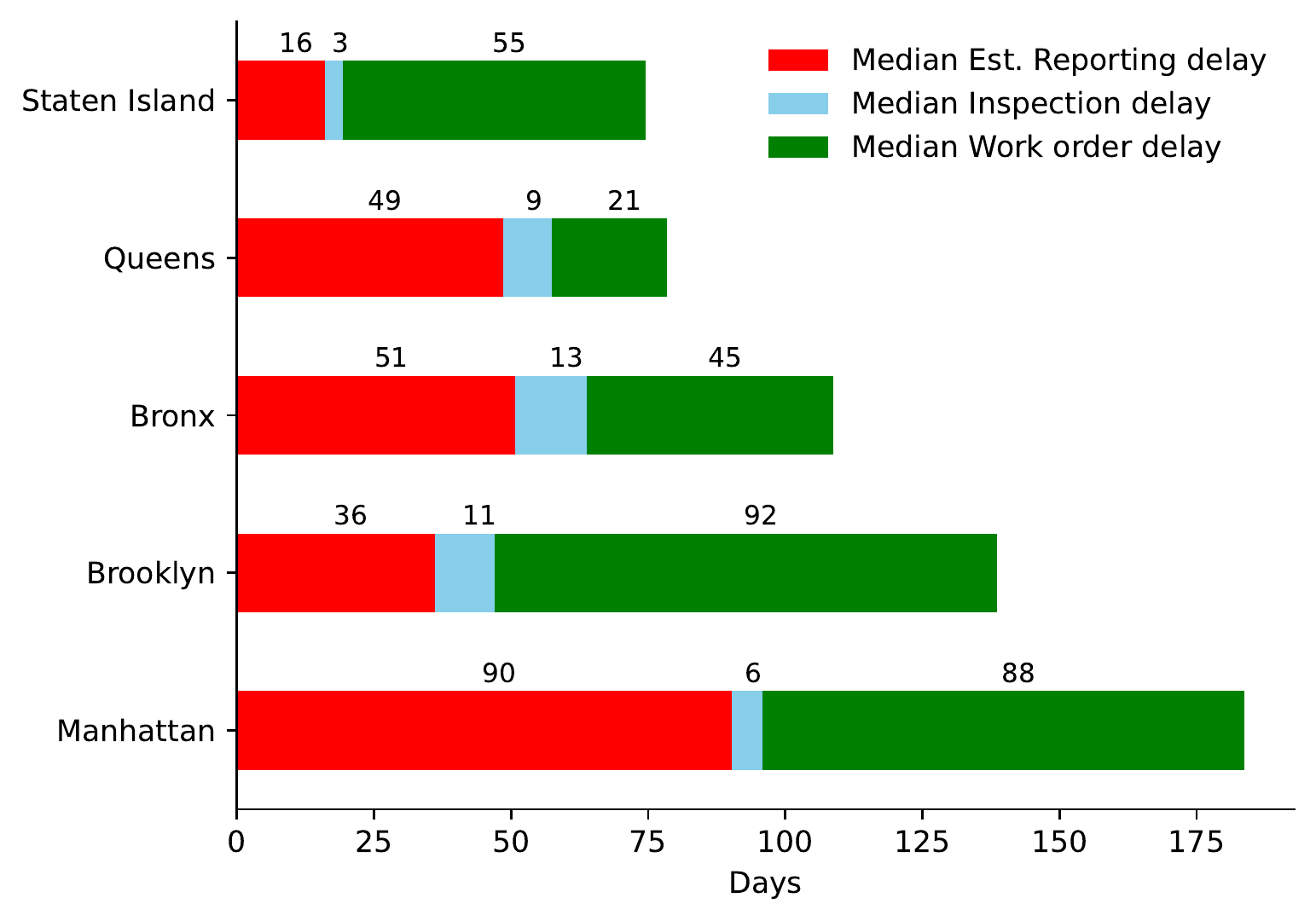}
   }
      \subfloat[][\textit{Prune} fraction addressed]{
 		\includegraphics[width=.3\textwidth]{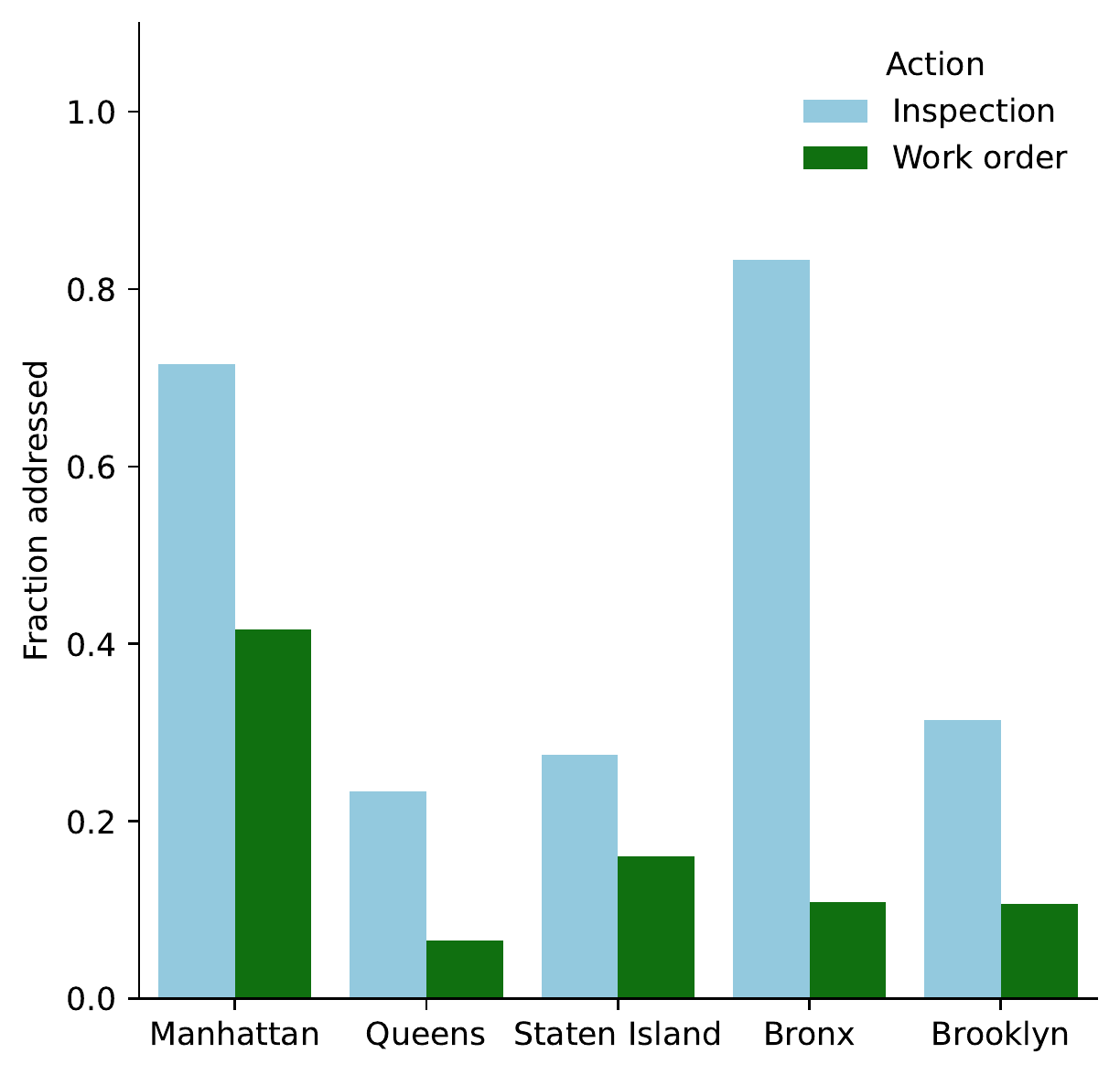}
   }

      \subfloat[][\textit{Remove Tree} delays]{
 		\includegraphics[width=.3\textwidth]{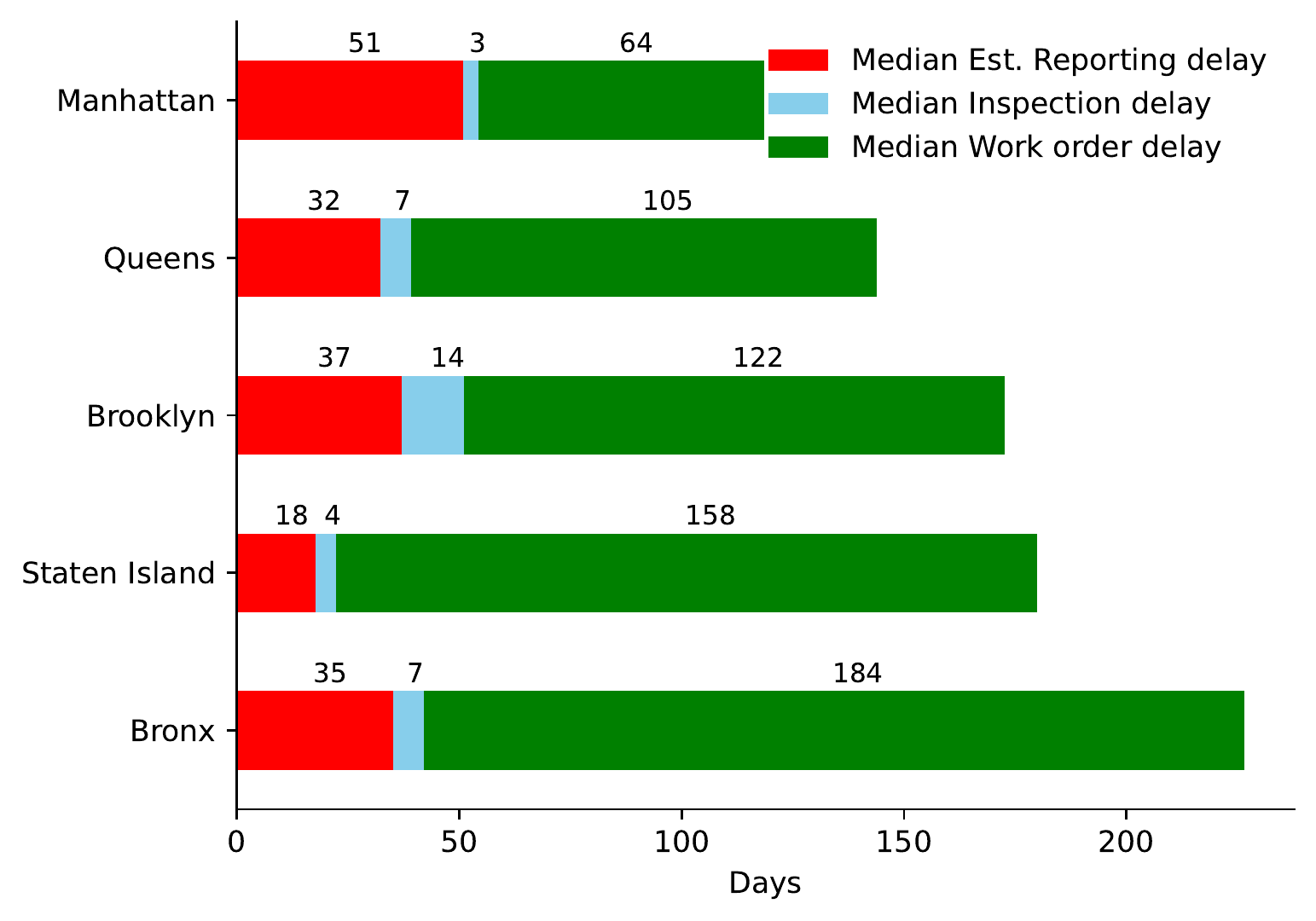}
   }
      \subfloat[][\textit{Remove Tree} fraction addressed]{
 		\includegraphics[width=.3\textwidth]{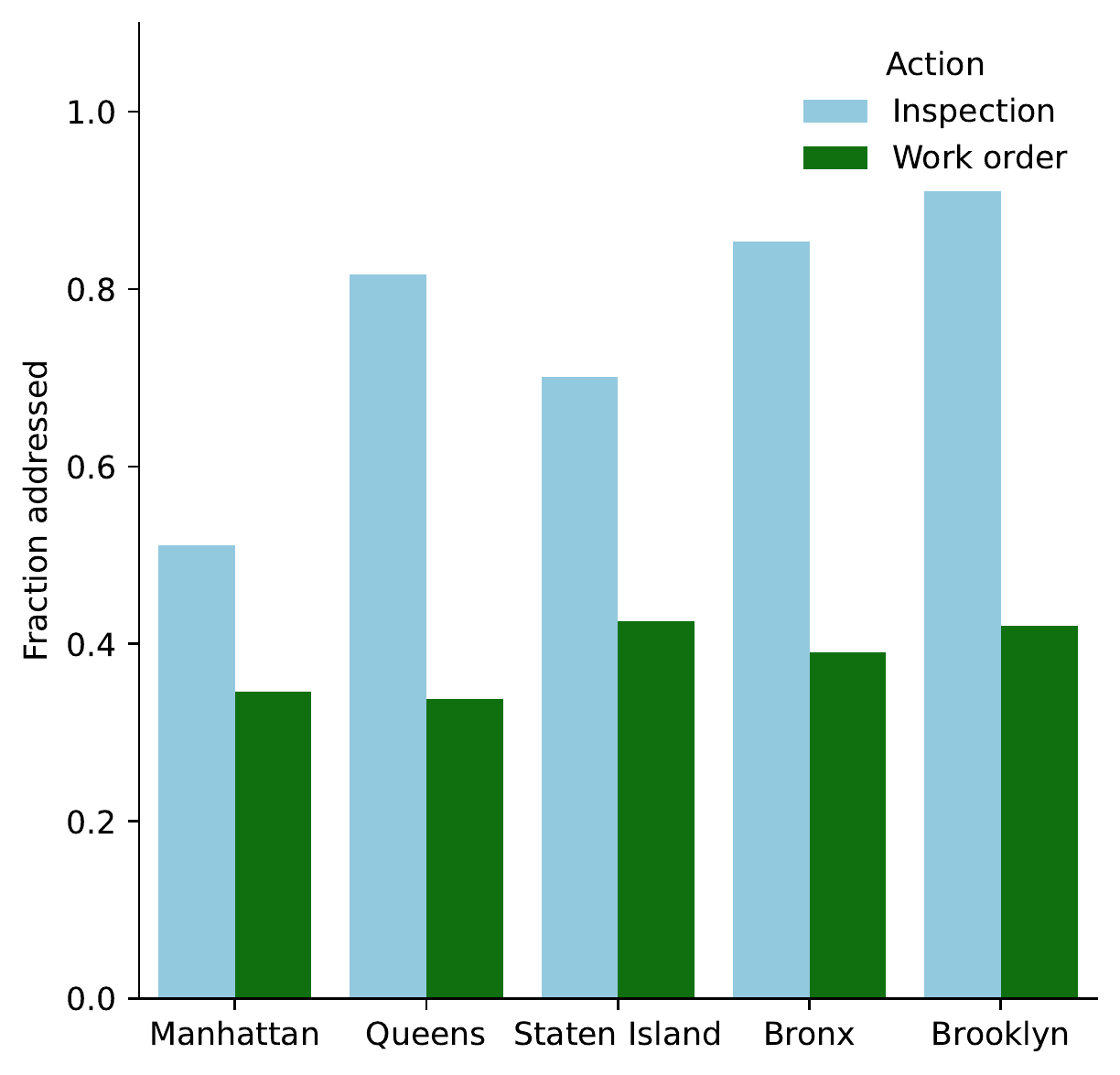}
   }

      \subfloat[][\textit{Root/Sewer/Sidewalk} delays]{
 		\includegraphics[width=.3\textwidth]{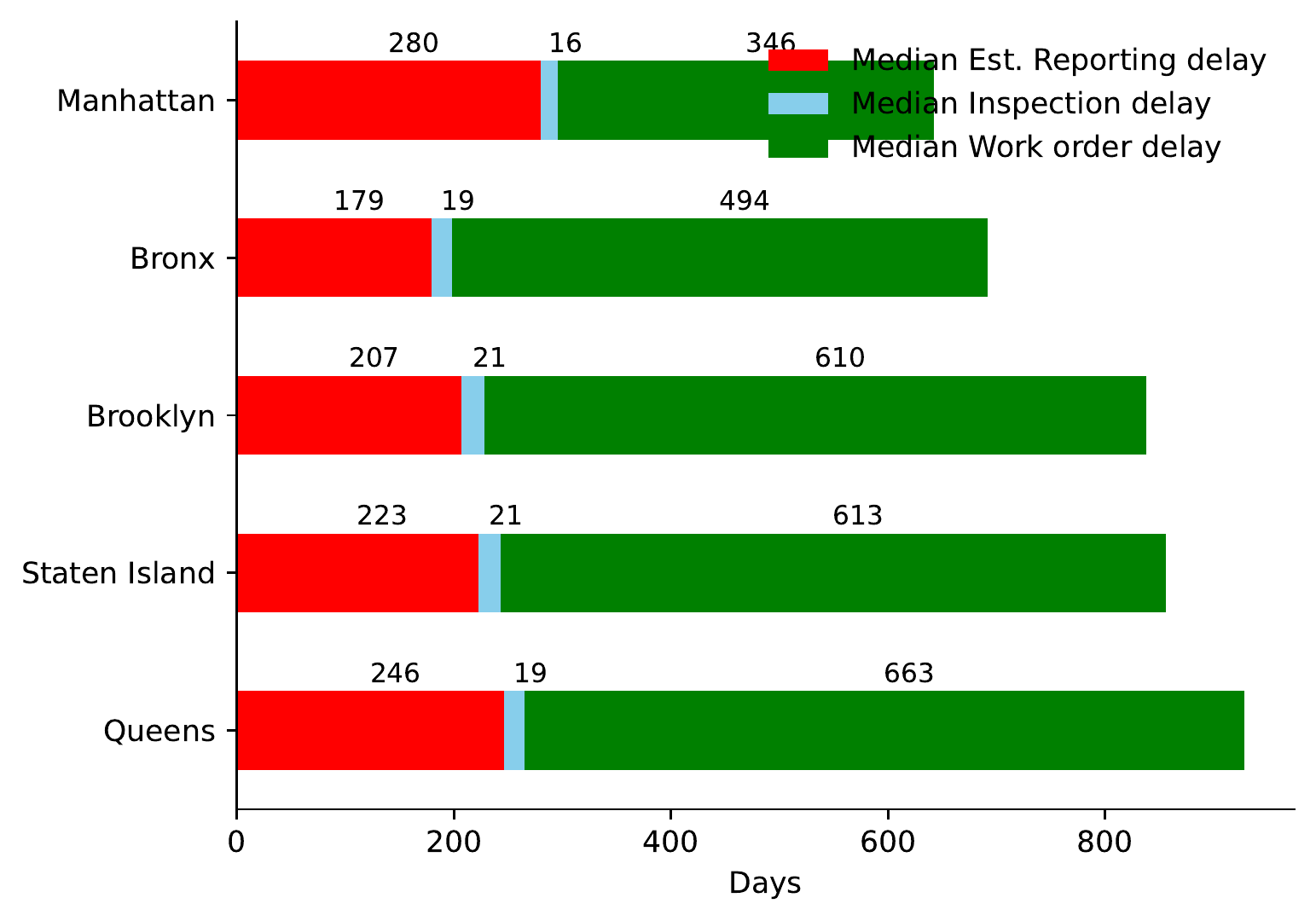}
   }
      \subfloat[][\textit{Root/Sewer/Sidewalk} fraction addressed]{
 		\includegraphics[width=.3\textwidth]{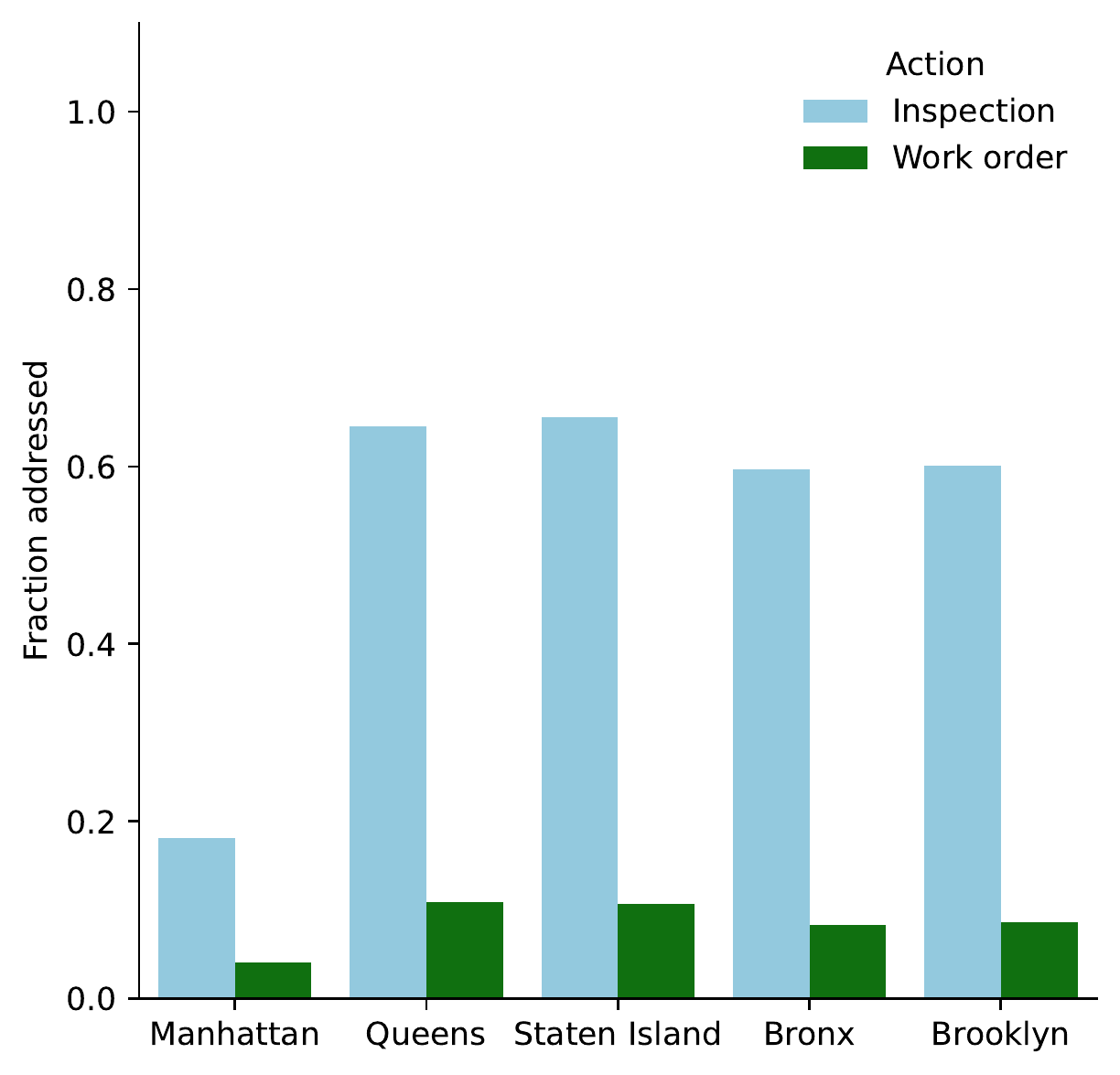}
   }
 	\caption{Same as \Cref{fig:appendixHazarddelays} with other \textit{Service Reqest Categories}.}
  \label{fig:appendixothercatsdelays}
 \end{figure}

\begin{figure}[bth]
  \vspace{-0.2cm}
 	\centering
   \subfloat[][Risk prioritization \textit{B}]{
 		\includegraphics[width=.45\textwidth]{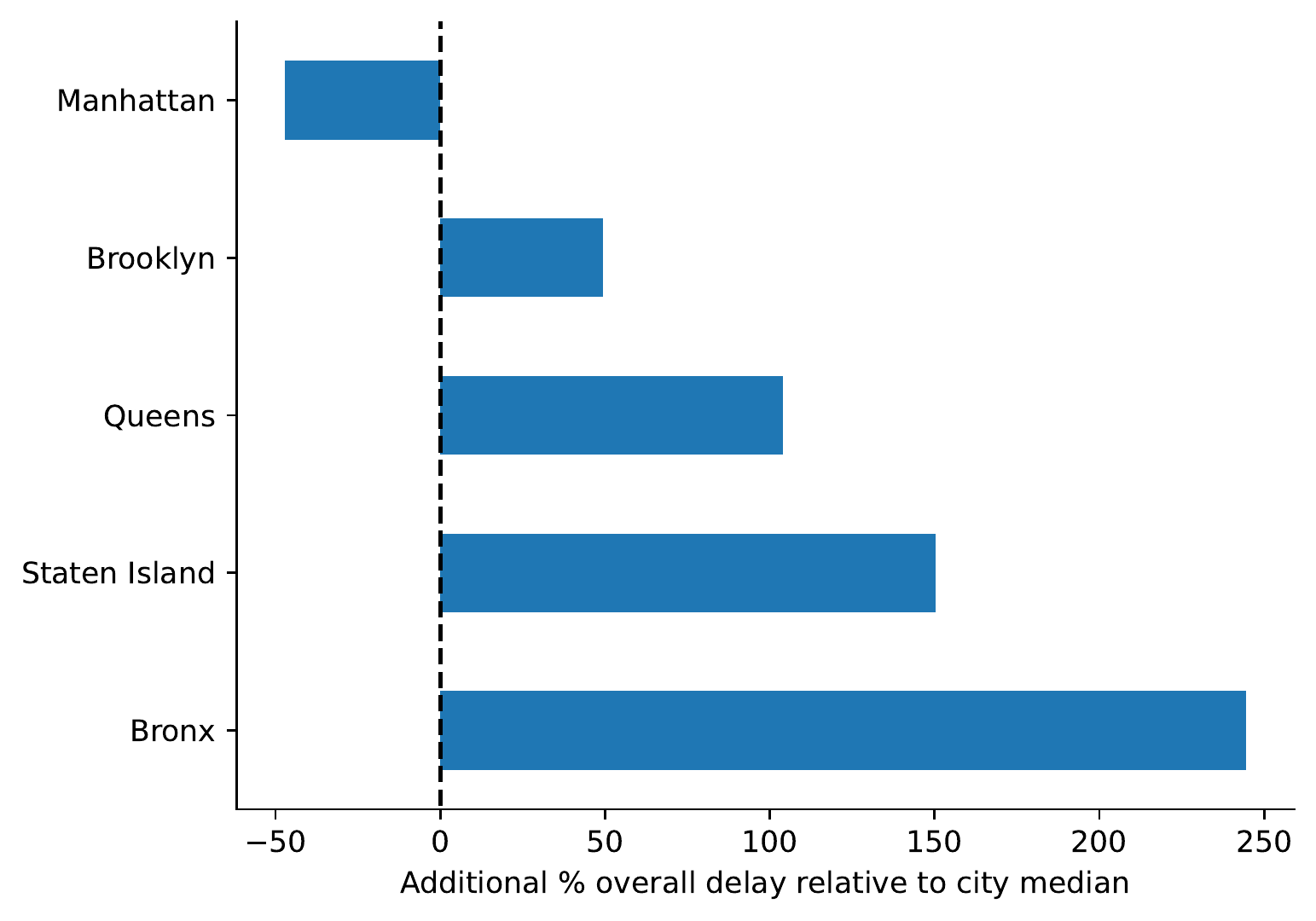}
   }
      \subfloat[][Risk prioritization \textit{C}]{
 		\includegraphics[width=.45\textwidth]{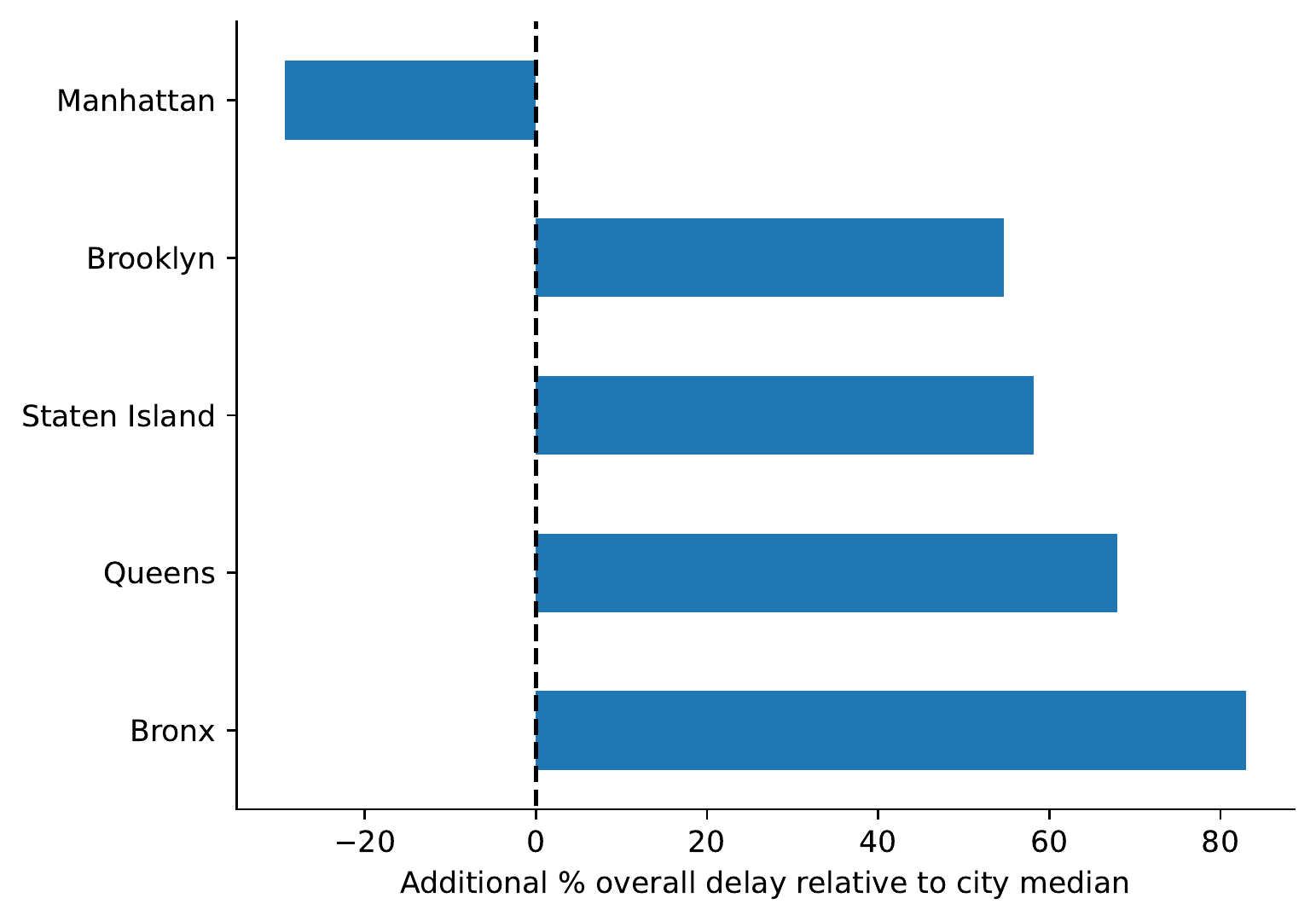}
   }
   %    \subfloat[][Risk prioritization \textit{D}]{
 		% \includegraphics[width=.3\textwidth]{plots/impact_analysis/publicdatarisk_risk_notimputed_D_Borough.pdf}
   % }
 	\caption{Same as \Cref{fig:delaysall_equity} except for other risk prioritization levels. Risk prioritization group D is omitted due to a low percentage of incidents that eventually got worked on.}
 \end{figure}

 \FloatBarrier
\FloatBarrier
\subsection{Preprocessing for Chicago dataset}\label{app:chicagopreprocessing}

In this section, we detail the preprocessing done on the Chicago dataset before model training.

For each request, we also have location and incident-level covariates. We further observe whether and when this service request is marked as `completed', `open' or `canceled', and whether this is a duplicate request; if the request is indeed a duplicate, CDOT and CDWM mark which other requests refer to the same incident. At most three timestamps are associated with each service request: for `open' service requests, we observe the created time and last modified time; for `completed' and `canceled' service requests, we further observe the closed time. One crucial difference in the Chicago dataset, compared with New York City, is that duplicate reports on ``open'' incidents are also marked. This requires different treatments in constructing observation intervals, as we detail below. 

\paragraph{Constructing an observation interval $(S_i, E_i]$} 
For Chicago reports, we make the following choice. Let $\tilde{t}_i$ be the first time an incident is reported, $t_i^{\text{CLOSED}}$ be the time an incident is marked closed, and $\bar{t}$ be the time that we retrieved the dataset (and that it was last updated), 7/4/2022 at 21:40 EDT. Then, $E_i$ of each incident $i$ is:
\begin{equation}
	E_i = \min \left\{100\text{ days}+\tilde{t}_i,\  t_i^{\text{CLOSED}}, \ \bar{t}\right\}. 	\label{eq:durationchicago}
\end{equation}

The choice of adding $\bar{t}$ ensures that for ``open'' incidents, the observation interval is not overly long, which could bias our estimates for the reporting rate downward. 

Though the Chicago dataset contains reports that were generated as early as 7/1/2018, it was not until 2/27/2019 that reports started to have a closed time associated with them\footnote{In private correspondence with the team responsible for the maintenance of this dataset, they pointed out that the current version of the 311 system went live on 12/18/2018, and that the loss of data is likely due to both migration of data from an older version of the system and not correctly logging data at the very beginning of the current version. However, they also confirmed that beyond 2/28/2019, the integrity of data should not be a concern.}. As a consequence, we filter out any reports made prior to 3/1/2019 and filter out ``E-scooter" and ``Vehicle Parked in Bike Lane Complaint" incidents, for which no duplicates are marked. After this, we are left with \num{949352} reports, which represent \num{698365} unique incidents. Appendix \Cref{tab:chicagosummary} lists some summary statistics of this dataset at this stage. We further filter out any incidents that have a negative or extremely short ($<0.01$ days) observation interval. These most likely represent human errors in logging the time. This leaves us with a total of \num{575882} unique incidents. 

% \usepackage{graphicx}
\begin{table}[t]
\caption{Summary statistics from the Chicago dataset, after filtering out service requests prior to 2019/03/01 and categories with no duplicates. There are a total of 28 categories and we selected the top 15 in terms of number of reports. We note that the variation in median days to completion is large in this dataset, and some categories have extremely short completion time. The subsequent filtering of incidents with short duration addresses the concern for mislabeling incidents as completed.}
\label{tab:chicagosummary}
\centering
\resizebox{.9\textwidth}{!}{%

\begin{tabular}{r|r|rr|rrr}
\multicolumn{1}{c|}{\textbf{}}             & \multicolumn{1}{c|}{\textbf{\begin{tabular}[c]{@{}c@{}}Service \\ requests\end{tabular}}} & \multicolumn{2}{c|}{\textbf{Completions}}                                                                                                                          & \multicolumn{3}{c}{\textbf{\begin{tabular}[c]{@{}c@{}}Incidents\\ (from all reports)\end{tabular}}}                                                                                                                                             \\
\multicolumn{1}{l|}{}                      & \multicolumn{1}{l|}{}                                                                     & \multicolumn{1}{c}{\begin{tabular}[c]{@{}c@{}}Completed \\ SRs\end{tabular}} & \multicolumn{1}{c|}{\begin{tabular}[c]{@{}c@{}}Percentage\\ completed\end{tabular}} & \multicolumn{1}{c}{\begin{tabular}[c]{@{}c@{}}Unique\\ incidents\end{tabular}} & \multicolumn{1}{c}{\begin{tabular}[c]{@{}c@{}}Avg. reports\\per incident\end{tabular}} & \multicolumn{1}{c}{\begin{tabular}[c]{@{}c@{}}Median Days\\to Completion\end{tabular}} \\ \hline
\multicolumn{1}{l|}{\textbf{Total number}} &

   \num{949352} &     \num{853989} &                 0.9 &           \num{698365} &                  1.36 &                             6.676 \\
\multicolumn{1}{l|}{\textbf{By Owner Department}}   &                                                                                           &                                                                              &                                                                                     &                                                                                &                                                                                &                                                                                     \\
      \textit{CDOT} &    \num{789063} &     \num{720925} &                0.91 &           \num{559777} &                  1.41 &                             6.783 \\
\textit{CDWM}  &    \num{160289} &     \num{133064} &                0.83 &           \num{138969} &                  1.15 &                             5.704 \\
\multicolumn{1}{l|}{\textbf{By Top 15 Categories}}  &                                                                                           &                                                                              &                                                                                     &                                                                                &                                                                                &                                                                                     \\
\textit{Street Light Out Complaint} &    \num{225408} &     \num{209441} &                0.93 &           \num{122302} &                  1.84 &                             9.786 \\
                       \textit{Pothole in Street Complaint} &    \num{200132} &     \num{189879} &                0.95 &           \num{138832} &                  1.44 &                            10.450 \\
             \textit{Sign Repair Request - All Other Signs} &     \num{96019} &      \num{88775} &                0.92 &            \num{91295} &                  1.05 &                             0.002 \\
                      \textit{Traffic Signal Out Complaint} &     \num{75828} &      \num{69443} &                0.92 &            \num{58705} &                  1.29 &                             0.283 \\
                         \textit{Alley Light Out Complaint} &     \num{53010} &      \num{49329} &                0.93 &            \num{31831} &                  1.67 &                            65.671 \\
                 \textit{Sewer Cleaning Inspection Request} &     \num{39120} &      \num{34390} &                0.88 &            \num{36056} &                  1.08 &                            24.965 \\
                           \textit{Alley Pothole Complaint} &     \num{36333} &      \num{33981} &                0.94 &            \num{25528} &                  1.42 &                            28.973 \\
                         \textit{Water On Street Complaint} &     \num{29494} &      \num{25605} &                0.87 &            \num{26048} &                  1.13 &                             5.517 \\
                       \textit{Sidewalk Inspection Request} &     \num{25515} &       9,821 &                0.38 &            \num{21724} &                  1.17 &                           323.228 \\
                       \textit{Open Fire Hydrant Complaint} &     \num{25105} &      \num{21103} &                0.84 &            \num{15548} &                  1.61 &                             0.313 \\
                  \textit{Sewer Cave-In Inspection Request} &     \num{22308} &      \num{18301} &                0.82 &           \num{21300} &                  1.05 &                            58.093 \\
               \textit{Snow – Uncleared Sidewalk Complaint} &     \num{19532} &      \num{18519} &                0.95 &            \num{17338} &                  1.13 &                             6.810 \\
                       \textit{Water in Basement Complaint} &     \num{16952} &      \num{15563} &                0.92 &            \num{16292} &                  1.04 &                             0.821 \\
                   \textit{Sign Repair Request - Stop Sign} &     \num{16850} &      \num{16746} &                0.99 &            \num{16023} &                  1.05 &                             0.128 \\
                \textit{Street Light Pole Damage Complaint} &     \num{16171} &      \num{14393} &                0.89 &            \num{15194} &                  1.06 &                             1.204 \\
\end{tabular}%
}

\end{table}

\paragraph{Covariate selection and processing} Next, we select the covariates that compose type $\theta$. Similar to the NYC dataset, most service requests in Chicago come with latitude-longitude coordinates, using which we identified which of the over 2,000 census block groups in Chicago this incident is in, through the FCC API. The reason we used census block groups in Chicago instead of census tracts, which were used in NYC, is that the Chicago dataset contains far more incidents to allow finer-grained analyses; further, there are approximately as many census block groups in Chicago as there are census tracts in NYC. We then join this information with the 2020 Census Data from IPUMS NHGIS. Finally, we log transform several variables, filter out the incidents for which any of the covariates are missing, and standardize all data. During this process, we filter out \num{10452} incidents (1.80\% of the total number of incidents), due to either missing covariates or unable to match them to census tracts, and are left with \num{565430} unique incidents for our further analysis, which represent \num{794132} unique service requests. \Cref{fig:numreportshistchicago} shows the histogram of the number of reports per incident during the observation interval; Appendix \Cref{fig:durationhistchicago} shows the distribution of durations. Appendix \Cref{tab:chicagocovariates} lists the covariates we use. 

\begin{figure}[tb]
	\centering
	%	\begin{subfigure}{.5\textwidth}
	% 		\centering
	\subfloat[][Number of reports per incident.]{
		\includegraphics[width=.45\textwidth]{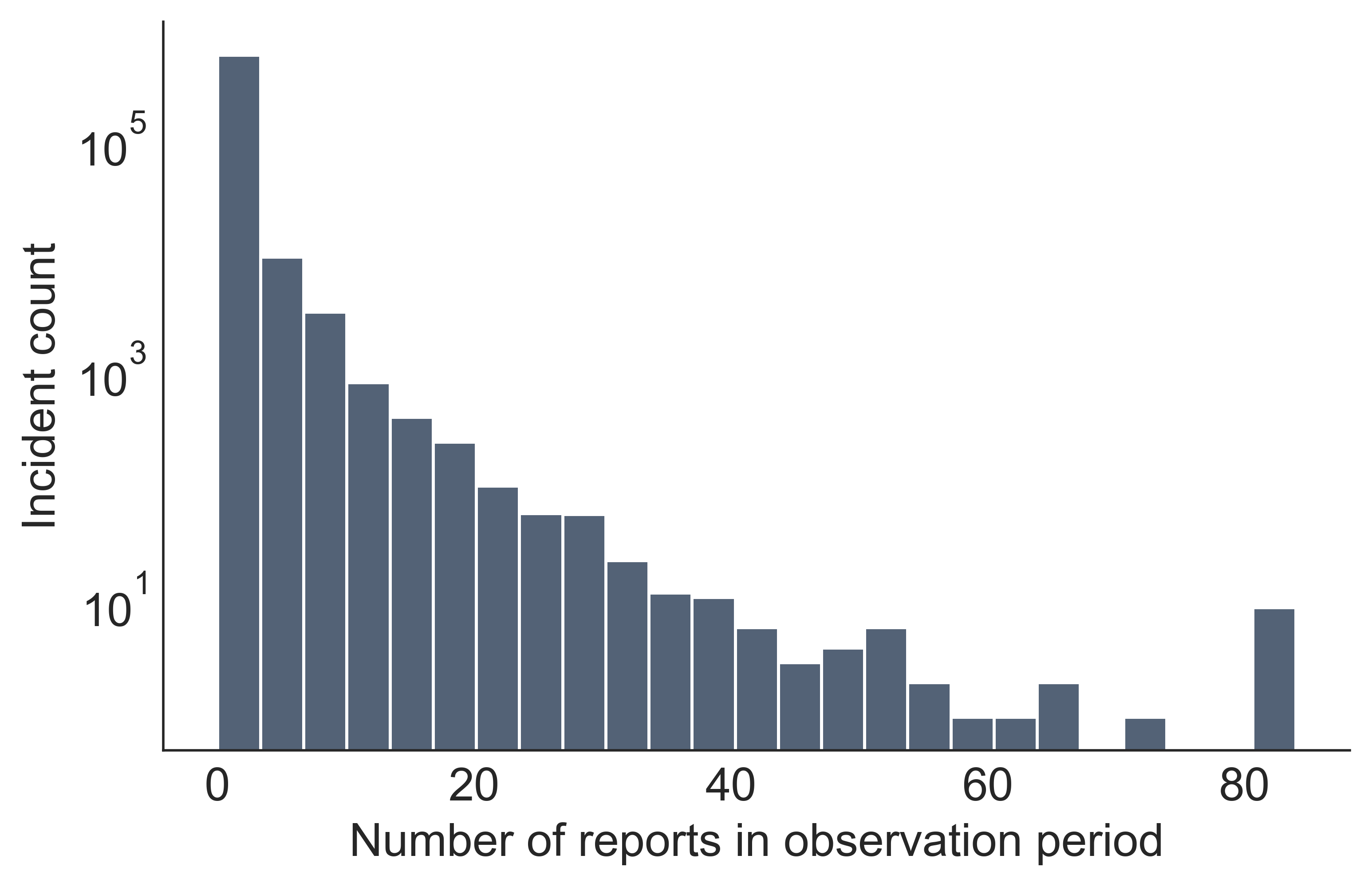}
		% 		\caption{Census tract fixed effects}
		\label{fig:numreportshistchicago}
		% 		\end{subfigure}
	}
	\hfill
	\subfloat[][Length of observation period.]{
		\includegraphics[width=.45\textwidth]{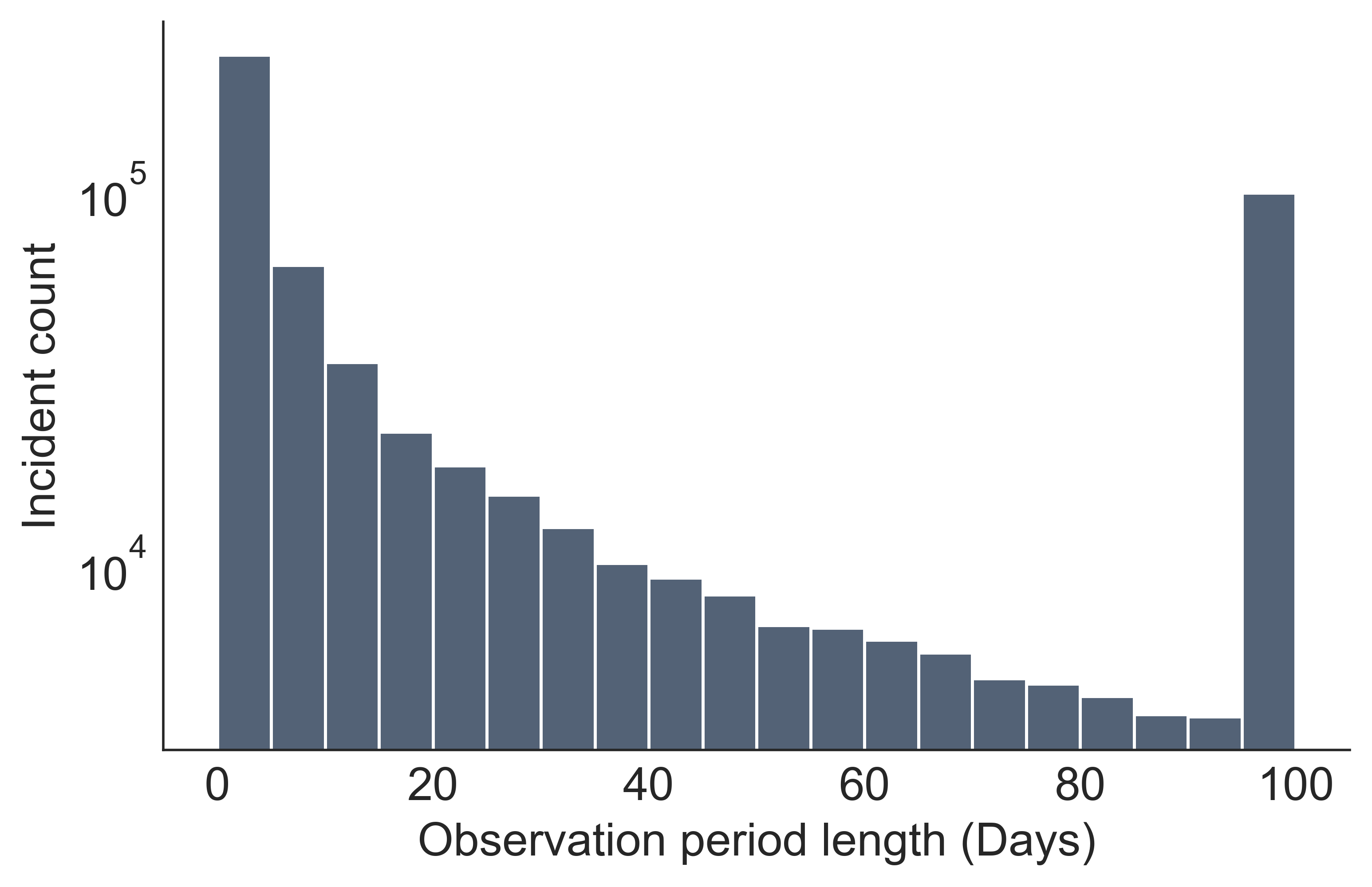}
		% 		\caption{Census tract fixed effects}
		\label{fig:durationhistchicago}
		% 		\end{subfigure}
	}
	\caption{Distribution of number of reports and length of observation for each unique incident in the Chicago aggregated dataset. For most incidents, there are no reports after the first report (at least not in the observation period). There is one incident not included in the histogram of number of reports that attracted 162 reports, which happened to be a malfunctioning traffic light at a busy junction. There is a peak at 100 days for the observation period, due to our configuration in \Cref{eq:durationchicago}.}
\end{figure}

% \usepackage{graphicx}
\begin{table}[tbh]
    \centering
    \caption{Description of covariates in the Chicago aggregated dataset}
    \label{tab:chicagocovariates}
    \resizebox{\textwidth}{!}{%
    \begin{tabular}{l|l}
    \multicolumn{1}{c|}{\textbf{Covariate}} & \multicolumn{1}{c}{\textbf{Description}}                          \\ \hline
    \textbf{Incident Global ID}             & An identifier unique to each incident.                            \\ \hline
    \textbf{Duration}              & The observation duration as defined in \Cref{eq:durationchicago}   
    \\ \hline
    \textbf{Owner Department}           & Which department (CDOT or DWM) the service request is directed to.  
    \\ \hline
    \textbf{Service Request Type}           & The incident type as reported.                                    \\ \hline
    \textbf{Created Month}           & The month that the first report of each incident came in.                                    \\ \hline
    \textbf{Census Block Group}                   & Which census block group the incident occurred in as reported.          \\ \hline
\textbf{Median Age}                     & Median age in the census block group.                                                     \\ \hline
\textbf{Fraction Hispanic}              & Fraction of residents that identify as Hispanic in the census block group.                \\ \hline
\textbf{Fraction white}                 & Fraction of residents that identify as white in the census block group.                   \\ \hline
\textbf{Fraction Black}                 & Fraction of residents that identify as Black in the census block group.                   \\ \hline
\textbf{Fraction no high school degree}              & Fraction of residents that have not graduated from high school in the census block group. \\ \hline
\textbf{Fraction college degree}          & Fraction of residents that have graduated from college in the census block group.         \\ \hline
\textbf{Fraction poverty}            & Fraction of residents that are identified to be in poverty in the census block group.     \\ \hline
\textbf{Fraction renter}                & Fraction of residents that rent their current residence in the census block group.        \\ \hline
\textbf{Fraction family}           & Fraction of family household in the census block group.         \\ \hline
\textbf{Median household value}     & Median value of household in the census block group. \\ \hline
\textbf{Income per capita}                     & Income per capita of residents in the census block group.                                    \\ \hline
\textbf{Density}                        & Population density in the census block group.                        
    \end{tabular}%
    }
    \end{table}

\paragraph{Sampling the dataset for tractability} Due to computational constraints, when training the Stan models on Chicago data, we randomly sample \num{100000} incidents for each model run.

\subsection{Results from Chicago dataset}\label{app:chicagoresults}

In this section, we provide results from applying our methods to the Chicago dataset. Appendix \Cref{fig:chicagoppc} confirms the posterior samples from the zero-inflated model with Base variables reasonably match the observed distribution; Appendix \Cref{tab:chicagobasic} lists the coefficients for the set of Base variables; Appendix \Cref{fig:spatialmapchicago} illustrates the coefficients on census tract Spatial covariates; Appendix \Cref{tab:censuscoefficientschicago} and Appendix \Cref{tab:chicagomultidemo} show the coefficients on census tract Socioeconomic covariates.

\begin{figure}[tb]
	\centering
 % \begin{subfigure}[b]{\textwidth}
 %    \centering
    \includegraphics[width=.38\textwidth]{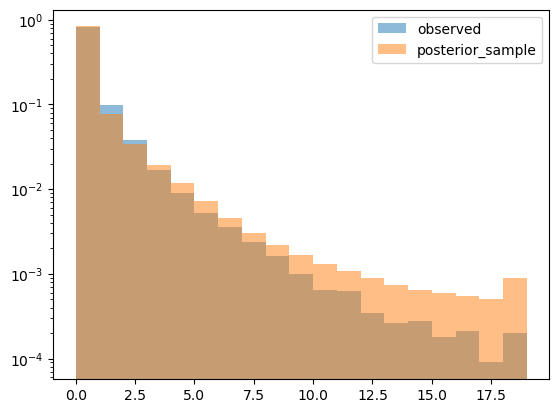}
    \includegraphics[width=.38\textwidth]{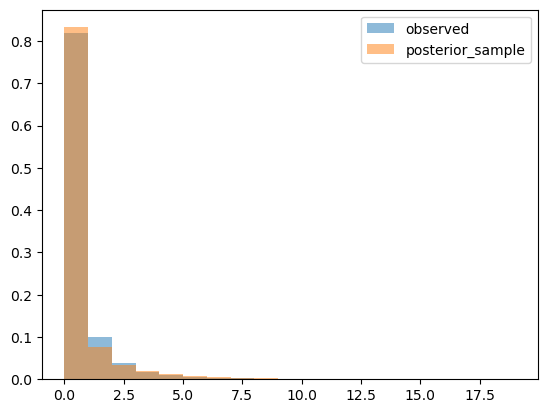}
    \hfill
 % \end{subfigure}
	\caption{Comparison between posterior distributions sampled from the Zero-inflated Poisson regression model and the observed distribution in the data. The left-hand side plot is in log scale for the $y$ axis, while the right-hand side plot is in the natural scale. The posterior samples reasonably match the observed distribution. Note that due to sampling of the full dataset, the observed distribution is not entirely the same as in Appendix \Cref{fig:numreportshistchicago}.}\label{fig:chicagoppc}
\end{figure}

\begin{table}
\centering
\caption{Regression coefficients for Base variables in Chicago, for Max Duration 100 days}\label{tab:chicagobasic}
\begin{tabular}{lrrr}
\toprule
{} &   Mean &  StdDev &  R\_hat \\
\midrule
Intercept                                          & -4.053 &   0.061 &    1.0 \\
Zero Inflation fraction                            &  0.539 &   0.003 &    1.0 \\
Category[Sewer Cave-In Inspection Request]         & -1.645 &   0.094 &    1.0 \\
Category[Water On Street Complaint]                &  0.523 &   0.078 &    1.0 \\
Category[Water in Basement Complaint]              &  0.627 &   0.115 &    1.0 \\
Category[Alley Sewer Inspection Request]           & -0.908 &   0.127 &    1.0 \\
Category[Sewer Cleaning Inspection Request]        & -1.018 &   0.074 &    1.0 \\
Category[Pavement Cave-In Inspection Request]      & -2.996 &   1.122 &    1.0 \\
Category[Protected Bike Lane - Debris Removal]     & -0.308 &   0.285 &    1.0 \\
Category[Sidewalk Inspection Request]              & -2.246 &   0.106 &    1.0 \\
Category[Sign Repair Request - Stop Sign]          &  2.764 &   0.111 &    1.0 \\
Category[Sign Repair Request - One Way Sign]       &  2.147 &   0.212 &    1.0 \\
Category[Sign Repair Request - Do Not Enter Sign]  &  1.080 &   0.474 &    1.0 \\
Category[Sign Repair Request - All Other Signs]    & -1.413 &   0.077 &    1.0 \\
Category[Bicycle Request/Complaint]                & -0.878 &   0.184 &    1.0 \\
Category[Alley Pothole Complaint]                  &  0.265 &   0.067 &    1.0 \\
Category[Pothole in Street Complaint]              &  0.831 &   0.062 &    1.0 \\
Category[Alley Light Out Complaint]                &  0.383 &   0.064 &    1.0 \\
Category[Traffic Signal Out Complaint]             &  2.569 &   0.064 &    1.0 \\
Category[Viaduct Light Out Complaint]              & -0.491 &   0.163 &    1.0 \\
Category[Street Light Out Complaint]               &  1.767 &   0.061 &    1.0 \\
Category[Street Light Pole Damage Complaint]       &  0.680 &   0.117 &    1.0 \\
Category[Street Light On During Day Complaint]     & -0.267 &   0.106 &    1.0 \\
Category[Street Light Pole Door Missing Complaint] & -2.499 &   0.280 &    1.0 \\
Category[Snow Removal - Protected Bike Lane or ... & -2.457 &   0.813 &    1.0 \\
Category[Snow – Uncleared Sidewalk Complaint]      &  0.286 &   0.086 &    1.0 \\
Category[No Water Complaint]                       &  2.433 &   0.124 &    1.0 \\
Category[Low Water Pressure Complaint]             & -0.961 &   0.168 &    1.0 \\
Category[Open Fire Hydrant Complaint]              &  4.221 &   0.065 &    1.0 \\
Category[Water Quality Concern]                    & -2.491 &   0.397 &    1.0 \\
\bottomrule
\end{tabular}
\end{table}

\begin{table}
\centering
\caption{Regression coefficients for Base variables in Chicago, for Max Duration 30 days}\label{tab:chicagobasic30}
\begin{tabular}{lrrr}
\toprule
{} &   Mean &  StdDev &  R\_hat \\
\midrule
Intercept                                          & -3.308 &   0.048 &    1.0 \\
Zero Inflation fraction                            &  0.620 &   0.003 &    1.0 \\
Category[Sewer Cave-In Inspection Request]         & -1.356 &   0.089 &    1.0 \\
Category[Water On Street Complaint]                &  0.592 &   0.069 &    1.0 \\
Category[Water in Basement Complaint]              &  0.782 &   0.109 &    1.0 \\
Category[Alley Sewer Inspection Request]           & -0.373 &   0.116 &    1.0 \\
Category[Sewer Cleaning Inspection Request]        & -0.733 &   0.065 &    1.0 \\
Category[Pavement Cave-In Inspection Request]      & -1.494 &   0.664 &    1.0 \\
Category[Protected Bike Lane - Debris Removal]     & -0.219 &   0.344 &    1.0 \\
Category[Sidewalk Inspection Request]              & -2.156 &   0.120 &    1.0 \\
Category[Sign Repair Request - Stop Sign]          &  2.514 &   0.094 &    1.0 \\
Category[Sign Repair Request - One Way Sign]       &  1.666 &   0.222 &    1.0 \\
Category[Sign Repair Request - Do Not Enter Sign]  &  0.194 &   0.500 &    1.0 \\
Category[Sign Repair Request - All Other Signs]    & -1.189 &   0.068 &    1.0 \\
Category[Bicycle Request/Complaint]                & -0.317 &   0.166 &    1.0 \\
Category[Alley Pothole Complaint]                  &  0.086 &   0.060 &    1.0 \\
Category[Pothole in Street Complaint]              &  0.542 &   0.050 &    1.0 \\
Category[Alley Light Out Complaint]                &  0.070 &   0.054 &    1.0 \\
Category[Traffic Signal Out Complaint]             &  2.733 &   0.051 &    1.0 \\
Category[Viaduct Light Out Complaint]              & -0.639 &   0.183 &    1.0 \\
Category[Street Light Out Complaint]               &  1.490 &   0.049 &    1.0 \\
Category[Street Light Pole Damage Complaint]       &  0.190 &   0.118 &    1.0 \\
Category[Street Light On During Day Complaint]     & -0.968 &   0.124 &    1.0 \\
Category[Street Light Pole Door Missing Complaint] & -1.309 &   0.222 &    1.0 \\
Category[Snow Removal - Protected Bike Lane or ... & -1.584 &   0.562 &    1.0 \\
Category[Snow – Uncleared Sidewalk Complaint]      &  0.037 &   0.076 &    1.0 \\
Category[No Water Complaint]                       &  1.849 &   0.110 &    1.0 \\
Category[Low Water Pressure Complaint]             & -1.258 &   0.184 &    1.0 \\
Category[Open Fire Hydrant Complaint]              &  3.686 &   0.055 &    1.0 \\
Category[Water Quality Concern]                    & -2.835 &   0.597 &    1.0 \\
\bottomrule
\end{tabular}
\end{table}

\begin{table}
\centering
\caption{Regression coefficients for Base variables in Chicago, for Max Duration 200 days}\label{tab:chicagobasic200}
\begin{tabular}{lrrr}
\toprule
{} &   Mean &  StdDev &  R\_hat \\
\midrule
Intercept                                          & -4.448 &   0.082 &    1.0 \\
Zero Inflation fraction                            &  0.496 &   0.003 &    1.0 \\
Category[Sewer Cave-In Inspection Request]         & -1.537 &   0.107 &    1.0 \\
Category[Water On Street Complaint]                &  0.448 &   0.094 &    1.0 \\
Category[Water in Basement Complaint]              &  1.050 &   0.113 &    1.0 \\
Category[Alley Sewer Inspection Request]           & -0.823 &   0.140 &    1.0 \\
Category[Sewer Cleaning Inspection Request]        & -1.081 &   0.092 &    1.0 \\
Category[Pavement Cave-In Inspection Request]      & -2.813 &   1.231 &    1.0 \\
Category[Protected Bike Lane - Debris Removal]     &  0.070 &   0.340 &    1.0 \\
Category[Sidewalk Inspection Request]              & -2.287 &   0.111 &    1.0 \\
Category[Sign Repair Request - Stop Sign]          &  2.458 &   0.113 &    1.0 \\
Category[Sign Repair Request - One Way Sign]       &  2.830 &   0.202 &    1.0 \\
Category[Sign Repair Request - Do Not Enter Sign]  & -1.655 &   0.933 &    1.0 \\
Category[Sign Repair Request - All Other Signs]    & -1.418 &   0.094 &    1.0 \\
Category[Bicycle Request/Complaint]                & -1.372 &   0.189 &    1.0 \\
Category[Alley Pothole Complaint]                  &  0.378 &   0.085 &    1.0 \\
Category[Pothole in Street Complaint]              &  0.976 &   0.082 &    1.0 \\
Category[Alley Light Out Complaint]                &  0.515 &   0.084 &    1.0 \\
Category[Traffic Signal Out Complaint]             &  2.661 &   0.083 &    1.0 \\
Category[Viaduct Light Out Complaint]              & -0.298 &   0.152 &    1.0 \\
Category[Street Light Out Complaint]               &  2.026 &   0.082 &    1.0 \\
Category[Street Light Pole Damage Complaint]       &  0.987 &   0.124 &    1.0 \\
Category[Street Light On During Day Complaint]     & -0.668 &   0.118 &    1.0 \\
Category[Street Light Pole Door Missing Complaint] & -1.728 &   0.188 &    1.0 \\
Category[Snow Removal - Protected Bike Lane or ... & -2.747 &   1.161 &    1.0 \\
Category[Snow – Uncleared Sidewalk Complaint]      &  0.877 &   0.099 &    1.0 \\
Category[No Water Complaint]                       &  2.342 &   0.131 &    1.0 \\
Category[Low Water Pressure Complaint]             & -0.877 &   0.186 &    1.0 \\
Category[Open Fire Hydrant Complaint]              &  4.165 &   0.085 &    1.0 \\
Category[Water Quality Concern]                    & -2.479 &   0.426 &    1.0 \\
\bottomrule
\end{tabular}
\end{table}

% \begin{figure}[tb]
%     \centering
%     \includegraphics[width=.6\textwidth]{plots/chicago_smooth.eps}
%     \caption{Coefficients on spatial covariates. The census tract spatial coefficients are estimated using the ICAR spatial zero-inflated Poisson regression. As seen in the plot, there is also substantial difference in reporting rates across census tracts in Chicago, with difference as large as more than 3 times ($e^{1.2}$).}
%     \label{fig:spatialmapchicago}
% \end{figure}

\begin{table}
\centering
\caption{Census Block Group Socio-economic coefficients in Chicago, estimated alone in a regression alongside
the incident-specific covariates.}
\label{tab:censuscoefficientschicago}
\begin{tabular}{lrrrrr}
\toprule
                              &   Mean &  StdDev &   2.5\% &  97.5\% \\
\midrule
                    Median age & -0.017 &   0.005 & -0.028 & -0.007 \\
             Fraction Hispanic &  0.063 &   0.005 &  0.053 &  0.073 \\
                Fraction white &  0.000 &   0.006 & -0.011 &  0.011 \\
                Fraction Black & -0.022 &   0.005 & -0.033 & -0.012 \\
Fraction no high school degree &  0.051 &   0.005 &  0.040 &  0.060 \\
       Fraction college degree & -0.051 &   0.006 & -0.063 & -0.041 \\
              Fraction poverty &  0.011 &   0.005 & -0.001 &  0.020 \\
               Fraction renter & -0.008 &   0.006 & -0.020 &  0.003 \\
               Fraction family &  0.029 &   0.006 &  0.017 &  0.039 \\
       Log(Median house value) & -0.056 &   0.007 & -0.069 & -0.044 \\
        Log(Income per capita) & -0.051 &   0.005 & -0.063 & -0.041 \\
                  Log(Density) &  0.070 &   0.005 &  0.059 &  0.081 \\
\bottomrule
\end{tabular}
\end{table}

\begin{table}
\centering
\caption{Census Block Group Socio-economic coefficients in Chicago, estimated together in a regression alongside the incident-level covariates.}
\label{tab:chicagomultidemonodensity}
\begin{tabular}{lrrrrr}
\toprule
                         &   Mean &  StdDev &   2.5\% &  97.5\% \\
\midrule
             Median age &  -0.009 &   0.003 & -0.015 &  -0.004 \\
         Fraction white &  0.017 &   0.003 &  0.010 &  0.023 \\
Fraction college degree & -0.054 &   0.006 & -0.066 & -0.042 \\
        Fraction renter & -0.006 &   0.003 & -0.011 & -0.001 \\
 Log(Income per capita) & -0.003 &   0.006 & -0.015 &  0.009 \\
\bottomrule
\end{tabular}
\end{table}

\begin{table}
\centering
\caption{Census Block Group Socio-economic coefficients in Chicago, estimated together in a regression alongside the incident-level covariates. Compared to \Cref{tab:chicagomultidemonodensity}, we further control for log population density, which only affects the direction of association of median age.}
\label{tab:chicagomultidemo}
\begin{tabular}{lrrrrr}
\toprule
                         &   Mean &  StdDev &   2.5\% &  97.5\% \\
\midrule
             Median age &  0.008 &   0.006 & -0.004 &  0.019 \\
         Fraction white &  0.021 &   0.007 &  0.007 &  0.035 \\
Fraction college degree & -0.037 &   0.011 & -0.059 & -0.015 \\
        Fraction renter & -0.037 &   0.006 & -0.051 & -0.025 \\
 Log(Income per capita) & -0.056 &   0.012 & -0.081 & -0.035 \\
           Log(Density) &  0.091 &   0.006 &  0.078 &  0.101 \\
\bottomrule
\end{tabular}
\end{table}

% \FloatBarrier
% \input{appendix/generalize_temporal.tex}
\FloatBarrier
\subsection{Stan source code for the \textbf{Zero-inflated} \textbf{Spatial} Poisson regression model}\label{app:stancode}

\begin{lstlisting}[language = Stan]

functions {// Using reduce sum for within-chain parallel processing.
// function for calculating the log-likelihood
real partial_sum_lpmf(int[] y_slice,
                        int start, int end,
                        matrix X_total,
                        vector logduration,
                        vector beta_total, real theta_zeroinflation) {
    int Nloc = end - start;
    real localtarget = 0;
    // two cases for calculation reflecting zero-inflation
    for (n in 1:Nloc+1) {
      int ind = start + n - 1; 
        if (y_slice[n] == 0) {
          localtarget += 
          log_sum_exp(bernoulli_lpmf(1 | theta_zeroinflation),
              bernoulli_lpmf(0 | theta_zeroinflation)
            + poisson_log_glm_lpmf(y_slice[n] | X_total[ind:ind, :], 
                  logduration[ind:ind], beta_total)
                  );
        } else {
          localtarget += 
          bernoulli_lpmf(0 | theta_zeroinflation)
              + poisson_log_glm_lpmf(y_slice[n] | X_total[ind:ind, :], 
                ogduration[ind:ind], beta_total
                    ); 
        }
    }
    return localtarget;
  }
}

data {// Define variables in data
  int<lower=0> N_incidents;// number of observations
  int<lower=0> covariate_matrix_width; // covariate matrix width
  // design matrix for other covariates
  matrix[N_incidents, covariate_matrix_width] X;
  vector<lower=1,upper=1>[N_incidents] ones; // vector of ones
  int<lower=0> N_category; // number of categories
  int<lower=0> N_tract; // number of tracts 
  int<lower=0> N_edges; // tract adjacency matrix number of edges 
  // node1[i] adjacent to node2[i]
  int<lower=1, upper=N_tract> node1[N_edges];  
  // and node1[i] < node2[i]
  int<lower=1, upper=N_tract> node2[N_edges];  
  // design matrix for category
  matrix[N_incidents, N_category] X_category;  
  matrix[N_incidents, N_tract] X_tract;  // design matrix for tract
  vector<lower=0>[N_incidents] duration;// alive time for incident
  // count outcome - duplicates for the incident
  int<lower=0> y[N_incidents]; 
}

transformed data {// Transform for succinctness and better performance
  vector[N_incidents] logduration; // log of duration
  logduration = log(duration);
  // total design matrix
  matrix[N_incidents, 1 + N_tract +
            covariate_matrix_width + N_category] X_total; 
  X_total = append_col(append_col(append_col(ones, 
                X_tract), X_category), X);
}

parameters {// Define parameters to estimate
  vector[N_tract - 1] beta_tract_raw; // coefficents for tract
  vector[N_category-1] beta_category_raw; // coefficients for category
  // coefficents for other covariates in the model
  vector[covariate_matrix_width] beta; 
  real intercept; // intercept term in the model
  // zero inflation parameter
  real<lower=0, upper=1> theta_zeroinflation; 
}

transformed parameters  {// Transform parameters for better performance
  vector[N_tract] beta_tract; // coefficents for tract
  vector[N_category] beta_category; // coefficients for category
  // full coefficient vector
  vector[1 + N_tract + covariate_matrix_width + N_category] beta_total; 
  // zero centering coefficients
  beta_category = append_row(beta_category_raw, 
                      -sum(beta_category_raw)); 
  // zero centering coefficients
  beta_tract = append_row(beta_tract_raw, 
                  -sum(beta_tract_raw)); 
  // combine coefficient vectors to get full coefficient vector
  beta_total = append_row(append_row(append_row(intercept, 
    beta_tract), beta_category), beta); 
}

model {// Prior part of Bayesian inference 
  beta_tract ~ normal(0, 1); 
  beta_category ~ normal(0, 2*inv(sqrt(1 - inv(N_category)))); 
  beta ~ normal(0, 1); 
  intercept ~ normal(0, 5);
  // Use reduce sum to calculate the target likelihood
  int grainsize = 1;
  target += reduce_sum(
    partial_sum_lpmf, y, grainsize, X_total, 
    logduration, beta_total, theta_zeroinflation
    );
  // add adjacency priors to beta_tract
  target += -5 * dot_self(beta_tract[node1] - beta_tract[node2]);
}

\end{lstlisting}

\FloatBarrier

% \section{Validation}
% \input{appendix_validation}

\end{document}